%% file: thesis.tex
\documentclass[MS]{iitmdiss}
\usepackage{times}

\usepackage{graphicx}
\usepackage{epstopdf}
\usepackage[hidelinks]{hyperref} % hyperlinks for references.
\usepackage{amsthm,amsmath} % easier math formulae, align, subequations \ldots
\usepackage{caption}
\usepackage{xparse}
\usepackage{pgffor}
\usepackage{xifthen,etoolbox}
\usepackage{subcaption}
\usepackage[ruled,vlined,linesnumbered]{algorithm2e}
\usepackage{subcaption}
\usepackage[utf8]{inputenc} % allow utf-8 input
\usepackage{url}            % simple URL typesetting
\usepackage{booktabs}       % professional-quality tables
\usepackage{amsfonts,dsfont,bbold}       % blackboard math symbols
\usepackage{microtype}      % microtypography
\usepackage{enumerate}
\usepackage{mathtools}
\usepackage{color,latexsym,amssymb}
\usepackage{array,multirow}

\def\ShowComment{false}
\ifdefined\ShowComment

\newcommand{\todo}[1]{{\color{red}\bf [TODO: #1]}}
\else
\def\john#1{}
\newcommand{\todo}[1]{\hspace{-1.5mm}} %Replace with proper letter width if available.
\fi

\newcommand{\john}[1]{{}}

\newcommand{\bigboxplus}{
  \mathop{
    \vphantom{\bigoplus} 
    \mathchoice
      {\vcenter{\hbox{\resizebox{\widthof{$\displaystyle\bigoplus$}}{!}{$\boxplus$}}}}
      {\vcenter{\hbox{\resizebox{\widthof{$\bigoplus$}}{!}{$\boxplus$}}}}
      {\vcenter{\hbox{\resizebox{\widthof{$\scriptstyle\oplus$}}{!}{$\boxplus$}}}}
      {\vcenter{\hbox{\resizebox{\widthof{$\scriptscriptstyle\oplus$}}{!}{$\boxplus$}}}}
  }\displaylimits 
}

\DeclareMathOperator{\firstBranchNode}{firstBranchNode}
\DeclareMathOperator{\LCA}{LCA}
\DeclareMathOperator{\child}{children}
\DeclareMathOperator{\trunc}{trunc}

\newtheorem{thm}{Theorem} % reset theorem numbering for each chapter
\numberwithin{thm}{chapter}
\newtheorem{defn}[thm]{Definition} % definition numbers are dependent on theorem numbers
\newtheorem{obs}[thm]{Observation}
\newtheorem{lemma}[thm]{Lemma}
\newtheorem{property}[thm]{Property}

\newtheorem{corollary}[thm]{Corollary}
\newtheorem{cor}[thm]{Corollary}
\newtheorem{conj}{Conjecture}

\newcommand{\curly}[1]{\left\{ #1 \right\}}
\newcommand{\paren}[1]{\left( #1 \right)}
\newcommand{\angularbracs}[1]{\langle #1 \rangle}

\newcommand{\tree}{\mathcal{T}}
\newcommand{\angularbrac}[1]{\angularbracs{#1}}
\newcommand{\CONGEST}{{\cal{CONGEST }}}

\DeclareMathOperator{\msg}{msg}

\newcommand{\subgraphsequence}[2]{G_{\paren{{#1},{#2}}}}

\newcommand{\edgesatvertex}[1]{{E}_{{#1}^\downarrow}}
\newcommand{\graphatvertex}[1]{G_{#1}}
\newcommand{\Zidentity}{\mathbb{e}}
\newcommand{\Zabsorb}{\nparallel}

\newcommand{\mytitle}{Improved Algorithm for Min-Cuts in Distributed Networks}

\newcommand{\plog}{\text{polylog }}

\newcommand\desc[2][]{%
  \ifstrempty{#1}{%
	{#2}^{\downarrow}
  }{%
	{#2}^{\downarrow {#1}}
  }%
}

\newcommand\ancestor[2][]{%
  \ifstrempty{#1}{%
	\mathcal{A}\paren{#2}
  }{%
	\mathcal{A}_{#1}\paren{#2}
  }%
}
\newcommand\level[2][]{%
  \ifstrempty{#1}{%
	\ell\paren{#2}
  }{%
	\ell_{#1}\paren{#2}
  }%
}
\newcommand\parent[2][]{%
  \ifstrempty{#1}{%
	\pi\paren{#2}
  }{%
	\pi_{{#1}}\paren{#2}
  }%
}

\newcommand\Sketch[3][]{%
  \ifstrempty{#1}{%
	\mathcal{S}^{#2}(#3)
  }{%
	\mathcal{S}_{#1}^{#2}(#3)
  }%
}

\newcommand\rSketch[4][]{%
  \ifstrempty{#1}{%
	{\mathcal{S}}^{#2}(#3,#4)
  }{%
	{\mathcal{S}}_{#1}^{#2}(#3,#4)
  }%
}

\newcommand\bnum[2][]{%
	\mathcal{\xi}_{{#1}}\paren{#2}
}

\newcommand{\zetaoperator}{\texttt{zeta-operator}}

\newcommand{\CASE}[1]{\textbf{CASE}-#1}

\begin{document}

%%%%%%%%%%%%%%%%%%%%%%%%%%%%%%%%%%%%%%%%%%%%%%%%%%%%%%%%%%%%%%%%%%%%%%
% Title page
\titlename{\mytitle}
\title{\mytitle}

\author{Mohit Daga}

\date{OCTOBER 2017}
\department{COMPUTER SCIENCE AND ENGINEERING}

%\nocite{*}
\maketitle

%%%%%%%%%%%%%%%%%%%%%%%%%%%%%%%%%%%%%%%%%%%%%%%%%%%%%%%%%%%%%%%%%%%%%%
% Certificate
\certificate

\vspace*{0.5in}

\noindent This is to certify that the thesis titled {\bf \mytitle}, 
submitted by {\bf Mohit Daga},  
  to the Indian Institute of Technology, Madras, for
the award of the degree of {\bf Master of Science}, is a bona fide
record of the research work done by him under my supervision.  The
contents of this thesis, in full or in parts, have not been submitted
to any other Institute or University for the award of any degree or
diploma.

\vspace*{1.5in}

\begin{singlespacing}
\hspace*{-0.25in}
\parbox{2.5in}{
\noindent {\bf Dr. John Augustine} \\
\noindent Research Guide \\ 
\noindent Associate Professor \\
\noindent Dept. of CSE \\
\noindent IIT-Madras, 600 036 \\
} 
\hspace*{1.0in} 
%\parbox{2.5in}{
%\noindent {\bf Prof.~S.~C.~Rajan} \\
%\noindent Research Guide \\ 
%\noindent Assistant Professor \\
%\noindent Dept.  of  Aerospace Engineering\\
%\noindent IIT-Madras, 600 036 \\
%}  
\end{singlespacing}
\vspace*{0.25in}
\noindent Place: Chennai\\
Date:

%%%%%%%%%%%%%%%%%%%%%%%%%%%%%%%%%%%%%%%%%%%%%%%%%%%%%%%%%%%%%%%%%%%%%%
% Acknowledgements
\acknowledgements

\input{thesis_subparts/ack.tex}

\abstract{

\noindent KEYWORDS: \hspace*{0.5em} \parbox[t]{4.4in}{Connectivity, Reliability, Edge Min-Cuts}

\vspace*{24pt}

\noindent 
\input{thesis_subparts/abstract.tex}
}
\pagebreak

%%%%%%%%%%%%%%%%%%%%%%%%%%%%%%%%%%%%%%%%%%%%%%%%%%%%%%%%%%%%%%%%%
% Table of contents etc.

\begin{singlespace}
\tableofcontents
\thispagestyle{empty}

\listoftables
\addcontentsline{toc}{chapter}{LIST OF TABLES}
\listoffigures
\addcontentsline{toc}{chapter}{LIST OF FIGURES}
\end{singlespace}

%%%%%%%%%%%%%%%%%%%%%%%%%%%%%%%%%%%%%%%%%%%%%%%%%%%%%%%%%%%%%%%%%%%%%%
% Abbreviations
\abbreviations

\noindent 
\begin{tabbing}
xxxxxxxxxxx \= xxxxxxxxxxxxxxxxxxxxxxxxxxxxxxxxxxxxxxxxxxxxxxxx \kill
\textbf{IITM}   \> Indian Institute of Technology, Madras \\
\textbf{TRSF} \> Tree Restricted Semigroup function\\
\end{tabbing}

\pagebreak

%%%%%%%%%%%%%%%%%%%%%%%%%%%%%%%%%%%%%%%%%%%%%%%%%%%%%%%%%%%%%%%%%%%%%%
% Notation

\chapter*{\centerline{NOTATION}}
\addcontentsline{toc}{chapter}{NOTATION}

\begin{singlespace}
\input{thesis_subparts/notations.tex}
\end{singlespace}

\pagebreak
\clearpage

% The main text will follow from this point so set the page numbering
% to arabic from here on.
\pagenumbering{arabic}

%%%%%%%%%%%%%%%%%%%%%%%%%%%%%%%%%%%%%%%%%%%%%%%%%%
% Introduction.

\chapter{Introduction}
\label{chapter:1_introduction}
\input{thesis_subparts/chapter1/introduction.tex}
\chapter{Preliminaries and Background}
\label{chpater:2_prelim_background}
\input{thesis_subparts/chapter2/chapter2_main.tex}

\chapter{Technical Overview}
\label{chpater:3_technical_overview}
\input{thesis_subparts/chapter3/chapter_3.tex}
\chapter{Tree Restricted Semigroup Function}
\label{chapter:4_trsf}
	\input{thesis_subparts/chapter4/restricted_semigroup.tex}
	\input{thesis_subparts/chapter4/tree-cuts.tex}
	\input{thesis_subparts/chapter4/connectivity.tex}
\chapter{Min-Cuts of size three}
\label{chapter:5_min-cut-3}
	\input{thesis_subparts/chapter5/three_cut.tex}
	\input{thesis_subparts/chapter5/sketching.tex}
	\input{thesis_subparts/chapter5/layered_algorithm.tex}	
		\input{thesis_subparts/chapter5/summary.tex}	

\chapter{Future Work}
	\input{thesis_subparts/futureWork.tex}

\begin{singlespace}
  \bibliography{refs}
\end{singlespace}
\newpage
\large{\chapter*{\centering CURRICULUM VITAE}
\input{thesis_subparts/cv.tex}}
\newpage
\large{\chapter*{\centering GENERAL TEST COMMITTEE}
\input{thesis_subparts/gtc.tex}}
\end{document}

%% file: thesis_subparts/ack.tex
%TEX root = ../thesis.tex
First and foremost, I would like to extend my sincere gratitude to my thesis supervisor Dr. John Augustine. I am thankful to John for his constant encouragement, guidance and optimism. John actively helped me in improving my technical writing and presentation with lot of patience. Our countless discussions about research (we literally burnt the midnight oil) enabled the timely completion of this thesis. His opinion and vision about a career in research and academia are well placed.

I also closely collaborated with Dr. Danupon Nanongkai, who introduced me to the problem covered in this thesis. I thank him for inviting me to KTH in April'17. Danupon also reviewed a draft of this work and gave me valuable inputs.

During the initial days at IIT-Madras, I interacted with Dr. Rajseakar Manokaran, in his capacity as my faculty adviser. He inspired me to pursue research in Theoretical Computer Science. I thank Dr. N.S. Narayanaswamy,  chairperson of my General Test Committee (GTC)  for being critical of my seminar presentations and helping me to grow as a researcher. I also thank Dr. Krishna Jagannathan and Dr. Shweta Agrawal for accepting to be part of my GTC.

I feel proud to be part of IIT Madras and CSE Department. I thank the administration here for establishing tranquility and serenity of campus after the devastating 2016 cyclone. I thank the support staff of my department. Special thanks to Mr. Mani from the CSE office for scheduling the required meetings and helping me to understand the procedures and requirements with patience. 

During my stay at IIT Madras, I was part of two groups: ACT Lab and TCS Lab. I thank my lab mates for maintaining good work culture. My various discussions on a wide range of topics with my lab mates were enjoyable and enriching. This experience will be helpful when I move to Stockholm for my Ph.D. studies.

Last but not the least, I thank my mother Beena, my father Manohar and my sisters: Kaushal and Sheetal. 
%They have taught me to appreciate the present, work for future and learn from the past. 
This thesis is dedicated to them. 

%% file: thesis_subparts/abstract.tex
In this thesis, we present fast deterministic algorithm to find small cuts in distributed networks. 
Finding small min-cuts for a network is essential for ensuring the quality of service and reliability. 
Throughout this thesis, we use the \CONGEST model which is a typical message passing model used to design and analyze algorithms in distributed networks. 
We survey various algorithmic techniques in the \CONGEST model and give an overview of the recent results to find cuts. 
We also describe elegant graph theoretic ideas like cut spaces and cycle spaces that provide useful intuition upon which our work is built. 

Our contribution is a novel fast algorithm\john{Do you want to to claim your contribution as a single algorithm or as several algorithms? I am assuming the singular.} to find small cuts. 
Our algorithm relies on a new characterization of trees and cuts introduced in this thesis. Our algorithm is built upon several  new algorithmic ideas that, when coupled with our characterization of trees and cuts, help us to find the required min-cuts. 
Our novel techniques include  a \emph{tree restricted semigroup function (TRSF)}, a novel \emph{sketching} technique, and a \emph{layered algorithm}. 
TRSF is defined with respect to a spanning tree and is based on a commutative semigroup. 
This simple yet powerful technique helps us to deterministically find min-cuts of size one (bridges) and min-cuts of size two optimally. 
Our sketching technique samples a small but relevant vertex set which is enough to find small min-cuts in certain cases. 
Our layered algorithm finds min-cuts in smaller sub-graphs \emph{pivoted} by nodes at different levels in a spanning tree and uses them to make the decision about the min-cuts in the complete graph. 
This is interesting because it enables us to show that even for a global property like finding min-cuts, local information can be exploited in a coordinated manner. 

%% file: thesis_subparts/notations.tex
\begin{tabbing}
xxxxxxxxxxx \= xxxxxxxxxxxxxxxxxxxxxxxxxxxxxxxxxxxxxxxxxxxxxxxx \kill
\textbf{$\desc[T]{v}$}  \> set of descendants of vertex $v$ (including itself) in some tree $T$ \\
\textbf{$\ancestor[T]{v}$} \> set of ancestors of vertex $v$ (including itself) in some tree $T$ \\
\textbf{$\parent[T]{v}$} \> parent of vertex $v$ in some tree $T$\\
\textbf{$\child_T(v)$} \> children of node $v$ in some tree $T$\\
\textbf{$\delta(A)$} \> cut set induced by vertex set $A$\\
\textbf{$\eta_T(v)$} \> $|\delta(\desc[T]{v})|$\\
\textbf{$\gamma(A,B)$} \> $|\delta(A) \cap \delta(B)|$, where $A,B$ are vertex set\\
\textbf{$\level[T]{v}$} \> level of node $v$ in a rooted tree $T$, root at level $0$\\
 \textbf{$\alpha_T(x,l)$} \> ancestor of vertex $x$ at level $l$ \\
\textbf{$R_T(v)$} \> canonical tree of a node $v$ with respect to the spanning $T$\\
\textbf{$\Sketch[T]{k}{v}$} \> $k$-Sketch of node $v$ w.r.t a spanning tree $T$ \\
\textbf{$\bnum[T]{v}$} \> branching number of node $v$ w.r.t some tree $T$\\
\textbf{$\rSketch[T]{k}{v}{a}$} \> reduced $k$-Sketch of a node $v$ and some $a \in \desc[T]{v}\setminus v$ w.r.t a spanning tree $T$\\
\textbf{$\rho_T(v)$} \> a path from the root to node $v$ in the tree $T$\\
\textbf{$\LCA_T(v_1,v_2)$} \> Lowest Common Ancestor of two node $v_1,v_2$ in some tree T\\
\textbf{$\overline{\mathcal{N}}_T({A})$} \> $\curly{y \mid x \in A, y \notin A, (x,y) \in E, (x,y)\ \text{is a non-tree edge in tree }T}$\\
\textbf{$\mathcal{P}_T(A)$}\> $\curly{\rho_T(a) \mid a\in A}$
\end{tabbing}

%% file: thesis_subparts/chapter1/introduction.tex
%TEX root = ../../thesis.tex
Networks of massive sizes are all around us. From large computer networks which are the backbone of the Internet to sensor networks of any manufacturing unit, the ubiquity of these networks is unquestionable. Each node in these kinds of networks is an independent compute node or abstracted as one. Further, nodes in these networks only have knowledge of their immediate neighbors, but not the whole topology. There are many applications which require us to solve problems related to such networks. To solve such problems, a trivial idea is as follows: Each node sends its adjacency list to one particular node through the underlying communication network; that node then solves the required problem using the corresponding sequential algorithm. Such an idea may not be cost effective due to the constraints of the capacity of links and large sizes of these networks.  The question then is: Does there exist a faster cost-effective algorithm through which nodes in a network can coordinate among each other to find the required network property? Such a question is part of a more generalized spectrum of \emph{Distributed Algorithms}. 

In this thesis, we give fast distributed algorithm for finding small min-cuts. An edge cut is the set of edges whose removal disconnects the network and the one with a minimum number of edges is called as a \emph{min-cut}. Finding small min-cuts has applications in ensuring reliability and quality of service. This enables the network administrator to know about the critical edges in the network.

Our approach to finding the min-cut is based on our novel distributed algorithmic techniques and backed by new characterizations of cuts and trees. Our algorithmic techniques are as follows:
\begin{itemize}
	\item \emph{Tree Restricted Semigroup Function} wherein a semigroup function is computed for each node in the network based on a rooted spanning tree
	\item \emph{Sketching Technique}: which allows us to collect small but relevant information of the network in a fast and efficient manner 
	\item \emph{Layered Algorithm} which finds min-cuts in smaller sub-graphs and strategically uses the information to make the decision about min-cuts in the complete graph.
\end{itemize}

\section{Notations and Conventions}
We consider an undirected, unweighted and simple graph $G = (V,E)$ where $V$ is the vertex set and $E$ is the edge set. We use $n$ and $m$ to denote the cardinalities of $V$ and $E$ respectively. Edges are the underlying physical links and vertices are the compute nodes in the network. We will use the term vertex and node interchangeably. Similarly, the terms edges and links will be used interchangeably. 

A cut $C = (A,V\setminus A)$, where $\emptyset \subsetneq A \subsetneq V$ is a partition of the vertex set into two parts. The cut-set is the edges crossing these two parts. A min-cut is a cut-set with a minimum number of edges. Let $d_G(u,v)$ be the distance between any two nodes $u,v$ in the network which is defined by the number of edges in the shortest path between them. We will consider connected networks, so $d_G(u,v)$ is guaranteed to be well-defined and finite. The diameter of the graph $G$ is defined as $\max\limits_{u,v \in V} d_G(u,v)$. We will use $D$ to denote the diameter of the graph and $n$ for the size of the vertex set and $m$ for the size of edge set. 

In a typical communication network, the diameter $D$ is much smaller than the number of vertices $n$. This is because the networks are often designed in such a way  that there is an assurance that a message sent from any  node to any other node should not be required to pass through too many links.

\section{Distributed Model}
Different types of communication networks have different parameters. This is primarily based on how communication takes place between any two nodes in a distributed system.  For the algorithm designers, these details are hidden and abstracted by appropriate modeling of the real world scenario. An established book in this domain titled \lq Distributed Computing: A Locality-Sensitive Approach\rq\ \citep{peleg2000distributed} remarks: \lq\lq\emph{The number of different models considered in the literature is only slightly lower than the number of papers published in this field}". But certainly, over the years, the research community working around this area have narrowed down to a handful of models. Out of them, the \CONGEST model has been quite prominent and captures several key aspects of typical communication networks.

In the \CONGEST model, as in most other distributed networks, each vertex acts as an independent compute node. Two nodes may communicate to each other in synchronous rounds only if they have a physical link between them. The message size is restricted to $O(\log n)$ bits across each link in any given round. We will focus our attention on communication which is more expensive than computation. 

In the \CONGEST model, the complexity of an algorithm is typically measured in terms of the number of rounds required for an algorithm to terminate. Parallels of this could be drawn to the time complexity of sequential algorithms. Much of the focus has been on the study of the distributed round complexity of algorithms but recently there has been some renewed interest to study the message complexity (the sum total of messages transmitted across all edges) as well \citep{Pandurangan:2017:TMD:3055399.3055449,Elkin:2017:SDD:3087801.3087823}. Apart from these some recent research is also being directed towards the study of asymptotic memory required at each node during the execution of the algorithm \citep{DBLP:journals/corr/ElkinN17}. In this thesis, we will limit our focus on the analysis of round complexity. 

\section{Classification of network problems based on locality}
Any problem pertaining to a communication network may depend either on local information or global information. Local information of a node comprises knowledge of immediate neighbors or nodes which are just a constant number of hops away. Global information includes knowledge of nodes across the network.

One of the prominent examples of a problem which just requires local information is of finding a \emph{maximal independent set} (MIS). An \emph{independent set} is a set of vertices such that no two vertices in the set have an edge between them. An MIS is an independent set such that no vertex can be added into it without violating independence. It has been shown by \cite{Luby:1985:SPA:22145.22146} that finding MIS requires $O(\log n)$ rounds. 

On the other side, to gather global information at all nodes, it takes at least $D$ rounds. This is because the farthest node could be as much as $D$ distance away. Thus the complexity of algorithms which require global information is $\Omega(D)$. 

Simple problems like finding the count of the number of nodes in the network, constructing a BFS tree, leader election etc are problems which require global information and takes $O(D)$ rounds. Apart from this, there are some problems which require global information and have been proved to take at least $\Omega(D + \sqrt{n})$ rounds. A prominent example is that of finding a Minimum Spanning Tree (MST). It has been shown that to even find an approximate MST we require $\Omega(D + \sqrt{n})$ rounds \citep{DasSarma:2011:DVH:1993636.1993686}. One of the key idea used in most of the algorithms which take $O(D + \sqrt{n})$ is a result from \citep{Kutten:1998:FDC:296903.296906}, which gives a distributed algorithm that can partition nodes of any spanning tree $T$ into $O(\sqrt{n})$ subtrees such that each subtree has $O(\sqrt{n})$ diameter. Basically, this approach can be considered as a classical divide and conquer paradigm of distributed algorithms. To solve any problem using this idea, nodes are divided into $O(\sqrt{n})$ clusters and then the problem is solved in each of the clusters and subsequently, a mechanism is given to combine the results.

\section{Overview of past results}
In this section, we will give an overview of the current progress in finding a min-cut. A more detailed technical review of the relevant historical results will be given in next chapter. 

In the centralized setting, the problem of finding min-cut has seen a lot of advances
\citep{
	Karger:1993:GMR:313559.313605,Karger:1994:URS:314464.314582,Karger:1994:RSC:195058.195422,karger1993o,nagamochi1992computing,Matula:1993:LTP:313559.313872,Gabow:1991:MAF:103418.103436,Stoer:1997:SMA:263867.263872,Karger:2000:MCN:331605.331608}. Recently, in the centralized setting even a linear time algorithm have been proposed for finding min-cut by \cite{kawarabayashi2015deterministic} and further improved by \cite{Henzinger:2017:LFP:3039686.3039811}. 

A few attempts at finding min-cuts in the distributed setting (in the \CONGEST model in particular) have been made in past. These include work by \cite{ahuja1989efficient, thurimella1995sub, tsin2006efficient, DBLP:journals/corr/abs-cs-0602013, pritchard2008fast, nanongkai2014almost,ghaffari2016distributed}. A  decade-old result by \cite{pritchard2008fast}, currently stands out as the most efficient way to find min-cut of small sizes in particular of size one. They gave an $O(D)$ round algorithm using random circulation technique for finding min-cuts of size $1$. Their result of finding a min-cut of size one has also been extended to give a Las Vegas type randomized algorithms to find min-cut of size 2 in $O(D)$ rounds. But a deterministic analogue of the same has so far been elusive.

More recently \cite{nanongkai2014almost} have given a $O((\sqrt{n} \log^* n + D)k^4 \log^2 n)$\footnote{$\log^* n$ (usually read "log star"), is the number of times the logarithm function must be iteratively applied before the result is less than or equal to $1$} rounds algorithm to find min-cut of size $k$.  There have also been results to find an approximate min-cut. \cite{ghaffari2013distributed} gave a distributed algorithm that, for
any weighted graph and any $\epsilon \in (0,1)$, with high probability finds a cut of size at most $O(\epsilon^{-1}k)$ in $O(D) + \tilde{O}(\sqrt{n})$ \footnote{$\widetilde{O}$ hides  logarithmic factors} rounds, where $k$ is the size of the minimum cut. \cite{nanongkai2014almost} improves this aproximation factor and give a $(1 + \epsilon)$  approximation algorithm in $O((\sqrt{n}\log^* n + D) \epsilon^{-5} \log^3 n)$ time. Also, \cite{ghaffari2016distributed} have given a min-cut algorithm for planar graphs which finds a $(1-\epsilon)$ approximate min-cut in $\widetilde{O}(D)$ rounds.\john{I moved the \cite{ghaffari2016distributed} result to the end of the sequence. Also, I think the approximation notation is a bit inconsistent. Shouldn't it be an $(1+\epsilon)$-approximate algorithm in both cases? Why is \cite{nanongkai2014almost} a $(1+\epsilon)$-approximation?}

While a linear time algorithm for finding min-cut exists in the centralized setting, in distributed setting the best known exact algorithm still takes $O((\sqrt{n} \log^* n + D)k^4 \log^2 n)$ rounds to deterministically find a min-cut of size $k > 1$. Moreover, a $\widetilde{O}(D)$ round algorithm exists to find a min-cut in planar graphs. The following question arises: Is $\Omega(\sqrt{n} + D)$ a lower bound for finding min-cuts or can we save the large factor of $\sqrt{n}$ as known for planar graphs for at least the small min-cuts. We answer this question through our thesis. We present a new algorithm\john{Here again, I am rewriting to emphasize one algorithm.} to deterministically find all the min-cuts of size one, two, and three. For min-cut of size either one or two our algorithm takes $O(D)$ rounds. For min-cuts of size three our algorithm takes $O(D^2)$ rounds.

\section{Organization of this thesis}
In Chapter \ref{chpater:2_prelim_background}, we will survey simple distributed algorithms for the \CONGEST model and also review graph properties like the cuts spaces and cycle spaces. We also give a brief overview of previously know techniques to find min-cut in distributed setting, in particular the techniques based on greedy tree packing and random circulations.

In Chapter \ref{chpater:3_technical_overview}, we will present our characterization of trees and cuts. Here we will also give a road-map of the novel distributed techniques introduced in this thesis. 
In Chapter \ref{chapter:4_trsf} and Chapter \ref{chapter:5_min-cut-3}, we will present distributed algorithms which will ensure that whenever there exists a min-cut of the kind the required quantities by our charecterization will be communicated to at least one node in network. In Chapter \ref{chapter:4_trsf}, we introduce \emph{Tree Restricted Semigroup Function} and show that this is enough for finding min-cuts of size $1$ and $2$ optimally. Further, in Chapter \ref{chapter:5_min-cut-3}, we will introduce two new techniques: \emph{Sketching} and \emph{Layered Algorithm} to find min-cuts of size $3$.

%% file: thesis_subparts/chapter2/chapter2_main.tex
%TEX root = ../../thesis.tex
In this chapter, we will review fundamental network algorithms in the distributed setting. Among the distributed algorithms, we will discuss the construction of various types of spanning trees including Breadth First Search (BFS) Tree, Depth First Search (DFS) tree and Minimum Spanning Tree (MST). Further, we will discuss some simple tree based algorithmic paradigms generally used in distributed networks.  We also review  important and relevant concepts like cut spaces and cycle spaces and show that these are vector spaces and orthogonal to each other. Finally, we will survey two important recent techniques used to find min-cuts in distributed networks.

\section{Simple algorithms in Distributed Networks}
\label{section:chapter_2_simple_algorithms}
\input{thesis_subparts/chapter2/2_simple_algorithms.tex}

\section{Cuts Space and Cycle Space}
\label{section:chapter_2_cut_cycles}
\input{thesis_subparts/chapter2/2_cut_cycles.tex}

\section{Min-Cuts in Distributed Networks}
\label{section:chapter_2_past_techniques}
\input{thesis_subparts/chapter2/2_past_techniques.tex}

%% file: thesis_subparts/chapter2/2_simple_algorithms.tex
%TEX root = ../../thesis.tex
As mentioned in the earlier chapter, we use the \CONGEST model of distributed networks which considers point-to-point, synchronous communications between any two nodes in a network. In any given synchronous round, communication is only allowed between nodes with a physical link and is restricted to $O(\log n)$ bits. 

The network will be modeled as a simple undirected graph $G=(V,E)$. As stated earlier, we will use $n$ to denote the number of nodes, $m$ to denote the number of edges and $D$ as the diameter of the network. Each node in this network has a unique $id$ which is of size $O(\log n)$ bits. Whenever we say node $x$ in the network, we mean a node with $x$ as its $id$. 

Various algorithms in distributed setting start with constructing a tree. A tree gives structure to the problem by defining a relationship between nodes as decedents or as ancestors or neither of them. Moreover, it also serves as a medium for information flow between nodes in a structured way using the tree edges. In this section, we state results about the construction of various types of spanning trees and review simple distributed algorithms on trees. 

\subsection{Tree Constructions and Notations}
Let ${T}$ be any spanning tree rooted at node $r_T$ (or just $r$ when clear from the context). Let $\level[T]{x}$ be the level of node $x$ in $T$ such that the root $r$ is at level $0$ and for all other nodes $x$, $\level[T]{x}$ is just the distance from root following the tree edges. For $x \neq r$, let $\parent[T]{x}$ denote the parent of the vertex $x$ in $T$. For every node $v$, let $\child_T(v)$ be the set of children of $v$ in the tree $T$.  Further for any node $v$, we will use $\ancestor[T]{v}$ to denote the set of all ancestors of node $v$ in tree $T$ including $v$ and $\desc[T]{v}$ as the set of vertices which are descendants of $v$ in the tree $T$ including $v$. We will briefly review construction of different types of trees. For details the reader is referred to \citep{peleg2000distributed}.
\paragraph{BFS Tree.} Our approach to find min-cuts uses an arbitrary BFS tree. Although, our technique is invariant of the type of the spanning tree chosen but a BFS tree is ideal because it can be constructed in just $O(D)$ rounds. The following lemma states the guarantees regarding construction of a BFS trees
\begin{lemma}
	\label{lemma:bfs-tree}
	There exists a distributed algorithm for construction of a BFS tree which requires at most $O(D)$ rounds.
\end{lemma}
A brief description of an algorithm to construct BFS tree as follows: an arbitrary node $r$ initiates the construction of BFS tree, it assigns itself a level $\level{r} = 0$. 
Then the node $r$ advertises to all its adjacent neighbors about its level, who then joins the BFS tree at level $1$ as children of $r$. Subsequently, each node at level $1$ also advertises its level to all its neighbors, but only those nodes join the BFS tree as children of nodes at level $1$ who are not part of the BFS tree already. Here, ties are broken arbitrarily. This process goes on until all the nodes are part of the BFS tree. 

\paragraph{DFS Tree.} Unlike BFS tree, the depth of a DFS tree is oblivious of the diameter. In fact, two of the early approaches to find cut edges in distributed networks used a DFS tree \citep{Hohberg:1990:FBC:89472.89477,ahuja1989efficient}. The following lemma states the guarantees regarding construction of a DFS trees. 
\begin{lemma}
	There exists a distributed algorithm for construction of a DFS tree which requires at most $O(n)$ rounds.
\end{lemma}

\paragraph{MST Tree.} The following lemma is known regarding the complexity to construct an MST tree. In fact, it is also known that we cannot do better in terms of round complexity for construction of the MST tree. \citep{DasSarma:2011:DVH:1993636.1993686}.
\begin{lemma}
For a weighted network, there exists a distributed algorithm for construction of MST tree which requires $O(\sqrt n + D)$ rounds.
\end{lemma} 
A recent technique to find min-cuts given by \cite{nanongkai2014almost} assigns weights to each edge in an unweighted graph and compute MST tree, further these weights are updated and a new MST tree is constructed. This goes on for $\plog n$ iterations. We review this technique in Section \ref{section:chapter_2_past_techniques}.

\subsection{Convergecast and Broadcast}
\label{sec:covergecast_broadcast}
We will now give brief details about simple distributed algorithms. In particular, we will describe \emph{Broadcasts} and \emph{Convergecasts}. These are mechanisms to move data along the edges of any fixed tree. In this thesis, we will often use these techniques to compute properties of the given network which will, in turn, help us to find min-cuts.

\paragraph{Convergecasts} In distributed network algorithms, there are applications where it is required to collect information upwards on some tree, in particular, this is a flow of information from nodes to their ancestors in the tree. This type of a technique is called as \emph{Convergecast}. In a BFS tree $T$, any node has at most $D$ ancestors, but for any node $v$ the number of nodes in descendants set $\desc[T]{v}$ is not bounded by $D$. Thus during convergecast, all nodes in the set $\desc[T]{v}$ might want to send information upwards to one of the ancestor node $v$. This may cause a lot of contentions and is a costly endeavor. In this thesis, whenever we use a covergecast technique, we couple it with either an \emph{aggregation based strategy} or a \emph{contention resolution mechanism}.

In an aggregation based strategy, during covergecast, the flow of information from nodes at a deeper level is synchronized and aggregated together to be sent upwards. In Chapter \ref{chapter:4_trsf}, we introduce \emph{tree restricted semigroup function} which is such an aggregation based technique. Moreover, our algorithm to compute the specially designed \emph{Sketch} also uses this kind of convergecasting. Further, in a contention resolution mechanism, only certain information are allowed to move upwards based on some specific criteria, thus limiting contention. In Chapter \ref{chapter:5_min-cut-3}, we use such contention resolution mechanism to convergecast details found by our \emph{Layered Algorithm}.

\paragraph{Broadcasts} In this technique, the flow of information takes place in the downward direction along the tree edges. When such an information is aimed to be communicated to all the nodes in the decedent set then it is called \emph{broadcast}. In the general broadcast algorithm, for some spanning tree $T$, the root is required to send a message of size $O(\log n)$ bits to all the nodes. We state the following lemma from \cite{peleg2000distributed} whose proof uses a simple flooding based algorithm.
\begin{lemma}[Lemma 3.2.1 from \citep{peleg2000distributed}]
	For any spanning tree $T$, let the root have a message of size $O(\log n)$ bits, then it can be communicated to all the nodes in $O(Depth(T))$ rounds.
	\label{lemma:broadcast_general}
\end{lemma}
\begin{proof}
	Here we assume that there exists a spanning tree $T$ and each node is aware of the spanning tree edges incident to it. Now, all we need is to flood the message to the entire tree. In the first round, the root sends this message to all its children, who are at level $1$. Further, in the second round, nodes at level $1$ sends the message to all their children who are at level $2$ and so on. Thus, in all, we will require $O(Depth(T))$ rounds. 
\end{proof}

At various stages of the algorithms presented in this thesis, we require two variants of the general broadcast techniques. They are described below. We also give the respective algorithm for them in Lemma \ref{lemma:broadcat_types}.
\begin{enumerate}[itemindent=1cm]
	\item[\emph{Broadcast Type-1}] Given a spanning tree $T$, {\em all} nodes $v \in V$ have a message of size $O(\log n)$ bits which is required to be communicated to all the nodes in the vertex set $\desc[T]{v}$.
	\item[\emph{Broadcast Type-2}] Given a spanning tree $T$, {\em all} nodes $v \in V$ have $Depth(T)$ messages of size $O(\log n)$ bits which is required to be communicated to all the nodes in the vertex set $\desc[T]{v}$.
\end{enumerate} 
\begin{lemma}
	\label{lemma:broadcat_types}
	There exist distributed algorithms which require $O(Depth(T))$ and \sloppy{$O(Depth(T)^2)$} rounds respectively for Broadcast Type-1 and Broadcast Type-2.
	\label{lemma:broadcast_1_2}
\end{lemma}
\begin{proof}
	First, let us consider \emph{Broadcast Type-1}. 
	Here we will use Lemma \ref{lemma:broadcast_general} coupled with \emph{pipelining}. 
	Each node $v$ initiates a broadcast of its $O(\log n)$ bits sized message to the subtree rooted at node $v$. using the algorithm from Lemma \ref{lemma:broadcast_general} in the first round. 
	If $v$ is not the root of the spanning tree, then in the first round it will have received the message broadcasted by its parent and it is imperative to node $v$ to send this message to all its children. 
	Note that this does not hinder with the broadcast initiated by node $v$ to the subtree rooted at it.
	In all the subsequent rounds a train of messages is \emph{pipelined} one after another through the tree edges, enabling broadcasts initiated by nodes to run without any hindrance or congestion. 
	Thus it only takes $O(Depth(T))$ rounds for this type of broadcast. 
	For \emph{Broadcast Type-2}, we just need to run $Depth(T)$ iterations of \sloppy\emph{Broadcast Type-1} taking $O(Depth(T)^2)$ rounds.
\end{proof}

%% file: thesis_subparts/chapter2/2_cut_cycles.tex
%TEX root = ../../thesis.tex
In this section, we will discuss about the \emph{cut space} and \emph{cycle spaces}. We will first define \emph{cut spaces}, \emph{cycle spaces} and prove that, they are vector spaces. Further, we will show that they are orthogonal. All these facts are well known and have been part of standard graph theory books, for example, see \citep{bondy1976graph}.

As defined earlier, $E$ and $V$ are edge set and vertex set respectively. We will work with the field $\mathbb{Z}_{2}$ and the
space $\mathbb{Z}_{2}^{|E|}$.  Let $S \subseteq E$,
then its characteristics vector is defined as follows: $\chi_{e}^{S} = 1$  (at coordinate $e$) for $e \in S$ and
$\chi_{f}^{S} = 0$ for $f \notin E$. Here we will use the operator $\oplus$ (symmetric difference or XOR) for the space $\mathbb{Z}_{2}^{|E|}$. Let
$R,S \subseteq E$. 
Note that $R,S$ are edge sets whereas $\chi^{R}$ and $\chi^{S}$ are binary vectors in $\mathbb{Z}_{2}^{|E|}$. 
When we use $R \oplus S$, we will mean symmetric difference of the two sets and when we use $\chi^{R} \oplus \chi^{S}$, we will mean XOR operation between the two vectors. 
It is easy to see that $\chi^{R} \oplus \chi^{S} = \chi^{R \oplus S}$

Let $\phi \subseteq E$. If the subgraph $(V,\phi)$ has all the vertex with even degrees then $\phi$ is a \emph{binary circulation}. The set of all binary circulations of a graph is called the \emph{cycle space} of the graph. Similarly, let $A \subseteq V$. Define an induced cut $\delta(A)$ as the set of edges with exactly one endpoint in $A$ or essentially the cut set $(A,V\setminus A)$. The set of all the induced cuts	is called the \emph{cut space} of the graph.

\begin{thm}
	The cut space and cycle space are vector spaces of $\mathbb{Z}_{2}^{|E|}$
\end{thm}

\begin{proof}
	To show that both the structures are vector spaces it is sufficient to show that there exists a $0$ for the vector space, an addition operator and inverse of each element. 
	We will show that $\emptyset_E$ is the required $0$ of the vector space and $\oplus$ is the operator. 
	For cut space take a set $A = \emptyset$. 
	Then $\delta(A) = \emptyset_E$ that is $\chi^{\delta(A)} = \vec 0_2$. 
	Similarly for cycle space if take $\phi = \emptyset \subset E$ then all the vertex in the graph $(V,\phi)$ have degree as $0$ which is even. 
	Let $A,B \subseteq V$. 
	Then $\delta(A)$ and $\delta(B)$ are induced cut defined for A,B. Let $A \oplus B$ be the symmetric difference of the set $A$ and $B$. 
	As per the definition $\delta(A)$ is the set of edges which have exactly one end in $A$ and similarly for $\delta(B)$ which are the edges with one end in $B$. 
	Now $A \oplus B$ is a symmetric difference between the set $A,B$  and $\delta(A \oplus B)$ is the set of edges with one end in $A \oplus B$. 
	To see that $\delta(A) \oplus \delta(B)$ is equal to $\delta(A \oplus B)$ just take two cases where $A \cap B = \emptyset$ and $A \cap B \not= \emptyset$. 
	
	Now let us turn our attention to  the cycle space. Let $R,S\subseteq E$. And let $G_1 = (V,R)$ and $G_2 = (V,S)$ be two sub-graphs such that all vertices are of even degrees in them; that is both $R,S$ are binary circulations. Now consider the subgraph $G_{1,2} = (V,R\oplus S)$ and a vertex $v \in V$. The degree of vertex $v$ in $G_{1,2}$ is $\deg_R(v) + \deg_S(v) - 2\deg_{R\cap S}(v)$ which is even. Hence $R\oplus S$ is also a binary circulation.
\end{proof}
The following corollary will be useful for proving our characterization of cuts and trees given in 
Chapter \ref{chpater:3_technical_overview}.
\begin{cor}
	\label{main-cor}
	Let $A_1,A_2,\ldots, A_j$ be set of vertices such that $\forall i\ A_i \subset V$ then $\delta(A_1 \oplus A_2 \oplus \ldots \oplus A_j) = \delta(A_1) \oplus \delta(A_2) \oplus \ldots \oplus \delta(A_j)$
\end{cor}

\begin{thm}
	The cut space and cycle space are orthogonal.
\end{thm}

\begin{proof}
	Let $A \subset V$ and $\phi$ be a binary circulation. Let $\chi^\phi$ and $\chi^{\delta(A)}$ be the vectors corresponding to the $\phi$ and induced cut $\delta(A)$.
	We have to show that $\chi^{\delta(A)}.\chi^\phi = 0\paren{\mod 2}$ which is equivalent to showing that $|\chi^\phi \cap \chi^{\delta(A)}|$ is even.
	Now $\sum_{a \in A} \deg_\phi(a) = \sum_{a\in A}|\phi \cap \delta(a)|$. Observe that $\sum_a \deg_\phi(a)$ is even because $\phi$ is a binary circulation. Further $\sum_{a\in A}|\phi \cap \delta(a)|$ counts the edges with both ends in $A$ twice and every edge in $|\phi \cap \delta(A)|$ once and does not count any other edge. Hence $|\phi \cap \delta(A)|$ is even.
\end{proof}

%% file: thesis_subparts/chapter2/2_past_techniques.tex
%TEX root = ../../thesis.tex
In this section we will review some past results for finding min-cuts, in particular, the ideas employed in two recent works \citep{pritchard2011TALG,nanongkai2014almost}. While \cite{nanongkai2014almost} relied on greedy tree packing introduced by \cite{thorup2001fully}; \cite{pritchard2011TALG} gave a novel technique of random circulations.
\subsection{Using Greedy Tree Packing}
First, we will define the concept of greedy tree packing. A \emph{tree packing} $\mathbb{T}$ is a multiset of spanning trees. Let the load of any edge $e$ w.r.t a tree packing $\mathbb{T}$, denoted  by $\mathcal L^{\mathbb{T}}(e)$, be defined as the number of trees in $\mathbb{T}$ containing $e$. A tree packing $\mathbb{T} =\curly{T_1,\ldots, T_j}$ is greedy if each $T_i$
is a minimum spanning tree (MST) with respect to the loads induced by $\curly{T_1,\ldots, T_{i-1}}$. \cite{thorup2001fully} has given the following results related to tree packing:

\begin{thm}[\cite{thorup2001fully}]
	Let $G = (V,E)$ be an unweighted graph. Let $m$ be the number of edges in $G$ and $k$ be the size of min-cut. Then a greedy tree-packing with $\omega(k^7 \log ^3 m)$ contains a tree crossing some min-cut exactly once.
	\label{thorup-tree-packing}
\end{thm}
Based on the above theorem, \cite{nanongkai2014almost} construct a greedy tree packing of $\omega(k^7 \log ^3 m)$ trees. There is a well-known algorithm to find MST in $O(D + \sqrt{n})$ rounds. Further, in each tree, the authors find the size of the smallest edge cut which shares only one edge with the tree. They give an algorithm to do the same in $O(D + \sqrt{n})$ rounds for each tree. Note that they cannot find all the min-cuts because of the limitations of Theorem \ref{thorup-tree-packing} which only guarantees that there will exists some (but not all) min-cut which will share exactly one edge with at least one tree in {tree packing} $\mathbb{T}$.

\subsection{Using Random Circulations}
The random circulation technique has its foundation based on a well-known fact that cut spaces and cycle spaces are orthogonal to each other \citep{bondy1976graph}. Based on this technique \cite{pritchard2008fast} gave the following result
\begin{thm}[\cite{pritchard2008fast}]
	There is a Las Vegas distributed algorithm to compute all the min-cuts of size 1 and 2 (cut edges and cut pairs) in $O(D)$ time.
\end{thm}
The above result cannot be made deterministic for min-cuts of size 2. Moreover extending the random circulation technique for min-cuts of size $3$ or more do not\john{I added ``do not". } seem plausible due to its fundamental limitations.  But as shown in \citep{DBLP:journals/corr/abs-cs-0602013}, 1-cuts can be found deterministically in $O(D)$ rounds.

%% file: thesis_subparts/chapter3/chapter_3.tex
%TEX root = ../../thesis.tex
We give deterministic algorithm for finding  min-cuts of size one, two, and\john{Just to confirm, we will be able to output {\em all} min-cuts of size 1, 2, and 3, right? Previously: ... for finding if there exists a min-cut of size one, two, and three.} three.
We find the min-cut of size one and two in $O(D)$ time and of size three in $O(D^2)$ time. We recreate the optimal result for min-cut of size one.
For min-cuts of size two and three our results resolve the open problem from \cite{pritchard2011TALG}.
We give a new characterization involving trees and min-cuts. 
We also introduce a new algorithmic technique named as the \textit{Tree restricted Semigroup function} (TRSF) which is defined with respect to a tree $T$. The TRSF is a natural approach to find min-cuts of size one. We also show a non-trivial application of TRSF in finding  min-cuts of size two, which is quite intriguing. 
For finding, if there exists a min-cut of size $3$ we introduce a new \emph{sketching technique} where each node finds a view of the graph with respect to a tree by getting rid of edges which may not be a part of a min-cut. 
The sketching idea coupled with our characterization results helps us to find the required min-cut.
\section{Characterization of Trees and Cuts}
In this subsection, we establish some fundamental relationships between trees and cuts which will help us find the required min-cuts. When a min-cut shares $k$ edges with a spanning tree then we say that it $k$-respects the tree. Each min-cut will at least share one edge with a spanning tree otherwise the said min-cut will not be a cut because there will exist a spanning tree which connects all the nodes. We give the following simple observation in this regard.
\begin{obs}
	\label{obs:fundamental_observation_about_cut_respecting_trees}
For any induced cut of size k and a spanning tree $T$, there exists some $j \in [1, \min(k, n)]$ such that induced cut $j -$ respects the tree $T$.
\end{obs}

The above simple observation lets us break the problem of finding small min-cuts into different cases. A min-cut of size $1$ will always share one edge with any spanning tree. A min-cut of size $2$ shares either $1$ or $2$ edges with any spanning tree. Similarly a min-cut of size $3$ either shares $1,2$ or $3$ edges with any spanning tree. 

We now make a simple observation along with a quick proof sketch
\begin{obs}
	When an induced cut of size $k$ 1-respects a spanning tree $T$, then there exists at least one node $v$ such that $|\delta(\desc[T]{v})| = k$.
\end{obs}
\begin{proof}[Proof Sketch]
Consider the cut edge $(u,v)$ that is  in $T$ with $u$ being closer to the root $r_T$. Since the tree only 1-respects the cut, $\desc[T]{v}$ remains on one side of the cut, while $r_T$ is on the other. Moreover, $V \setminus \desc[T]{v}$ fully remains on the side with $r_T$; otherwise, $T$ will not be limited to 1-respecting the cut. Thus, $\desc[T]{v}$ will serve our required purpose.
\end{proof}
 We will now give a distributed algorithm such that each node $v$ finds $|\delta(\desc[T]{v})|$ in $O(D)$ rounds. For the remaining cases, we will give lemmas to characterize the induced cuts and tree. These characterizations follow from Corollary \ref{main-cor}. We begin with the characterization of an induced cut  of size 2 when it 2-respects a tree $T$. 

% % % % % % % % % % % % % % % % % % % % % % % % % % % %
% % % % Begin: Characterizing Lemma about 2-cuts % % % %
% % % % % % % % % % % % % % % % % % % % % % % % % % % %

\begin{lemma}[2-cut, 2-respects a spanning tree $T$]
\label{lemma:2-respect-2-cut}
Let $T$ be any spanning tree. When $u \neq v$ and $u,v \in V\setminus r$. Then the following two statements are equivalent.
\begin{enumerate}[$P_1:$]
\item $\curly{(\parent[T] u,u),(\parent[T] v,v)}$ is a cut set induced by ${\desc[T]{u}}\oplus {\desc[T]{v}}$.
\item $|\delta(\desc[T]{u})| = |\delta(\desc[T]{v})| = |\delta(\desc[T]{u}) \cap \delta\paren{\desc[T]{v}}| + 1$.
\end{enumerate}
\end{lemma}
To prove the above lemma the following simple observation will be helpful.
\begin{obs}
\label{obs:simple_observation_2_cuts}
For any $u\neq v$ and $u,v \in V\setminus r$, and a spanning tree $T$ the edge $(\parent[T]{v},v) \in \delta({\desc[T]{v}})$  but $(\parent[T]{v},v) \notin \delta({\desc[T]{u}})$; also $(\parent[T]{u},u) \in \delta({\desc[T]{u}})$  but $(\parent[T]{u},u) \notin \delta({\desc[T]{v}})$. 
\end{obs}
\begin{proof}
Consider the given spanning tree $T$. Here the edge $(\parent[T]{u},u)$ has one end in the vertex set $\desc[T]{u}$ and the other end in the vertex set $V \setminus \desc[T]{u}$. Thus the edge $(u,\parent[T]{u})$ is part of the cut set $(\desc[T]{u}, V\setminus \desc[T]{u}) = \delta(\desc[T]{u})$. Now, consider any other node $v$. Either $v \in \desc[T]{u}$ or it does not. Let's take the first case when $v \in \desc[T]{u}$. Since $v \neq u$, thus both the end points of the edge $(\parent[T]{u},u)$ are outside the set $\desc[T]{v}$ hence $(\parent[T]{u},u) \notin \delta(\desc[T]{v})$.

When $v \notin \desc[T]{u}$, then also similar argument holds. If $v \in \ancestor[T]{u}$, then both the endpoints are in the set $\desc[T]{v}$, otherwise both of them are outside the set $\desc[T]{v}$. In either of them the edge $(\parent[T]{u},u) \notin \delta(\desc[T]{v})$.
\end{proof}
\begin{proof}[Proof of Lemma \ref{lemma:2-respect-2-cut}]
Consider the forward direction, $\curly{(\parent[T] u,u),(\parent[T]v,v)}$ is a cut set induced by ${\desc[T]{u}}\oplus {\desc[T]{v}}$. Therefore $\delta({\desc[T]{v}} \oplus {\desc[T]{u}}) = \curly{(\parent[T]u,u),(\parent[T]v,v)}$. Further from Corollary \ref{main-cor}, we have $\delta({\desc[T]{v}} \oplus {\desc[T]{u}})=\delta({\desc[T]{u}}) \oplus \delta({\desc[T]{v}})$. Therefore $\curly{(\parent[T]u,u),(\parent[T]v,v)} = \delta({\desc[T]{u}}) \oplus \delta({\desc[T]{v}})$ which along with Observation \ref{obs:simple_observation_2_cuts} implies that $\delta({\desc[T]{u}})\setminus (\parent[T]{u},u) = \delta({\desc[T]{v}})\setminus(\parent[T]{v},v)$ thus $|\delta({\desc[T]{v}})|-1 = |\delta({\desc[T]{v}}) \cap\delta({\desc[T]{u}})|$ and $|\delta({\desc[T]{u}})|- 1= |\delta({\desc[T]{v}}) \cap\delta({\desc[T]{u}})|$. 

For the other direction, given that $|\delta({\desc[T]{u}}) \cap \delta({\desc[T]{v}})| + 1 = |\delta({\desc[T]{v}})| = |\delta({\desc[T]{u}})|$, which along with Observation \ref{obs:simple_observation_2_cuts} implies that $\delta({\desc[T]{u}})\setminus (\parent[T]{u},u) = \delta({\desc[T]{v}})\setminus (\parent[T]{v},v)$. Hence $\delta({\desc[T]{u}}) \oplus \delta({\desc[T]{v}}) = \curly{(\parent[T]{u},u),(\parent[T]{v},v)}$. Using the fact that $\delta({\desc[T]{v}}) \oplus \delta({\desc[T]{u}}) = \delta({\desc[T]{v}} \oplus {\desc[T]{u}})$, the edge set $\curly{(\parent[T]u,u),(\parent[T]v,v)}$ is a cut induced by ${\desc[T]{v}} \oplus {\desc[T]{u}}$. 
\end{proof}

When an induced cut of size $2$, 2-respects a spanning tree there could be two sub-cases. Consider nodes $u,v$ as in the above lemma. Here either $\desc[T]{u} \cap \desc[T]{v} = \emptyset$ or $\desc[T]{u} \subseteq \desc[T]{v}$ (which is same as $\desc[T]{v} \subseteq \desc[T]{u}$). The above lemma unifies these two cases. In the next chapter, we will show that how one of the nodes in the network finds all the three quantities as required by the above lemma when there exists a min-cut of this kind.

Now we will see the similar characterizations for an induced cut of size 3 when it 2 respects the cut and 3-respects the cut.

% % % % % % % % % % % % % % % % % % % % % % % % % % % %
% % % % End: Characterizing Lemma about 2-cuts % % % %
% % % % % % % % % % % % % % % % % % % % % % % % % % % %
% % % % % % % % % % % % % % % % % % % % % % % % % % % % % % % % %
% % % % % % % Begin: 3-cut which 2 respect% % % % % % % % % % % % 
% % % % % % % Characterizing Lemma Introduction % % % % % % % % % 
% % % % % % % % % % % % % % % % % % % % % % % % % % % % % % % % %
\begin{lemma}[3-cut, 2-respects a spanning tree $T$]
Let $T$ be any tree. Let $v_1,v_2 \in V\setminus r$ be two different nodes and $e$ be a non-tree edge, then the following two are equivalent.
\begin{enumerate}[$P_1:$]
\item $\curly{(\parent[T]{v_1},v_1),(\parent[T]{v_2},v_2),e}$ is a cut set induced by the vertex set $\desc[T]{v_1}\oplus \desc[T]{v_2}$.
\item $|\delta(\desc[T]{v_1})| -2 = |\delta(\desc[T]{v_2})| -1 = \gamma({\desc[T]{v_1}},{\desc[T]{v_2}})$ or $\delta(\desc[T]{v_1}) -1 = |\delta(\desc[T]{v_2})| -2 = \gamma({\desc[T]{v_1}},{\desc[T]{v_2}})$. 
\end{enumerate}
\label{lemma:gen_3_cut_2_respect}
\end{lemma}

\begin{proof}
The proof is similar to Lemma \ref{lemma:2-respect-2-cut}. 
Consider the forward direction, \sloppy$\curly{(\parent{v_1},v_1),(\parent{v_2},v_2),e}$ is a cut set induced by ${\desc{v_1}}\oplus {\desc{v_2}}$. Therefore $\delta({\desc{v_2}} \oplus {\desc{v_1}}) = \curly{(\parent{v_1},v_1),(\parent{v_2},v_2),e}$. 
Further from Corollary \ref{main-cor}, we have $\delta({\desc{v_1}} \oplus {\desc{v_2}})=\delta({\desc{v_1}}) \oplus \delta({\desc{v_2}})$. Therefore $\curly{(\parent{v_1},v_1),(\parent{v_2},v_2),e} = \delta({\desc{v_1}}) \oplus \delta({\desc{v_2}})$. Since $e \in \delta(\desc{v_1}) \oplus \delta(\desc{v_2})$. Thus $e \in  \delta(\desc{v_1})$ or  $e \in \delta(\desc{v_2})$ but not in both due to the symmetric difference operator. 
Without loss in generality, let $e \in \delta(\desc{v_1})$. Now Observation \ref{obs:simple_observation_2_cuts} implies that $\delta({\desc{v_1}})\setminus \curly{(\parent{v_1},v_1),e} = \delta({\desc{v_2}})\setminus(\parent{v_2},v_2)$ thus $|\delta({\desc{v_1}})|-2 = |\delta({\desc{v_2}}) \cap\delta({\desc{v_2}})|$ 
and $|\delta({\desc{v_2}})|- 1= |\delta({\desc{v_1}}) \cap\delta({\desc{v_2}})|$. 
Similarly when $e \in \delta(\desc{v_2})$ then $|\delta({\desc{v_1}})|-1 = |\delta({\desc{v_2}}) \cap\delta({\desc{v_2}})|$ and $|\delta({\desc{v_2}})|- 2= |\delta({\desc{v_1}}) \cap\delta({\desc{v_2}})|$

For the other direction, without loss in generality, let us choose one of the two statements. In particular, let $|\delta({\desc{v_1}}) \cap \delta({\desc{v_2}})| = |\delta({\desc{v_1}})| - 1 = |\delta({\desc{v_2}})| -2$, which along with Observation \ref{obs:simple_observation_2_cuts} implies that $|\delta({\desc{v_1}})\setminus (\parent{v_1},v_1)| = |\delta({\desc{v_2}})\setminus (\parent{v_2},v_2)| -1$. 
Thus there exists exactly one edge $e$ which is not in $\delta(\desc{v_1})$ but in $\delta(\desc{v_2})$
Hence $\delta({\desc{v_1}}) \oplus \delta({\desc{v_2}}) = \curly{(\parent{v_1},v_1),(\parent{v_2},v_2),e}$. Using the fact that $\delta({\desc{v_2}}) \oplus \delta({\desc{v_1}}) = \delta({\desc{v_2}} \oplus {\desc{v_1}})$; the edge set $\curly{(\parent{v_1},v_1),(\parent{v_2},v_2),e}$ is a cut induced by ${\desc{v_1}} \oplus {\desc{v_2}}$.
\end{proof}
% % % % % % % % % % % % % % % % % % % % % % % % % % % % % % % % %
% % % % % % % Begin: Characterizing Lemma % % % % % % % % % % % % 
% % % % % % % Lemmas. for 3-cut 3-respect % % % % % % % % % % % % 
% % % % % % % % % % % % % % % % % % % % % % % % % % % % % % % % %
\begin{lemma}[3-cut, 3-respects a spanning tree $T$]
\label{lemma:gen_3_cut_3_respect}
Let $T$ be any tree. Let $v_1,v_2,v_3 \in V\setminus r$, further each of $v_1,v_2$ and $v_3$ are different then the following two statements are equivalent:\vspace{-0.7cm}
\begin{enumerate}[$P_1:$]
\item $\curly{(v_1,\parent[T]{v_1}),(v_2,\parent[T]{v_2}),(v_3,\parent[T]{v_3})}$ is a cut-set induced by $\desc[T]{v_1}\oplus \desc[T]{v_2} \oplus \desc[T]{v_3}$.
\item $|\delta(\desc[T]{v_i})| -1 = \sum\limits_{j \in \curly{1,2,3}\setminus \curly{i}}\gamma(\desc[T]{v_i},\desc[T]{v_j}),\ \forall i \in \curly{1,2,3}$.
\end{enumerate}
\end{lemma}

\begin{proof}
	Consider the forward direction  $\delta(\desc{v_1}\ \oplus\ \desc{v_2}\ \oplus\ \desc{v_3}) = \{(v_1,\parent{v_1}),(v_2,\parent{v_2}),$ \sloppy $(v_3,\parent{v_3})\}$. Also, by Corollary \ref{main-cor}, we know that $ \delta(\desc{v_1}) \oplus \delta(\desc{v_2}) \oplus \delta(\desc{v_3}) = \delta(\desc{v_1}\oplus \desc{v_2} \oplus \desc{v_3})$. 
	Notice that there are 3 equations in $P_2$ when $i \in [1,3]$. We will show for $i=1$ and the argument for the rest will follow similarly. 
	For all edge $e \neq (v_1,\parent{v_1})$ and $e \in \delta(\desc{v_1})$ then $e$ must be present exactly once in either of $\delta(\desc{v_2})$ or $\delta(\desc{v_3})$. Because if it is present in none of them then among the three sets $\delta(\desc{v_1}),\delta(\desc{v_2})$ and $\delta(\desc{v_3})$, $e$ is exactly in one that is $\delta(\desc{v_1})$. Thus $e \in \delta(\desc{v_1}\oplus \desc{v_2} \oplus \desc{v_3})$ which is not true. Also if it is present in both $\delta(\desc{v_2})$ and $\delta(\desc{v_3})$ then similarly $e \in \delta(\desc{v_1}\oplus \desc{v_2} \oplus \desc{v_3})$, which again is not true. Thus apart from the edge $(\parent{v_1},v_1)$ whenever there exists an $e \in \delta(\desc{v_1})$ then it is also present in exactly one of $\delta(\desc{v_2})$ or $\delta(\desc{v_3})$.  Thus $|\delta(\desc{v_1})| - 1  = |\delta(\desc{v_1}) \cap \delta(\desc{v_2})| + |\delta(\desc{v_1}) \cap \delta(\desc{v_3})|$.
	
	For the backward direction, we will again focus our attention to three different edge sets $\delta(\desc{v_1}),\delta(\desc{v_2})$ and $\delta(\desc{v_3})$. 
	We know that the edge $(\parent{v_1},v_1) \in \delta(\desc{v_1})$ but not in either of $\delta(\desc{v_2})$ and $\delta(\desc{v_3})$. 3
	Similarly for $(\parent{v_2},v_2)$ and $(\parent{v_3},v_3)$. Lets the edge set $C = \delta(\desc{v_1}) \oplus \delta(\desc{v_2}) \oplus \delta(\desc{v_3})$. 
	Thus $C$ is sure to contain the edges $(v_1,\parent{v_1}),(v_2,\parent{v_2})$ and $(v_3,\parent{v_3})$ because of the symmetric difference operator between the three sets.  
	Now based on the condition in $P_2$ we will show that there is no other edge $e$ in $C$. Imagine there is only one edge $e \neq (\parent{v_1},v_1)$ and $e \in \delta(\desc{v_1})$ and $e \in C$. Thus $e$ could either be in both $\delta(\desc{v_2})$ and $\delta(\desc{v_3})$ or none of them. Imagine $e$ is in both $\delta(\desc{v_2})$ and $\delta(\desc{v_3})$, then $|\delta(\desc{v_1})| -2 = |\delta(\desc{v_1}) \cap \delta(\desc{v_2})| + |\delta(\desc{v_1}) \cap \delta(\desc{v_3})|$ which contradicts $P_2$. Imagine $e$ is in none of $\delta(\desc{v_2})$ and $\delta(\desc{v_3})$, then $|\delta(\desc{v_1})| = |\delta(\desc{v_1}) \cap \delta(\desc{v_2})| + |\delta(\desc{v_1}) \cap \delta(\desc{v_3})|$ which again contradicts $P_2$. 
\end{proof}
Lemma \ref{lemma:gen_3_cut_2_respect} and \ref{lemma:gen_3_cut_3_respect} characterize an induced cut of size 3 when it shares 2 edges and 3 edges with a spanning tree respectively. Similar to a min-cut of size 2, here again we have few different cases. These lemmas unify all the different cases. We will discuss these different cases in Chapter \ref{chapter:5_min-cut-3}.

We will work with an arbitrary BFS tree $\tree$ of the given network. Whenever a quantity is computed with respect to this BFS tree $\tree$, we will skip the redundant $\tree$ from the subscript or superscript, for example $\parent{v}$ instead of $\parent[\tree]{v}$ and $\desc{v}$ instead of $\desc[\tree]{v}$. At the beginning each node $v\neq r$ knows $\level{v}$, $\parent{v}$ and the ancestor set $\ancestor{v}$. The BFS tree and these associated quantities can be computed in the beginning in $O(D)$ rounds. For any node $v$ which is not the root of the BFS tree, $\level{v}$ and  $\parent{v}$ are known to $v$ at the time of construction of the BFS tree as shown in the proof of the Lemma \ref{lemma:bfs-tree}. Further the ancestor set $\ancestor{v}$ can also be known to node $v$ in $O(D)$ rounds using \emph{Broadcast Type - 1}. And can be done as follows: each node $a$ tells all the nodes in the set $\desc{a}$ about it being one the ancestor. This information is just of size $O(\log n)$ bits. 

\section{Our Contribution}
In this thesis, we present deterministic algorithm to find min-cuts of size 1,2 and 3. 
\begin{figure}[h]
	\centering
	\includegraphics{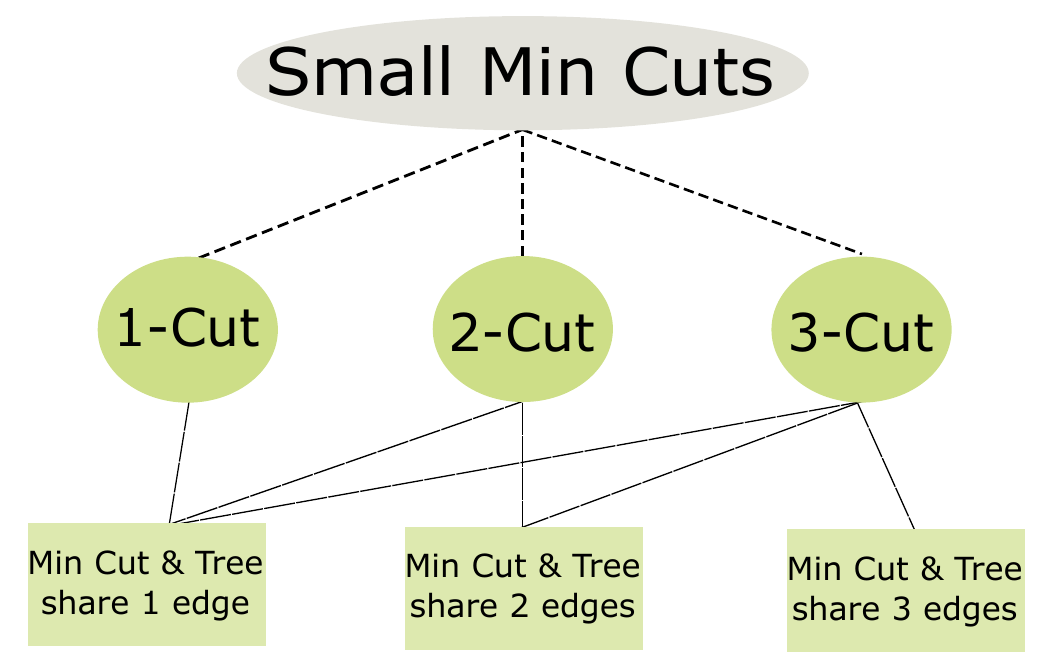}
	\caption{Roadmap for finding small min-cuts}
	\label{figure:roadmap}
\end{figure}
For min-cuts of size $1$ and $2$, our algorithm takes $O(D)$ rounds. For min-cuts of size $3$, our algorithm takes $O(D^2)$ rounds.

We use our characterization given in Lemma \ref{lemma:2-respect-2-cut}, Lemma \ref{lemma:gen_3_cut_2_respect} and Lemma \ref{lemma:gen_3_cut_3_respect}. Further, in subsequent chapters, we will give communication strategies which will ensure that at least one node in the network knows about the required information as per these lemmas whenever a min-cut of the kind exists. The idea is basically to break the problem into smaller parts. Recall from Observation \ref{obs:fundamental_observation_about_cut_respecting_trees} that a min-cut of size 1 shares an edge with any spanning tree, a min-cut of size 2 may share 1 or 2 edges and finally a min-cut of size $3$ may share 1,2 or 3 edges with the tree. In each of these cases, we our communication strategies ensure that at least one of the node will have the required information as per the aforementioned characterization lemmas.  We give a road-map  of this division in Figure \ref{figure:roadmap}.

Our results of finding min-cuts are split into two chapters. In Chapter \ref{chapter:4_trsf}, we present a new technique called the \emph{tree restricted semigroup function}. We will show that this is enough to optimally find the min-cuts of size $1$ and $2$. In Chapter \ref{chapter:5_min-cut-3}, we will give algorithm to find min-cut of size $3$. We introduce two techniques: \emph{sketching} and \emph{layered algorithm} to do the same.

%% file: thesis_subparts/chapter4/restricted_semigroup.tex
%TEX root = ../../thesis.tex
\label{chapter:restricted_semigroup_function}
In this chapter, we will define tree restricted semigroup function and demonstrate its utility in finding the induced-cuts of size 1 and 2. A tree restricted semigroup function is defined with respect to a tree $T$ and is based on a commutative semigroup. Recall that a commutative semigroup is a set $\mathcal X$ together with a commutative and associative binary operation $\boxplus$ which is a function $\boxplus: \mathcal{X} \times \mathcal{X} \rightarrow \mathcal{X}$. {Tree restricted semigroup function} is inspired by semigroup function defined in \cite[Def. 3.4.3]{peleg2000distributed}.

Let $\mathcal X$ be a commutative semigroup with operator $\boxplus$. Let $T$ be any spanning tree. We will formally define $f : V \rightarrow \mathcal X$ to be a tree restricted semigroup function with respect to the tree $T$ in Definition \ref{defn:restricted_semigroup_function}. 
For any node $v$, $f(v)$ will only depend on the vertex set $\desc[T]{v}$.
Let $a \in V$ be any node, for each $v \in \ancestor[T]{a}$, node $a$ computes the value $X_{a}^{v}$ through a pre-processing step which is the contribution by node $a$ for calculating  $f(v)$. 

\begin{defn}[Tree Restricted Semigroup function]
\label{defn:restricted_semigroup_function}
Let $\mathcal X$ be a commutative semigroup with the operator $\boxplus$. Then, the function $f : V \rightarrow \mathcal X$ is a tree restricted semigroup function if $f(v) = \bigboxplus\limits_{a \in \desc[T]{v}} X_{a}^{v} $
\end{defn}
Further for any $a \in V$ and $v \in \ancestor[T]{a}$, define $X_{\desc[T]{a}}^{v} = \bigboxplus\limits_{a' \in \desc[T]{a}} X_{a'}^{v}$. Note that $X_{\desc[T]{a}}^{v} \in \mathcal X$. We will say that $X_{\desc[T]{a}}^{v}$ is the contribution of nodes in the set $\desc[T]{a}$ for computing $f(v)$.
\begin{obs}
\label{obs:info_required_to_aggregate}
 Let $a$ be an internal node in the tree $T$ and $v \in \ancestor[T]{a}$ then $$X_{\desc[T]{a}}^{v} = X_{a}^{v} \boxplus \paren{\bigboxplus\limits_{c \in \child_T(a)} X_{\desc[T]{c}}^{v}}$$ \end{obs} 
\begin{obs} 
\label{obs:info_required_to_compute_trsf}
For any internal node $v$
$$f(v) = X_{v}^{v} \boxplus \paren{\bigboxplus\limits_{c \in \child_T(v)} X_{c^\downarrow}^{v}}$$ 
\end{obs}

\begin{algorithm}[h]
	
\DontPrintSemicolon
\SetKwFor{Forp}{for}{parallely}{endfor}
\SetKwFor{Forw}{for}{wait}{endfor}
\SetKwProg{phase}{Phase}{}{}
\SetKwProg{preprocessing}{Pre-Processing}{}{}
\SetKwProg{availableInfo}{Available Info:}{}{}
\preprocessing{}{For any node $a, \forall v \in \ancestor[T]{a}$, node $a$ knows $X_{a}^{v}$ through a pre-processing} 

\setcounter{AlgoLine}{0}
\phase{1 : Aggregation phase run on all node $a$, aggregates $X_{\desc[T]{a}}^{v}\ \forall v\in \ancestor{a}$}{

\lForw{rounds $t = 1$ to $Depth(T)-\level[T]{a}$}{}
$l \gets 0$\;
\For{rounds $t = Depth(T) - \level[T]{a} + 1$ to $h\paren{T}$}{
	$v \leftarrow$ ancestor of node $a$ at level $l$\;
	\lIf{$a$ is leaf node}{$X_{a^{\downarrow T}}^{v} \gets X_{a}^{v}$}
	\Else{
		\lForp{$c \in \child_T(a)$}{collect $\angularbrac{l,X_{c^{\downarrow T}}^{v}}$}
		$X_{a^{\downarrow T}}^{v}\gets X_{a}^{v} \boxplus \paren{ \bigboxplus\limits_{c \in \child_T(a)}X_{c^{\downarrow T}}^{v}}$\;
		} \label{interstep:cut-1-respects-tree-internal-node}
	send to the parent  $\angularbrac{l,X_{a^{\downarrow T}}^{v}}$\;
	$l \gets l+1$\;
}	
} \setcounter{AlgoLine}{0}
\phase{2: Computation Phase (run on all node $v \in V$), finds f(v)}{
\lavailableInfo{}{Each node $v$ knows $X_{c^{\downarrow T}}^{v}$ for all $c \in \child(v)$ }
\lIf{$v$ is a leaf node}{$f(v) \gets X_{v}^{v}$}
\Else{
	$f(v) \gets X_{v}^{v} \boxplus \paren{\bigboxplus\limits_{c \in \child_T(v)} X_{c^{\downarrow T}}^{v}}$
}
}
\caption{Tree Restricted Semigroup function}
\label{algo:trsf_phase-2}
\end{algorithm}
\begin{defn}
A tree restricted semigroup function $f(\cdot)$ is efficiently computable if for all $v$, $f(v)$ can be computed in $O(Depth(T))$ time in the \CONGEST model.
\end{defn}
In Lemma \ref{lemma:tree_restricted_semigroup_function},  we will give sufficient conditions for a tree restricted semigroup function to be efficiently computable.
\begin{lemma}
\label{lemma:tree_restricted_semigroup_function}
For all $a \in V$ and $v \in \ancestor[T]{a}$ if $X_a^v, X_{\desc[T]{a}}^v, f(v)$ are all of size $O(\log n)$ bits and if $X_a^v$ can be computed in $O(Depth(T))$ time by node $a$ then the tree restricted semigroup function $f(\cdot)$ is efficiently computable.
\end{lemma}
\begin{proof}
For any node $v$, tree restricted function $f(v)$ depends on the value $X_{a}^{v}$ for all $a\in \desc[T]{v}$, thus each such node $a$ \emph{convergecasts} (see Sec. \ref{sec:covergecast_broadcast}) the required information up the tree which is supported by aggregation of the values.  We will give an algorithmic proof for this lemma. The algorithm to compute a efficiently computable semigroup function $f(\cdot)$ is given in Algorithm \ref{algo:trsf_phase-2}. 

The aggregation phase of the algorithm given in Phase {1} runs for at most $Depth(T)$ time and facilitates a coordinated aggregation of the required values and convergecasts them in a synchronized fashion. Each node $a$ in Phase {2}, sends $\level[T]{a}$ messages of size $O(\log n)$  to its parent, each message include $X_{a^{\downarrow T}}^{v}$ where $v \in \ancestor[T]{a}$; which as defined earlier is the contribution of nodes in $a^{\downarrow T}$ to $f(v)$. This message passing takes $O(1)$ time since $X_{a^{\downarrow T}}^{v} \in \mathcal X$ is of size $O(\log n)$ bits. For brevity, we assume this takes exactly $1$ round, this enables us to talk about each round more appropriately as follows: Any node $a$ at level $\level[T]{a}$ waits for round $t = 1$ to $Depth(T) - \level[T]{a}$. For any $l \in [0, \level[T]{a}-1]$, in round $t = Depth(T) - \level[T]{a} + l + 1$  node  $a$ sends to its parent $\angularbrac{l, X_{a^{\downarrow T}}^{v}}$ where $v$ is the ancestor of $a$ at level $l$. When node $a$ is an internal node then as per Observation \ref{obs:info_required_to_aggregate}, $X_{\desc[T]{a}}^{v}$  depends on $X_{a}^{v}$ which can be pre-calculated in $O(Depth(T))$ rounds. Also, $X_{a^{\downarrow T}}^{v}$ depends on $X_{c^{\downarrow T}}^{v}$ for all $c \in \child_T(a)$ which are at level $\level[T]{a}+1$ and have send to $a$ (which is their parent) the message $\angularbrac{l, X^{v}_{c^{\downarrow T}}}$ in the $(Depth(T) - \level[T]{a} + l)^\text{th}$ round. For a leaf node $a$, $X_{a^{\downarrow T}}^{v} = X_{a}^{v}$ which again is covered in pre-processing step.

In Phase {2}, node $v$ computes tree restricted semigroup function $f(v)$. As per Observation \ref{obs:info_required_to_compute_trsf} of the algorithm each internal node $v$ requires $X_{\desc[T]{c}}^{v}\ \forall c \in \child_T(v)$ and $X_v^{v}$. $X_v^{v}$ is computed in the pre-processing step. And $X_{\desc[T]{c}}^{v}$ is received by $v$ in the aggregation phase. When node $v$ is a leaf node, $f(v)$ depends only on $X_{v}^{v}$.
\end{proof}
For any node $v$ and a tree $T$, let us define $\eta_T(v) = |\delta(\desc[T]{v})|$. In Subsection \ref{subsection:size-of-tree-cuts}, we will show that $\eta_T(\cdot)$ is an efficiently computable  {tree restricted semigroup function}. This will be enough to find if there exists a min-cut of size $1$ or a min-cut of size $2,3$ which 1-respects our fixed spanning tree. 

Further in Subsection \ref{subsection:induced_cuts_of_size_2}, we will find induced cuts of size 2 and define a tree restricted semigroup function $\zeta_T(\cdot)$ which will enable us to find induced-cuts of size 2. As mentioned earlier, we will work with a fixed BFS tree $\tree$. We will skip the redundant $\tree$ in the notations. For example $\eta(\cdot)$ instead of $\eta_T(\cdot)$ and $\zeta(\cdot)$ instead of $\zeta_T(\cdot)$.
%\input{subparts/restricted_semigroup_old.tex}

%% file: thesis_subparts/chapter4/tree-cuts.tex
%TEX root = ../../thesis.tex
\section{Min-Cuts of size 1}
\label{subsection:size-of-tree-cuts}
In this subsection we will prove that the function $\eta(\cdot)$ is an efficiently computable tree restricted semigroup function. Here we will work with our fixed BFS tree $\tree$.  Recall that the function $\eta: V \rightarrow [m]$ where $\eta(v) = |\delta(v^\downarrow)|$ and $\delta(v^\downarrow)$ is the cut induced by the vertex set $v^\downarrow$. 
For any node $a \in v^\downarrow$, we will use $H_{a}^{v} \triangleq |\delta(a) \cap \delta(\desc{v})|$. This is equal to the number of  vertices adjacent to $a$ which are not in vertex set $v^\downarrow$. The commutative semigroup associated here is the set of all positive integers $\mathbb Z^+$ with \lq addition\rq operator. Thus for a node $a\in \ancestor{v}$, $H_{\desc{a}}^{v} = \sum_{a' \in \desc{a}}H_{a'}^{v}$. We give the pre-processing steps in Algorithm \ref{algo:preprocessing-cut-1-respects-tree} which calculates $H_{a}^{v}$ for all $v \in \ancestor{a}$

\begin{algorithm}[h]
\DontPrintSemicolon
\SetKwFor{Forp}{for}{parallely}{endfor}
\lForp{$b$ adjacent to $a$}{send the ancestor set $\ancestor{a}$ to $b$} 
\lForp{$b$ adjacent to $a$}{receive the ancestor set $\ancestor{b}$} \label{line:pre-processing-eta-v-communication} 
\For{$v \in \ancestor{a}$}{
	$H_{a}^{v} \gets |\curly{b\mid (a,b)\in E,\ v\notin \ancestor{b}}|$\;
	\tcc{node $a$ can execute the above step locally because it knows $\ancestor{b}$, $\forall\ (a,b) \in E$}
}
\caption{Pre-Processing for computing the function $\eta(\cdot)$ (run at all node $a$ for finding $H^{v}_{a}\ \forall\ v \in \ancestor{a}$)}
\label{algo:preprocessing-cut-1-respects-tree}
\end{algorithm}
\begin{obs}
\label{prepo:pre-processing-size-tree-cuts}
Pre processing as given in Algorithm \ref{algo:preprocessing-cut-1-respects-tree} takes $O(D)$ time.
\end{obs}
\begin{proof}
Any node $a$ has at most $D$ ancestors in $\tree$. So for every node $b$ adjacent to $a$, it takes $O(D)$ time to communicate the set $\ancestor{a}$ to it. Similarly, it takes $O(D)$ time to receive the set $\ancestor{b}$ from node $b$. Now $H_{a}^{v}$ for any $v \in \ancestor{a}$ will just be an internal computation at node $a$. 
\end{proof}
\begin{lemma}
\label{lemma:eta_v_is_restricted_semigroup_function}
$\eta(\cdot)$ is an efficiently computable {tree restricted semigroup function}.
\end{lemma} 
\begin{proof}
Let $v \in V$, $\eta(v)$ is the number of edges going out of the vertex set $\desc{v}$. That is $\forall a \in \desc{v}$, $\eta(v)$ is the sum of the number of vertices incident to $a$ which are not in the set $\desc{v}$. Thus $\eta(v) = \sum_{a \in \desc{v}} H_a^{v}$. By Observation \ref{prepo:pre-processing-size-tree-cuts}, for all $v \in \ancestor{a}$, $H_{a}^{v}$ can be computed in $O(D)$ time. Also for any $v$, $\eta(v)$ could be as big as the number of edges $m$. Therefore, for any $a \in V$ and $v \in \ancestor{a}$ we have, $0 \leq H_a^{v} \leq H_{\desc{a}}^{v}\leq \eta(v) \leq m$. Thus $H_a^{v}$,$H_{\desc{a}}^{v}$ and $f(v)$ can be represented in $O(\log n)$ bits. Hence by Lemma \ref{lemma:tree_restricted_semigroup_function}, $\eta(\cdot)$ is an efficiently computable tree restricted semigroup function.
\end{proof}

To compute $\eta(\cdot)$, we use Algorithm \ref{algo:trsf_phase-2} given in Lemma \ref{lemma:tree_restricted_semigroup_function}. Further during the computation of $\eta(\cdot)$, each node $a$ computes $H_{\desc{a}}^{v}$ for all $v \in \ancestor{a}$ which is the aggregated value of all the decendents of node $a$. We summarize these results in the following Lemmas.
\begin{lemma}
The function $\eta(\cdot)$ can be computed in $O(D)$ time for every node $v\in V$.
\end{lemma}
\begin{proof}
In Lemma \ref{lemma:eta_v_is_restricted_semigroup_function}, we show that $\eta(\cdot)$ is an efficiently computable {tree restricted semigroup function} defined with respect to the BFS tree $\tree$. Also, $Depth(\tree) = O(D)$. Thus by Definition \ref{defn:restricted_semigroup_function}, $\eta(\cdot)$ can be computed in $O(D)$
\end{proof}
\begin{lemma}
\label{lemma:H_v_a_known_to_all_nodes}
For each node $a$ and $v \in \ancestor{a}$, node a knows $H_{\desc{a}}^{v} = |\delta\paren{\desc{a}} \cap \delta\paren{\desc{v}}|$ in $O(D)$ rounds.
\end{lemma}
\begin{proof}
During the computation of the tree restricted semigroup function $\eta(\cdot)$, each node $a$ for every $v \in \ancestor{a}$, computes the aggregated value $H_{\desc{a}}^{v} = \sum_{a' \in \desc{a}} H_{a}^{v}$ which is the contribution of nodes  in $\desc a$ towards the computation of $\eta(v)$. Thus the lemma follows.
\end{proof}
\begin{thm}\label{thm:min-cut-size-1}
Min-cut of size one (bridge edges)  can be found in $O(D)$ time.
\end{thm}
\begin{proof}
For any node $v\neq r$, the edge $(\parent{v},v) \in \delta(v^\downarrow)$. Thus when $\eta(v) = 1$, then $(\parent{v},v)$ is the only edge in $\delta(v^\downarrow)$ and is a cut edge. 
\end{proof}
\begin{lemma}
\label{lemma:induced-cut-one-respects}
If $\eta(v) = k$ then $\delta(v^\downarrow)$ is a cut-set of size $k$.
\end{lemma}
Having found $\eta(\cdot)$ for all the nodes we downcast them. That is each node $v$, downcasts its $\eta(v)$ value to the vertex set $\desc{v}$. This kind of broadcast is similar to \emph{Broadcast Type -1} defined in Chapter \ref{chpater:2_prelim_background}. Thus by Lemma \ref{lemma:broadcast_general}, this can be done in $O(D)$ time. That is all nodes $a \in \desc{v}$ will have the value of $\eta(v)$ in $O(D)$ time. This will be useful for the algorithm to find induced cuts of size 2 given in next subsection. We quantify the same in the following lemma.
\begin{lemma}
\label{lemma:eta_known_to_all_nodes}
For any node $v$, the nodes in the vertex set $\desc{v}$ know $\eta(v)$ in $O(D)$ rounds.
\end{lemma}

%% file: thesis_subparts/chapter4/connectivity.tex
%TEX root = ../../thesis.tex
\section{Min-Cuts of size 2}
\label{subsection:induced_cuts_of_size_2}
In this subsection, we will give an algorithm to find min-cuts of size 2. The theme here will be the use of tree restricted semigroup function. We will define a new tree restricted semigroup function $\zeta(\cdot)$ which will be based on a specially designed semigroup. 

Before we get into details of $\zeta(\cdot)$, we will give details about cuts of size $2$. Let $A \subset V$ and $|\delta(A)|=2$. The question here is to find $\delta(A)$. Here, $\delta(A)$ could share either one edge with the tree $T$, or it could share $2$ edges with the tree. When $\delta(A)$ shares one edge with the tree then it can be found in $O(D)$ rounds as described by Lemma \ref{lemma:induced-cut-one-respects}. We make the following observation for the case when $\delta(A)$,  2-respects the tree. In this case, $\delta(A) = \curly{(\parent{a},a), (\parent{w},w)}$ for some $a,w \in V\setminus r$ and $a \neq w$. 
\begin{obs}
\label{obs:induced-cut-size-2-2-respects-tree}
Let $A\subset V$ and $|\delta(A)|=2$. Also, let $\delta(A)$ 2-respect the tree $\tree$ then $\delta(A) = \delta(\desc{a}) \oplus \delta(\desc{w})$ for some $a\neq w$ and $a,w \in V\setminus r$. Further either ${\desc{a}} \subset {\desc{w}}$ (nested) or ${\desc{a}} \cap {\desc{w}} = \emptyset$ (mutually disjoint).
\end{obs}
\begin{proof}
Let $\delta(A) = \curly{(\parent a,a),(\parent w,w)}$ such that $a\neq w$ and $a,w \in V\setminus r$. WLOG let $\level a \geq \level w$. Because of the tree structure, either $\desc a \subset \desc w$ or $\desc a \cap \desc w = \emptyset$. 
When $\desc a \subset \desc w$ then the induced cut $\delta(A)$ is the cut set $(\desc w \setminus \desc a,V\setminus (\desc w \setminus \desc a)) = \delta(\desc w \setminus \desc a)$. 
Since $\desc a \subset \desc w$ thus $\desc w \setminus \desc a = \desc w \oplus \desc a$. 
Hence $\delta(\desc w \setminus \desc a) = \delta(\desc w \oplus \desc a)$ and finally by Corollary \ref{main-cor}, we have $\delta(\desc a \oplus \desc w) = \delta(\desc a) \oplus \delta(\desc w)$.
Similarly, when $\desc a \cap \desc w = \emptyset$ then the induced cut $\delta(A)$ is the cut set $(\desc a \cup \desc w, V\setminus \desc a \cup \desc w) = \delta(\desc a \cup \desc w) = \delta(\desc a \oplus \desc w) = \delta(\desc a) \oplus \delta(\desc w)$.
\end{proof}
% % % % % % % % % % % % % % % % % % % % % % % % % % % %
% % % % Begin: Characterizing Lemma about 2-cuts % % % %
% % % % % % % % % % % % % % % % % % % % % % % % % % % %
The above observation states that we have two different cases when an induced-cut of size $2$ shares both the edges with the tree. For any $A,B \subseteq V$, let $\gamma(A,B) = |\delta(A) \cap \delta(B)|$. 
In this subsection, we will prove the following lemma. This lemma will be enough to prove that the min-cuts of size $2$ can be deterministically found in $O(D)$ rounds. Moreover, it will also help us find min-cuts of size 3 as given in Chapter \ref{chapter:5_min-cut-3}.

\begin{lemma}
\label{lemma:induced-cut-2-cut-2-respects}
Let $a,w$ be two nodes. WLOG let $\level{a} \geq \level{w}$. If $\curly{(\parent a,a),(\parent w,w)}$ is a cut set induced by $\desc{a} \oplus \desc{w}$ and if $\gamma(\desc{a},\desc{w}) > 0$ then such an induced cut can be found in $O(D)$ rounds by node $a$.
\end{lemma}
Using the above lemma we now prove the following theorem.
\begin{thm}
Min-cuts of size $2$ can be found in $O(D)$ rounds.
\end{thm}
\begin{proof}
When there is a min-cut of size 2, it either 1-respects the tree or 2-respects the tree. When it 1-respects the tree, by Lemma \ref{lemma:induced-cut-one-respects}, we know that it can be found in $O(D)$ rounds because this only requires computation of $\eta(\cdot)$. 

When a min-cut of size 2, 2-respects the tree then by Lemma \ref{obs:induced-cut-size-2-2-respects-tree} we know that the cut is of the form $\delta(\desc{w}) \oplus \delta(\desc{a})$ for some nodes $a$ and $w$. Moreover $\gamma(\desc{a},\desc{w}) > 0$ since there is not cut of size $1$. From Lemma \ref{lemma:induced-cut-2-cut-2-respects} we know that this can be found in $O(D)$ rounds.
\end{proof}

We will now prove Lemma \ref{lemma:induced-cut-2-cut-2-respects}. 
When the induced cut of size 2, 2-respects the tree, we know by Observation \ref{obs:induced-cut-size-2-2-respects-tree} that there could be two different cases. 
For the easy case when ${\desc{a}} \subset {\desc{w}}$, we know from Lemma \ref{lemma:H_v_a_known_to_all_nodes} that $H_{\desc{a}}^{w}=|\delta(\desc{a}) \cap \delta(\desc{w})| = \gamma({\desc{a}},{\desc{w}})$ is known by node $a$ in $O(D)$ rounds.
Also from Lemma \ref{lemma:eta_known_to_all_nodes}, $\eta(w)$ is known by all the vertices $a \in {\desc{w}}$. Hence, if $\curly{(\parent w,w),(\parent a,a)}$ is an induced cut and $a \in \desc w$ and $a \neq w$, then as per Lemma \ref{lemma:2-respect-2-cut} node $a$ has the required information to find the induced cut. The following lemma summarizes this.

\begin{lemma}
Let $a,w$ be two vertices such that $\desc{a} \subset \desc{w}$. Let \sloppy $\delta(\desc{a} \oplus \desc{w}) =\curly{(\parent a,a),(\parent w,w)}$ be an induced cut, then node $a$ can find such a cut in $O(D)$ rounds.
\label{lemma:2-cut-2-respect-nested}
\end{lemma}

The other case is when $\desc{a}$ and $\desc{w}$ are disjoint sets. 
This is the non-trivial part for finding an induced cut of size 2. 
Recall from Lemma \ref{lemma:2-respect-2-cut}, that to make a decision about an induced cut of the form $\curly{(\parent a,a),(\parent w,w)}$ we require one of the node in the network to know $\gamma(\desc{a},\desc{w}), \eta(a)$ and $\eta(w)$. The idea here is for node $a$ (when $\level{a}\geq \level{w}$) to find $\eta(w)$ and $\gamma(\desc{a},\desc{w})$ which is quite a challenge because there does not exists a straightforward way through broadcast or convergecast. Moreover, there may not even exist an edge between the vertices $a$ and $w$. To deal with this, we introduce a new tree restricted semigroup function $\zeta(\cdot)$ 

The tree restricted semigroup function $\zeta(\cdot)$ is based on a specially defined semigroup $\mathcal Z$.
Before we give further details and intuition about the function $\zeta(\cdot)$, let us define the semigroup $\mathcal Z$. There are two special elements $\Zidentity$ and $\Zabsorb$ in the semigroup $\mathcal Z$. Apart from these special elements all other $Z \in \mathcal Z$ are a four tuple. Let the four-tuple be given as $Z = \angularbrac{Z[1],Z[2],Z[3],Z[4]}$, here $Z[1],Z[2] \in V$ and $Z[3], Z[4] \in \mathbb{Z}^+$.

The operator associated with the semigroup $\mathcal Z$ is $\odot$. Special elements $\Zidentity$ and $\Zabsorb$ are the identity and the absorbing element with respect to the operator $\odot$ of the semigroup $\mathcal Z$. These special elements are defined in such a way that for any $Z \in \mathcal Z$, $Z \odot \Zidentity = Z = \Zidentity \odot Z$ and $Z\ \odot\Zabsorb\ =\  \Zabsorb\ =\ \Zabsorb \odot Z$. (Here the symbols $\mathbb{0},\mathbb{1}$ are not chosen as identity and absorbing elements, because $\Zidentity$ corresponds to zero edges and $\Zabsorb$ is considered as a zero of the semigroup which we will see in Property \ref{property:Z_a_v}). In Algorithm \ref{algo:zeta_operator}, we define the operator $\odot$.
% % % % % % % % % % % % % % % % % % 
% % % % Begin: ZETA OPERATOR % % % %
% % % % Operator, Supporting Lemma % %
% % % % % % % % % % % % % % % % % % 
\begin{algorithm}[h]
\DontPrintSemicolon
\tcp{for $i \in \curly{1,2,3,4}$ let $Z_1[i]$ be the $i^\text{th}$ element in 4-tuple of $Z_1$. Similarly for $Z_2$}
\lIf{one of $Z_1,Z_2$ is $\Zabsorb$}{return $\Zabsorb$}
\lElseIf{$Z_1 = \Zidentity$}{return $Z_2$}
\lElseIf{$Z_2 = \Zidentity$}{return $Z_1$}
\lElseIf{$Z_1[1:3] = Z_2[1:3]$}{return $\angularbrac{Z_1[1],Z_1[2],Z_1[3],Z_1[4] + Z_2[4]}$}
\lElse{return $\Zabsorb$}
\caption{$Z_1 \odot Z_2 $ (Both $Z_1,Z_2 \in \mathcal{Z}$) }
\label{algo:zeta_operator}
\end{algorithm}
\begin{obs}
\label{proposition:zeta_operator}
$\odot\ (\zetaoperator)$ is commutative and associative.
\begin{proof}
Commutativity of $\odot$ is straightforward and is implied by the commutativity of addition over $\mathbb Z^+$. For associativity, let's imagine there exist $Z_1,Z_2,Z_3 \in \mathcal Z$. Further let's imagine that $Z_1,Z_2,Z_3 \notin \curly{\Zabsorb,\Zidentity}$. Now when the first two elements of each of the four tuples $Z_1,Z_2,Z_3$ is same, then associativity is trivial and implied by the associativity of the addition operator over $\mathcal Z^+$. When the first two elements are not equal in $Z_1,Z_2,Z_3$ then $Z_1 \odot (Z_2 \odot Z_3) = (Z_1 \odot Z_2) \odot Z_3 =\ \Zabsorb$. Similarly, If one of $Z_1,Z_2,Z_3$ is $\Zabsorb$, then also $Z_1 \odot (Z_2 \odot Z_3) = (Z_1 \odot Z_2) \odot Z_3 =\ \Zabsorb$, hence associativity is implied. And lastly when one of $Z_1,Z_2,Z_3$ is $\Zidentity$, then $Z_1 \odot (Z_2 \odot Z_3)$ becomes an operation between just two elements and associativity is implied directly.
\end{proof}
\end{obs}
For any node $v$, define the edge set $E\curly{\zeta(v)} \triangleq \delta(\desc{v})\setminus (\pi(v),v)$. The value of $\zeta(v)$ depends on the edge set $E\curly{\zeta(v)}$. For any node $a$ and $v \in \ancestor{a}$, let $E\curly{Z_{a}^{v}} \triangleq E\curly{\zeta(v)} \cap \delta(a)$. We define $Z_a^v$ as per the Property \ref{property:Z_a_v}. 
% % % % % %  % % % % % % % % % % % % % % % % % %
 % %% % % % Begin:Zeta_a^v Operator% % % % % % %
% % % %Z_a^v property, Algorithm, Lemma% % % % %
% %  % % % % % % % % % % % % % % % % % % % % % %
\begin{property}
 \label{property:Z_a_v}
 For any node $a$ and $v \in \ancestor{a}$, $Z_{a}^{v}$ takes one of the following three values:-
 \begin{enumerate}[i.]
 \item $\Zidentity$ when $E\curly{Z_{a}^{v}} = \emptyset$ \label{bullet:Z_a_v_empty}
 \item ${\angularbrac{w,\parent{w},\eta(w),\gamma(a,{\desc{w}})}}$ when there exists a node $w$ at level $\level{v}$ such that all the edges in $E\curly{Z_{a}^{v}}$ have one endpoint in ${\desc{w}}$ \label{bullet:Z_a_v_when_w_exists}
 \item $\Zabsorb$ otherwise \label{bullet:Z_a_v_otherwise}
 \end{enumerate}
\end{property}

Similar to the previous subsection the semigroup function $\zeta(\cdot)$ is defined as $\zeta(v) \triangleq \bigodot_{a\in \desc{v}} Z_a^v$. 
We will say that $Z_{{\desc{a}}}^{v}$ is the contribution of nodes in vertex set $\desc{a}$ to compute $\zeta(v)$. 
 We will now prove that for all node $a \in V$ and $v \in \ancestor{a}$, $Z_a^v$ can be computed in $O(D)$ time in Lemma  \ref{lemma:Z_a_v}, then using Definition \ref{defn:restricted_semigroup_function} and Lemma \ref{lemma:tree_restricted_semigroup_function}, we will prove that $\zeta(\cdot)$ is an efficiently computable tree restricted semigroup function.

For any $l < \level{a}$, let $\alpha(a,l)$ be the ancestor of node $a$ at level $l$. For notational convenience let $\alpha(a,\level{a}) = a$. 
\begin{lemma}
 \label{lemma:Z_a_v}
For all node $a \in V$ and $v \in \ancestor{a}$, $Z_{a}^{v}$ can be computed in $O(D)$ time as defined in Property \ref{property:Z_a_v}. 
 \end{lemma}
  % % % % % % % % % % % % % % % % % % % % %
\begin{proof}
Here we will give an algorithmic proof, We describe the steps in Algorithm \ref{algo:preprocessing-cut-2-respects-tree}. Further we will prove that this algorithm correctly finds $Z_a^v$ and takes $O(D)$ time.

\begin{algorithm}[h]
 \DontPrintSemicolon
 \SetKwFor{Forp}{for}{parallely}{endfor}
 $\mathcal N(a) \gets \curly{b\mid (a,b) \in E, (a,b)\ \text{is a non-tree edge}} $ \label{line:pre-processing-N-a}  \;
 \lForp{all $b \in \mathcal N(a)$}{send tuples $\angularbrac{\level{u},\eta(u),u}$ for all $u \in \ancestor{a}$ to $b$}  \label{line:pre-processing-zeta-v-communication-1} 
 \lForp{all $b \in \mathcal N(a)$}{receive tuples $\angularbrac{\level{u},\eta(u),u}$ for all $u \in \ancestor{b}$}  \label{line:pre-processing-zeta-v-communication-2} 
 $l_{\min} \gets \min(\curly{\level{b}\ |\ b\in \mathcal N(a)} \cup \ell(a))$\;
 \If{$l_{\min} \neq \ell(a)$}{
 	\lFor{$v \in \curly{\alpha(a,l) \mid l \in (l_{\min},\ell(a)] }$}{$Z_{a}^{v} =\ \Zabsorb$}
 }
 \For{$l = l_{\min}$ to $1$}{
 	$v \gets \alpha(a,l)$ \;
    \tcp{$A^l$ is set of ancestors at level $l$ of nodes in $\mathcal N(a)$ except $v$} 
 	$A^l \gets \curly{\alpha(b,l) \mid b \in \mathcal N(a)} \setminus v$ \;
 	\lIf{$A^l = \emptyset$}{$Z^{v}_{a} \gets \Zidentity$} \label{line:pre-processing-Z_a^v=empty}
 	\ElseIf{$|A^l| = 1$}{ 
 		$w \gets$ element in singleton $A^l$\;
 		$\gamma({\desc{w}},a) \gets |\curly{b\mid b \in \mathcal N(a), \alpha(b,l) = w}|$\;
 		$Z^{v}_{a} \gets \angularbrac{w,\parent{w},\eta(w),\gamma({\desc{w}},a)}$ 
 	}\lElse{$Z^{v}_{a} \gets \Zabsorb$}
 }
 \caption{Pre-Processing step for $\zeta$ (run at all node $a$ for finding $Z_{a}^{v}$ for all $v \in \ancestor{a}$)}
 \label{algo:preprocessing-cut-2-respects-tree}
\end{algorithm}

As per Property \ref{property:Z_a_v},  $Z_{a}^{v}$ just depend on the nodes adjacent to the node $a$. The required information is received and sent from the adjacent nodes in Algorithm \ref{algo:preprocessing-cut-2-respects-tree} at line \ref{line:pre-processing-zeta-v-communication-1} and \ref{line:pre-processing-zeta-v-communication-2}. 
Each of these takes only $O(D)$ time because at max $O(D)$ messages of $O(\log n)$ bits are communicated. 
$\mathcal {N}(a)$ is the non-tree neighbors of node $a$ found in line \ref{line:pre-processing-N-a}.
In Algorithm \ref{algo:preprocessing-cut-2-respects-tree}, decision in regard to $Z_{a}^{v}$ is taken based on $A^{\ell(v)}$ which is the set of ancestors at level $\level{v}$ of nodes in $\mathcal N(a)$ except the node $v$. 

When $A^{\ell(v)} = \emptyset$, then in line \ref{line:pre-processing-Z_a^v=empty}, $Z_a^v$ is set to $\Zidentity$. For bullet (\ref{bullet:Z_a_v_empty}) of Property \ref{property:Z_a_v}, we need to prove that $E\curly{Z_a^v} = \emptyset \implies A^{\ell(v)} = \emptyset $. Here $E\curly{Z_a^v} = \delta(\desc v)\setminus (\parent v,v) \cup \delta(a)$. When $a \neq v$, then $E\curly{Z_a^v}$ is the set of all the edges which are incident on $a$ and goes out of the vertex set $\desc v$. Now when $E\curly{Z_a^v} = \emptyset$, then all edges incident on $a$ have the other endpoint in the vertex set $\desc v$. Thus $A^{\ell(v)} = \emptyset$. When $a = v$, here $E\curly{Z_v^v}$ is the set of all edges other then $(\parent v,v)$ which are incident on $v$ and goes out of the vertex set $\desc v$. Note that $(\parent v,v)$ is a tree-edge thus $\parent v \notin \mathcal N(a)$. Thus $E\curly{Z_v^v} = \emptyset \implies A^{\ell(v)} = \emptyset$.

Now for correctness, if $E\curly{Z_a^v} = \emptyset$ and $v \neq a$ then no edge incident on node $a$ goes out of $\desc v$, hence the set $A^{\ell(v)} = \emptyset$.  When $v = a$, we know that $\mathcal N(v)$ (in line \ref{line:pre-processing-N-a}) contains only non-tree neighbors, thus $\parent{v} \notin \mathcal N(v)$. We know that edge $(\parent v,v) \notin E\curly{Z_v^v}$. Thus here when $E\curly{Z_v^v} = \emptyset$ then no edge other than $(\parent v,v)$ incident on $v$ goes out of $\desc v$, thus $A^{\ell(v)} = \emptyset$. Hence in both the cases $Z_a^v$ is correctly set to $\Zidentity$. (in line \ref{line:pre-processing-Z_a^v=empty})

For the other two bullet's in Property \ref{property:Z_a_v}, we employ the same idea. If $A^{\level{v}} = \curly{w}$ (for some node $w$) then it implies that all the neighbours adjacent to node $a$ have a common ancestor $w$ other than $v$ and thus $Z_{a}^{v}$ captures the required information about node $w$. And when $|A^{\level{v}}|>1$  it simply means there are more than one node at level $\level{v}$ as given in bullet $(iii)$ of Property \ref{property:Z_a_v}. 
\end{proof}

Similar to previous section, $Z_{\desc{a}}^{{v}}\triangleq \bigodot_{a'\in \desc{a}} Z_{a'}^v$. 
Since for any $a' \in \desc{a}$, $Z_{{{a'}}}^{v}$ depends on $E\curly{\zeta(v)} \cap \delta({a'})$  thus  $Z_{{\desc{a}}}^{v}$ depends on  $E\curly{Z_{{\desc{a}}}^{v}} = \bigcup_{a' \in \desc{a}} E\curly{Z_{{{a'}}}^{v}} = E\curly{\zeta(v)} \cap \delta(\desc{a})$. 
The following lemma about $Z_{\desc{a}}^{v}$ is now implicit and immediately follows from Property \ref{property:Z_a_v}. 
% % % % % %  % % % % % % % % % % % % % % % % % %
 % %% % % % End:Zeta_a^v Operator% % % % % % %
% % % %Z_a^v property, Algorithm, Lemma% % % % %
% %  % % % % % % % % % % % % % % % % % % % % % %

% % % % % % % % % % % % % % % % % % % % % % % 
% % % % %Begin:Z_{\desc{a}}^{v}% % %  % % % %
% % % % %Charactrization, ComputationTime % % 
% % % % % % % % % % % % % % % % % % % % % % % 
\begin{lemma}
For any node $a$ and $v \in \ancestor{a}$, $Z_{\desc{a}}^{v}$ depends on the edge set $E\curly{Z_{\desc{a}}^v} = \bigcup\limits_{a' \in \desc{a}} E\curly{Z_{a'}^v}$ and takes one of the following values
\label{lemma:Z_desc_a_v}
\begin{enumerate}[a)]
 \item $\Zidentity$ when $E\curly{Z_{\desc a}^{v}} = \emptyset$ 
 \item ${\angularbrac{w,\parent{w},\eta(w),\gamma(\desc a,{\desc{w}})}}$ when there exists a node $w$ at level $\level{v}$ such that all the edges in $E\curly{Z_{\desc a}^{v}}$ have one endpoint in ${\desc{w}}$
 \item $\Zabsorb$ otherwise
\end{enumerate}
\end{lemma}
Basically, $Z_{\desc{a}}^{v}$ captures if there exists some node $w$ at level $\level{v}$ such that all the edges which are in $\delta(\desc{a})$ and go out of the vertex set $\desc{v}$, have the other end point in the vertex set $\desc w$.
\begin{lemma}
$\zeta(\cdot)$ is an efficiently computable semigroup function.
\end{lemma}
\begin{proof}
In Lemma \ref{lemma:Z_a_v}, it was shown that there exists a $O(D)$ time algorithm to compute $Z_{a}^{v}$ for any node $a$ and $v \in \ancestor{a}$. Further each such $Z_{a}^{v}$ is either a special symbol among $\Zidentity,\Zabsorb$ or a four tuple. It is easy to see that this four tuple is of $O(\log n)$ bits because the first two elements in the four tuple are node ids and the last two are integers which cannot exceed the number of edges and we require $O(\log n)$ bits to represent them. Now invoking Lemma \ref{lemma:tree_restricted_semigroup_function}, we know that $\zeta(\cdot)$ is an efficiently computable tree restricted semigroup function.
\end{proof}
\begin{lemma}
	For all nodes $a$ and $v \in \ancestor{a}$, $Z_{\desc{a}}^{v}$ can be computed in $O(D)$ time.
	\label{lemma:Z_desc_a_v_time}
\end{lemma}
\begin{proof}
Since $\zeta(\cdot)$ is an efficiently computable semigroup function and can be computed in $O(D)$ rounds. Also, during the computation of $\zeta(\cdot)$, for all node $a$ and $v \in \ancestor{a}$, $Z_{\desc a}^{v}$ is  also computed.
\end{proof}
\begin{lemma}
Let $a,w$ be two vertices such that $\desc{w} \cap \desc{a} = \emptyset$. WLOG, let $\level{a} \geq \level{w}$. Let $\curly{(\parent a,a),(\parent w,w)}$ be an induced cut such that $\gamma(\desc w, \desc a) > 0$, then node $a$ can find it in $O(D)$ rounds.
\label{lemma:2-cut-2-respect-disjoint}
\end{lemma}
\begin{proof}
	Here we have to prove that node $a$ will have access to $\eta(w)$ and $\gamma(\desc w,\desc a)$ if such a min cut occurs. Then confirming the min-cut is easy by Lemma \ref{lemma:2-respect-2-cut}. Let $v$ be the ancestor of node $a$ at level $\level{w}$.  By Observation \ref{lemma:Z_desc_a_v}, we know that $Z_{\desc{a}}^{v} = \angularbrac{w,\parent{w},\eta(w),\gamma(\desc{w},\desc{a})}$ in this case. Thus, the required information will be available at node $a$.
\end{proof}
% % % % % % % % % % % % % % % % % % % % % % % 
% % % % %End:Z_{\desc{a}}^{v}% % %  % % % %
% % % % %Charactrization, ComputationTime % % 
% % % % % % % % % % % % % % % % % % % % % % % 

% % % % % % % % % % % % % % % % % % % % % % % 
% % % % %Begin:Induced Cut of size 2  % % % %
% % % % %Algorithm, Proof % % % % % % % % % % 
% % % % % % % % % % % % % % % % % % % % % % % 
\begin{proof}[Proof of Lemma \ref{lemma:induced-cut-2-cut-2-respects}]
When the induced cut 2-respects the tree that is it is a symmetric difference of two tree cuts $\delta({\desc{a}})$ and $\delta({\desc{w}})$ for some $a,w\in V\setminus r$. As per Observation \ref{obs:induced-cut-size-2-2-respects-tree} here two cases could occur. The nested case when ${\desc{a}} \subset {\desc{w}}$ and the mutually disjoint case when ${\desc{a}} \cap {\desc{w}}$, results regarding them are given in Lemma \ref{lemma:2-cut-2-respect-nested} and Lemma \ref{lemma:2-cut-2-respect-disjoint}. In line \ref{line:algo-2-cut-nested-case} and \ref{line:algo-2-cut-mutually-disjoint} of Algorithm \ref{algo:2-cut algorithm}, we give details about the actual search. 

\begin{algorithm}[H]
	\DontPrintSemicolon
	\SetKwProg{availableInfo}{Available Info:}{}{}
	\lavailableInfo{}{
	Each node $a$, $\forall v \in \ancestor{a}$ knows $\eta(v),H_{\desc{a}}^{v}$ and $Z_{\desc{a}}^{v}$ (Lemma \ref{lemma:eta_known_to_all_nodes},\ref{lemma:H_v_a_known_to_all_nodes},\ref{lemma:Z_desc_a_v_time})		
	}
	\setcounter{AlgoLine}{0}
	\For{$l = 1$ to $\ell(a)$}{
		$v \gets$ $ \alpha(a,l)$  \tcp{ancestor of node $a$ at level $l$}
		\If{$Z_{{\desc{a}}}^{v} \notin \curly{\Zidentity,\Zabsorb}$}{
			Let $Z_{{\desc{a}}}^{v} = \angularbrac{w,\parent{w},\eta(w),\gamma({\desc{w}},{\desc{a}})}$ \;
			\If{$\eta(a) - \gamma({\desc{w}},{\desc{a}}) = 1$ \& $\eta(w) - \gamma({\desc{w}},{\desc{a}}) = 1$}{\label{line:algo-2-cut-nested-case}	
				$\curly{(\parent{w},w),(\parent{a},a)}$ is a 2-cut
			}
		}
		\If{$v \neq a$ \& $\eta(v) - H_{{\desc{a}}}^{v} = 1$ \& $\eta(a) - H_{{\desc{a}}}^{v} = 1$}{\label{line:algo-2-cut-mutually-disjoint}
			$\curly{(\parent{v},v),(\parent{a},a)}$ is a 2-cut
		}		
	}
	\caption{Algorithm to find 2-cut for node $a$}
	\label{algo:2-cut algorithm}
\end{algorithm}
\end{proof}
% % % % % % % % % % % % % % % % % % % % % % % 
% % % % % % %End:Induced Cut of size 2% % % %
% % % % % Algorithm, Proof% % % % % % % % % % 
% % % % % % % % % % % % % % % % % % % % % % % 
In this chapter, we formally defined tree restricted semigroup function. Further, we introduced two different types of tree restricted semigroup function $\eta(\cdot)$ and $\zeta(\cdot)$. We showed that these are enough to find min-cuts of size 1 and 2. In the next chapter, we will give algorithms to find min-cuts of size $3$.

%% file: thesis_subparts/chapter5/three_cut.tex
%!TEX root = ../../thesis.tex
In this chapter, we will give an algorithm to find a min-cut of size three. The idea here is to use Lemma \ref{lemma:gen_3_cut_2_respect} and Lemma \ref{lemma:gen_3_cut_3_respect} given in Chapter \ref{chpater:3_technical_overview} which characterizes the min-cut of size $3$. Having laid down these characterization lemmas, the critical aspect which remains to find the min-cut of size 3 (if it exists) is to communicate the required quantities by the characterizing lemmas to at least one node in the network whenever a min-cut of the kind occurs. For this chapter, we will assume that a min-cut of size $1,2$ do not occur.

Recall that we have fixed a BFS tree $\tree$ in the beginning. If there exists a min-cut in the network, there could be 7 different cases. These cases are different to each other based on the relation of the min-cut to the fixed BFS tree $\tree$. We enumerate these cases in Lemma \ref{lemma:min-cut-size-3}.  Further, give an algorithmic outline about how these cases can be found in Table \ref{table:overview_3_cut}.
\begin{lemma}
	If there exists a min-cut of size 3 then the following cases may arise for some $v_1,v_2,v_3 \in V\setminus r$ and $e,f$ as non-tree edge.
	\begin{enumerate}[\textbf{CASE}-1]
		\item $\curly{(\parent{v_1},v_1),e,f}$ is a min-cut
		\item $\curly{(\parent{v_1},v_1),(\parent{v_2},v_2),e}$ is a min-cut such that $\desc{v_2} \subset \desc{v_1}$
		\item $\curly{(\parent{v_1},v_1),(\parent{v_2},v_2),e}$ is a min-cut such that $\desc{v_2} \cap \desc{v_1} = \emptyset$
		\item $\curly{(\parent{v_1},v_1),(\parent{v_2},v_2),(\parent{v_3},v_3)}$ is a min-cut and  $\desc{v_3} \subset \desc{v_2} \subset \desc{v_1}$
		\item $\curly{(\parent{v_1},v_1),(\parent{v_2},v_2),(\parent{v_3},v_3)}$ is a min-cut $\desc{v_2} \subset \desc{v_1}$ and $\desc{v_3} \subset \desc{v_1}$ and $\desc{v_2} \cap \desc{v_3} = \emptyset$
		\item $\curly{(\parent{v_1},v_1),(\parent{v_2},v_2),(\parent{v_3},v_3)}$ is a min-cut and $\desc{v_3},\desc{v_2}$ and $\desc{v_1}$ are pairwise mutually disjoint
		\item $\curly{(\parent{v_1},v_1),(\parent{v_2},v_2),(\parent{v_3},v_3)}$ is a min-cut and $\desc{v_3} \subset \desc{v_2}$, $\desc{v_1} \cap \desc{v_2} = \emptyset$ and $\desc{v_1} \cap \desc{v_3} = \emptyset$
	\end{enumerate}
	\label{lemma:min-cut-size-3}
\end{lemma}
\begin{proof}
	From Observation \ref{obs:fundamental_observation_about_cut_respecting_trees}, we know that that if there exists a min-cut of size $3$ then it shares either $1,2$ or $3$ edges with the tree $\tree$. When a min-cut of size $3$ shares one edge with the tree then \CASE{1} applies. Here two edges are non-tree edges.
	
	When a min-cut of size $3$ shares 2-edges with the tree $\tree$ then there exists two tree edges. For some node $v_1,v_2 \in V\setminus r$ let these edges be $(\parent{v_1},v_1),(\parent{v_2},v_2)$. WLOG let $\level{v_1} \leq \level{v_2}$. Similar to  Observation \ref{obs:induced-cut-size-2-2-respects-tree}, there could be two cases here either $v_2 \subset v_1$ or $v_1 \cap v_2 = \emptyset$. Both of which are described in \CASE{2} and \CASE{3} respectively.
	
	The non-trivial part here is when the min-cut of size 3 shares all 3 edges with the tree. Here we will have 4 different cases. Let these cut edges be $(\parent{v_1},v_1),(\parent{v_2},v_2)$ and $(\parent{v_3},v_3)$. Also in the beginning let $\level{v_1} < \level{v_2} < \level{v_3}$. We start with the case when $\desc{v_3} \subset \desc{v_2} \subset \desc{v_1}$ which is described in \CASE{4}. Now when $\desc{v_3} \not\subset \desc{v_2}$ then since $v_2,v_3$ are different thus we have $\desc{v_2} \cap \desc{v_3} = \emptyset$ which is \CASE{5}. Notice that in both \CASE{4} and \CASE{5} we have $(\desc{v_2} \cup \desc{v_3}) \subset \desc{v_1}$. Now we move to the different scenario where $(\desc{v_2} \cup \desc{v_3}) \not\subset \desc{v_2}$. Here also we may have two cases when $\desc{v_2} \cap \desc{v_3} = \emptyset$ then we get \CASE{6} and when $\desc{v_2} \subset \desc{v_3}$ then we get \CASE{7}. Note that for \CASE{6} and \CASE{7} we may not require $\level{v_1} \leq \level{v_2}$ and $\level{v_1} \leq \level{v_3}$
\end{proof}
\input{thesis_subparts/chapter5/overview_table.tex}

The different cases as mentioned in Lemma \ref{lemma:min-cut-size-3} are pictorially shown in Figure \ref{fig:3-cut-cases}.
\input{thesis_subparts/chapter5/overview_figure.tex}
Among these cases, \CASE{1} is simple. Here for some node $v_1$ the induced cut is $(V\setminus \desc{v_1},\desc{v_1})$. Here we just need to find the size of tree cut $\eta(v_1)$ and if there exists such a min cut of size $3$, then $\eta(v_1) = 3$. This takes $O(D)$ time as shown in Lemma \ref{lemma:induced-cut-one-respects}.

Among the other cases only \CASE{2} and \CASE{4} have a simple broadcast based algorithm which is enough to let at least one of the node know the required quantities as per Lemma \ref{lemma:gen_3_cut_2_respect} and Lemma \ref{lemma:gen_3_cut_3_respect}. We give the details in the following lemmas.
\begin{lemma}
	When there exists a min-cut of size 3 as given in \CASE{2} (for some nodes $v_1,v_2 \in V\setminus r$ and a non-tree edge $e$ $\curly{(\parent{v_1},v_1),(\parent{v_2},v_2),e}$ is a min-cut such that $\desc{v_2} \subset \desc{v_1}$), then the cut edges can be found in $O(D)$ time. \label{lemma:3-cut-case-2}
\end{lemma}
\begin{proof}
	From Chapter \ref{chapter:4_trsf}, we know that each node $v$ knows $\eta(v) = |\delta(\desc{v})|$.
	Also for each $u \in \ancestor{v}$, $v$ knows $\eta(u)$ and $H^u_{\desc{v}}$ by Lemma \ref{lemma:eta_known_to_all_nodes} and \ref{lemma:H_v_a_known_to_all_nodes}. Thus if there exists an induced cut as in $\CASE{2}$ then we just need to make a simple evaluation based on Lemma \ref{lemma:gen_3_cut_2_respect}. We give the details of the evaluation in Algorithm \ref{Algo:3-cut-case-2}.
	\begin{algorithm}
		\DontPrintSemicolon
		\SetKwProg{availableInfo}{Available Info:}{}{}
		\lavailableInfo{}{
			Each node $x$, $\forall v \in \ancestor{x}$ knows $\eta(v),H_{\desc{x}}^{v}$ (Lemma \ref{lemma:eta_known_to_all_nodes} ,\ref{lemma:H_v_a_known_to_all_nodes})        
		}    \setcounter{AlgoLine}{0}
		\For{$v \in \ancestor{x}\setminus \curly{x,r}$}{
			\If{$\eta(v)-1 = H_{\desc{x}}^{v} = \eta(x) - 2$ OR  $\eta(v)-2 = H_{\desc{x}}^{v} = \eta(x) - 1$}{
				$\delta(\desc{v} \oplus \desc{x})$ is an induced cut of size 3\;
			}
		}
		\caption{Algorithm to find an induced cut of size $3$ as given in \CASE{2} (for some nodes $v_1,v_2 \in V\setminus r$ and a non-tree edge $e$ $\curly{(\parent{v_1},v_1),(\parent{v_2},v_2),e}$ is a min-cut such that $\desc{v_2} \subset \desc{v_1}$) run on all node $x \in V\setminus r$}
		\label{Algo:3-cut-case-2}
	\end{algorithm}    
\end{proof}

\begin{lemma}
	When there exists a min-cut of size 3 as given in \CASE{4} (For some $v_1,v_2,v_3 \in V\setminus r$, $\curly{(\parent{v_1},v_1),(\parent{v_2},v_2),(\parent{v_3},v_3)}$ is a min-cut and  $\desc{v_3} \subset \desc{v_2} \subset \desc{v_1}$) then the cut edges can be found in $O(D^2)$ time. \label{lemma:3-cut-case-4}
\end{lemma}
\begin{proof}
	
	By Lemma \ref{lemma:H_v_a_known_to_all_nodes}, each node $a$ knows $H_{\desc{a}}^{v} = \gamma(\desc{a},\desc{v})$ for all $v \in \ancestor{a}$. 
	Further, for any node $a \in V\setminus r$ we want to make sure that for any two nodes $v,u \in \ancestor{v}\setminus \curly{a,r}$, node $a$ knows 
	$H_{\desc{u}}^{v}$ if $\level{u} > \level {v}$. 
	For this, each node $x$ at level $i$ has $i-1$ such quantities to broadcast to its descendants in $\desc{x}$. This is similar to \emph{Broadcast Type-2} and takes $O(D^2)$ rounds. After this step every node $x$ knows $H_{\desc{y}}^{z}$ for all $y,z\in \ancestor{x} \setminus \curly{r,x}$ and $\level{y} > \level{z}$. 
	Now at each node $x$, to determine if $\curly{(\parent{x},x),(\parent{y},y),(\parent{z},z)}$ is a min-cut, node $x$, checks if $\eta(x) -1 = H_{\desc{x}}^{y} + H_{\desc{x}}^{z}$ and $\eta(y) -1 = H_{\desc{y}}^{z} + H_{\desc{x}}^{y}$ and $\eta(z) -1 = H_{\desc{x}}^{z} + H_{\desc{y}}^{z}$    
\end{proof}

For other cases the problem boils down to efficiently computing the required quantities as per Lemma \ref{lemma:gen_3_cut_2_respect} and \ref{lemma:gen_3_cut_3_respect} and communicating them to at least one node in the network, then this node can make the required decision about the min-cut of size 3. Unfortunately, simple broadcast and convergecast techniques do not seem plausible for the cases which are left. This is because of the arbitrary placements of the nodes in the tree.

In the remaining part of this chapter, we introduce two new techniques which will take care of this. In Section \ref{subsec:sketching_method}, we give Sketching Technique for \CASE{3},\CASE{6} and \CASE{7}. Further, in Section \ref{subsec:layered_technique}, we give Layered Algorithm which is enough to find the min-cut as given by \CASE{5} if it exists.

%% file: thesis_subparts/chapter5/overview_table.tex
%!TEX root = ../../thesis.tex
	\begin{table}[t]
	\centering
	\begin{tabular}{|c|c|l|l|}\hline		
		~ & 
		\multicolumn{1}{c|}{\textbf{Characterization}} & 
		\textbf{Case} 	  & 
		\multicolumn{1}{c|}{\textbf{Technique}}  		
		\\ \hline
		%%% 1-respect
		\multirow{2}{*}{1-respects $\tree$} & 
		\multirow{2}{*}{\centering \parbox{4cm}{\centering  -}} &
		\multirow{2}{*}{\CASE{1}} & 
		\multirow{2}{*}{\parbox{4cm}{check if for some $v$, $\eta(v) = 3$}} \\
		&&&\\
		%%% 2-respect
		\hline
		\multirow{4}{*}{2-respects $\tree$} & 
		\multirow{4}{*}{\centering \parbox{4cm}{\centering  Lemma \ref{lemma:gen_3_cut_2_respect}}} &		
		\multirow{2}{*}{\CASE{2}} & 
		\multirow{2}{*}{\parbox{4cm}{\emph{Broadcast Type - 1}\\ Section \ref{sec:covergecast_broadcast} }} \\
		&&&\\\cline{3-4}
		&&\multirow{2}{*}{\CASE{3}} & 
		\multirow{2}{*}{\parbox{4cm}{2-Sketch \\ Section \ref{subsec:sketching_method}}} \\
			&&&\\\cline{3-4}
		\hline
		%%% 3-respect
		
		\multirow{8}{*}{3-respects $\tree$} & 
		\multirow{8}{*}{\centering \parbox{4cm}{\centering Lemma \ref{lemma:gen_3_cut_3_respect}}} &		
		\multirow{2}{*}{\CASE{4}} & 
		\multirow{2}{*}{\parbox{4cm}{\emph{Broadcast Type - 2}\\ Section \ref{sec:covergecast_broadcast} }} \\
		&&&\\\cline{3-4}
		&&\multirow{2}{*}{\CASE{5}} & 
		\multirow{2}{*}{\parbox{4cm}{Layered Algorithm\\ Section \ref{subsec:layered_technique}}} \\
		&&&\\\cline{3-4}
		&&\multirow{2}{*}{\CASE{6}} & 
		\multirow{2}{*}{\parbox{4cm}{3-Sketch\\ Section \ref{subsec:sketching_method}}} \\
		&&&\\\cline{3-4}
		&&\multirow{2}{*}{\CASE{7}} & 
		\multirow{2}{*}{\parbox{4cm}{Reduced 2-Sketch\\ Section \ref{subsec:sketching_method}}} \\
		&&&\\\cline{3-4}
		\hline		
	\end{tabular}
\caption[Overview of the case-structure of min-cut of size $3$]{Overview of the case-structure of min-cut of size $3$. A min-cut of size $3$ if exists may share $1,2$ or $3$ edges with the fixed BFS tree $\tree$.}
\label{table:overview_3_cut}
\end{table}

%% file: thesis_subparts/chapter5/overview_figure.tex
%!TEX root = ../../thesis.tex
\begin{figure}
	\centering
	\begin{subfigure}[t]{0.5\textwidth}
		\centering
		\includegraphics[height=2.4in]{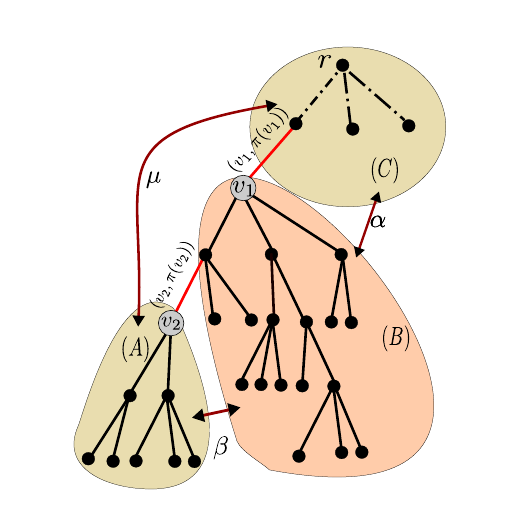}
		\caption{\CASE{2} (Either of $\alpha$ or $\beta$ is 1 and the other is $0$)}
	\end{subfigure}%
	~ 
	\begin{subfigure}[t]{0.5\textwidth}
		\centering
		\includegraphics[height=2.4in]{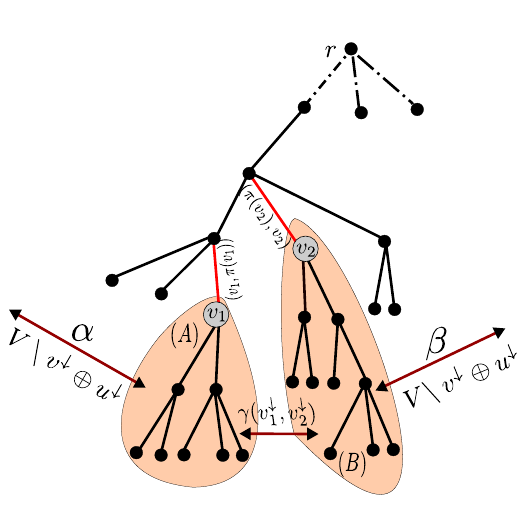}
		\caption{\CASE{3} (Either of $\alpha$ or $\beta$ is 1 and the other is $0$)}
	\end{subfigure}
	\begin{subfigure}[t]{0.5\textwidth}
		\centering
		\includegraphics[height=2.4in]{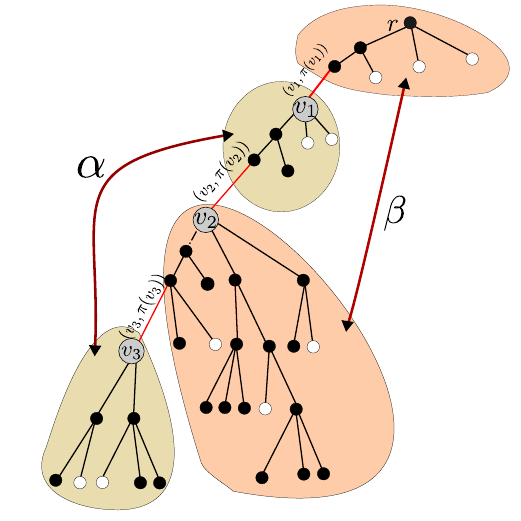}
		\caption{\CASE{4} (Both $\alpha$ and $\beta$ are non-zero)}
	\end{subfigure}%
	~ 
	\begin{subfigure}[t]{0.5\textwidth}
		\centering
		\includegraphics[height=2.4in]{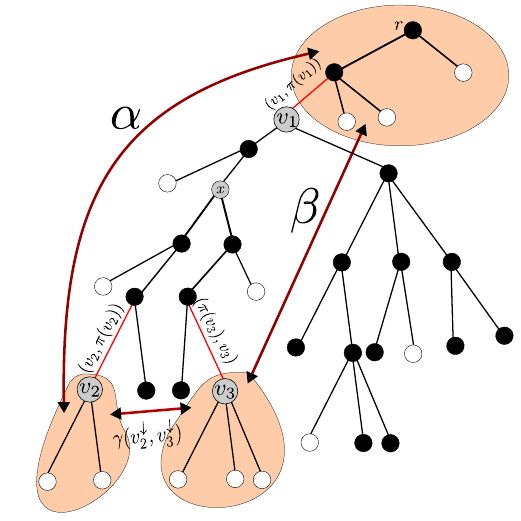}
		\caption{\CASE{5} (At least two of $\alpha$, $\beta$ and  $\gamma(\desc{v_2},\desc{v_3})$ are non-zero)}
	\end{subfigure}
	\begin{subfigure}[t]{0.5\textwidth}
		\centering
		\includegraphics[height=2.4in]{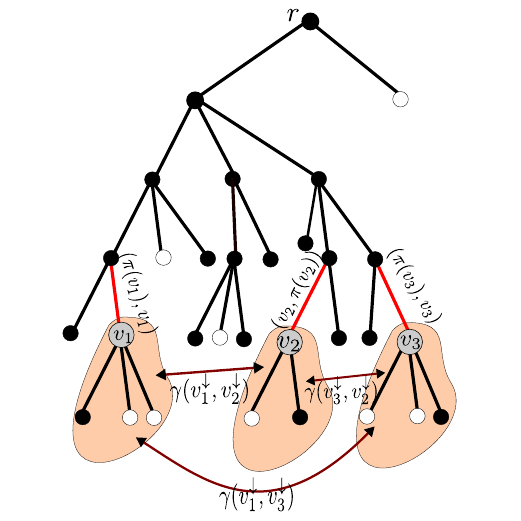}
		\caption{\CASE{6} (at least two of $\gamma(\desc{v_1},\desc{v_2})$, $\gamma(\desc{v_3},\desc{v_2})$, $\gamma(\desc{v_1},\desc{v_3})$ are non-zero)}
	\end{subfigure}%
	~ 
	\begin{subfigure}[t]{0.5\textwidth}
		\centering
		\includegraphics[height=2.4in]{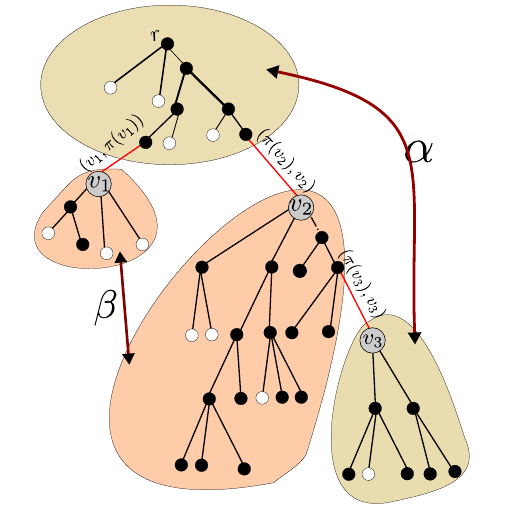}
		\caption{\CASE{7} (Both $\alpha$ and $\beta$ are non-zero)}
	\end{subfigure}
	\caption[Different cases of min-cuts of size three]{Different cases of an min cut of size three. Each figure in the above examples is a snippet of the tree and shows a different case of min-cut of size 3. Red edges are cut edges. Each shaded region in the above figures correspond to a vertex set. Sets with same color in the shaded region correspond to one side of cut. Thick edges with arrow ends may correspond to zero or more edges between two vertex set. The label on these edges represent the actual number of edges.}
	\label{fig:3-cut-cases}
\end{figure}

%% file: thesis_subparts/chapter5/sketching.tex
%TEX root = ../../thesis.tex
\section{Graph Sketching} \label{subsec:sketching_method}
In this section, we will introduce our graph sketching technique. 
Recall from Lemma \ref{lemma:gen_3_cut_2_respect} and Lemma \ref{lemma:gen_3_cut_3_respect} that to make a decision about min-cut of size $3$, a node requires certain information about other nodes. 
The whole graph has as many as $n$ nodes. 
It will be cost ineffective for every node to know the details about each of $n$ nodes. We introduce sketching technique, which reduces the number of nodes, any particular node has to scan, in order to find the required quantities to make a decision about the min-cut.

A sketch is defined for all the nodes in the network.
If there exists a min-cut, at least some node $v$, can make the decision regarding it using its sketch or appropriate sketch of other nodes communicated to it. 
In this section, first we will give the motivation behind the use of sketch, then in Subsection \ref{subsec:definition_sketch}, we will define sketch and introduce related 	notations. Here we will prove that the size of the sketch is not large. 
Later in Subsection \ref{subsec:algo_compute_k_sketch}, we will give the algorithm to compute the sketch and in Subsection \ref{subsec:application_graph_sketch}, we will showcase the application of graph sketch in finding a min-cut as given by \CASE{3}, \CASE{6} and \CASE{7}. Before going to the formal notation of the sketch definition, we will describe algorithmic idea for these cases.

First, we begin with \CASE{3}. Imagine that $\curly{(\parent u,u),(\parent v,v),e}$ is a min-cut of size $3$, for some nodes $v,u \in V\setminus r$ such that $\desc{v} \cap \desc{u} = \emptyset$ and a non-tree edge $e$ as given by \CASE{3} and shown in Fig. \ref{fig:motivation_sketch_case_3} (A). Now imagine node $u$ has to make a decision of this min-cut, then as per Lemma \ref{lemma:gen_3_cut_2_respect}, it will require information about $\eta(v), \gamma(\desc{u}, \desc{v})$. But upfront, node $u$ has no idea that it is part of such a min-cut and there exists some other node $v$, because it only has local information. Moreover, there are, as many as $n$ nodes, in the whole network which is lot of information for node $u$ to look at and is cost inefficient. Our sketching technique brings down the size of the set of nodes which any node $u$ has to scan to make a decision about a min-cut as give by \CASE{3}. Let \sloppy $\overline{\mathcal N}(A) \triangleq \curly{y \mid x \in A,y \notin A, (x,y) \in E\ \text{is a non-tree edge}}$. We make two simple observations: 
\begin{enumerate}[i)]
\item node $u$ needs to search only the paths from the root $r$ to all nodes $y \in \overline{\mathcal N}(\desc{u})$ shown in Fig. \ref{fig:motivation_sketch_case_3} (B). 
\item such paths can be significantly trimmed: for instance as shown in Fig. \ref{fig:motivation_sketch_case_3} (A), $\curly{(\parent u,u),(v,v_1),e}$ cannot be a min-cut because removing these edges does not partition the vertex set into two. 
\end{enumerate}
\begin{figure}
	\centering
	\includegraphics[scale=.50]{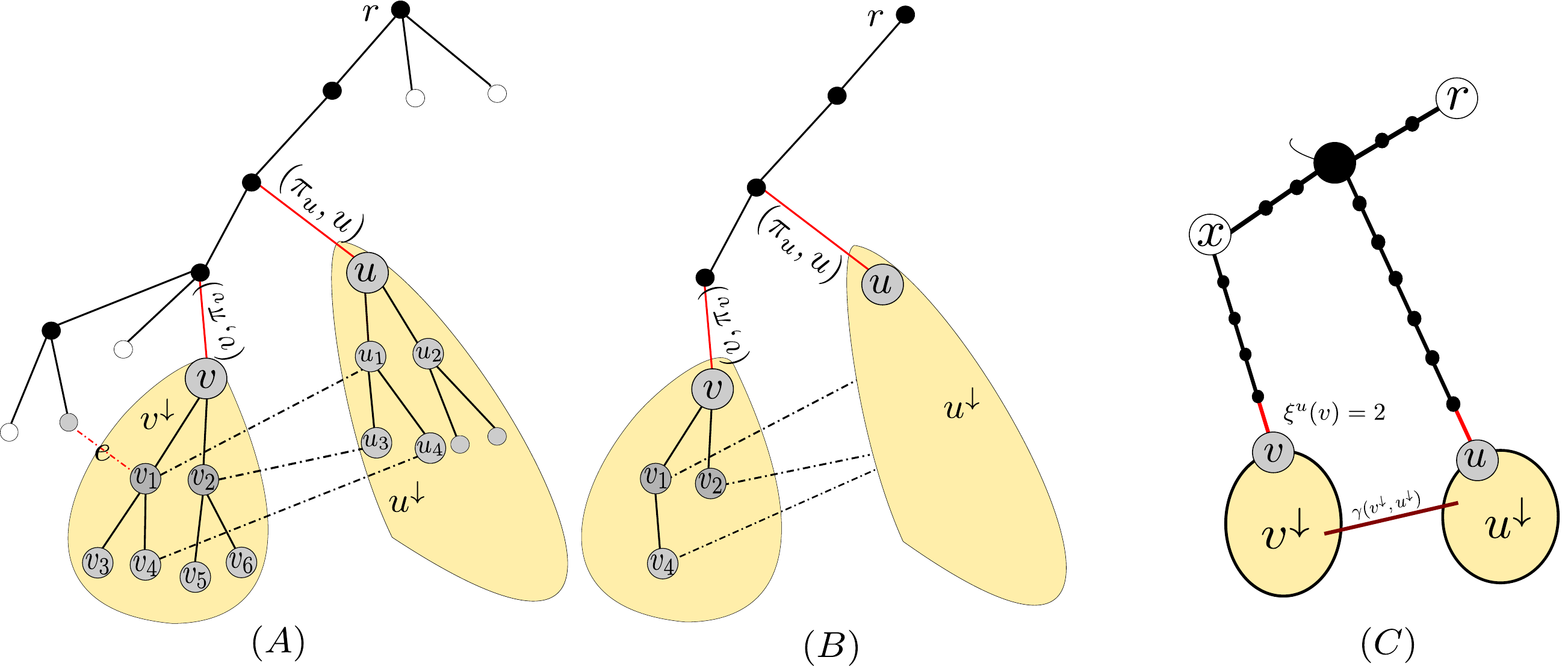}
	\caption[Motivation for Sketching Technique in \CASE{3}]{Demonstration of Sketching Technique for \CASE{3}. Each part in the above figure is a snippet of the tree. Edges with dashed stroke style are non-tree edges. Red edges are cut edges. Shaded region with vertices denote vertex sets.}
	\label{fig:motivation_sketch_case_3}
\end{figure}
Based on the above two simple observation, we can see that node $u$ can limit its scan of finding some node $x$ and subsequently $\eta(x), \gamma(\desc{x},\desc{u})$ to the bold path shown in Fig. \ref{fig:motivation_sketch_case_3} (C). Our sketch exactly computes this. We will give details about it in the later part of this section. The idea for \CASE{6} is similar to the one demonstrated here. 

To find a min-cut as given by \CASE{7}, we will use a different idea called \emph{reduced sketch}. Recall that, a min-cut as given by \CASE{7} is as follows: for $v_1,v_2,v_3 \in V\setminus r$, $\curly{(\parent{v_1},v_1),(\parent{v_2},v_2),(\parent{v_3},v_3)}$ is a min-cut such that $\desc{v_3} \subset \desc{v_2}$, $\desc{v_1} \cap \desc{v_2} = \emptyset$ such that $\desc{v_1} \cap \desc{v_3} = \emptyset$. Here we use the characterization given in Lemma \ref{lemma:gen_3_cut_3_respect} which requires that at least one node knows 6 quantities $\eta(v_1),\eta(v_2),\eta(v_3),\gamma(\desc{v_1},\desc{v_2}),\gamma(\desc{v_2},\desc{v_3}),\gamma(\desc{v_1},\desc{v_3})$.

In this case, we use a modified sketch. For any node $v$, our algorithm ensures that each node $c \in \desc{v}$ have information about strategically truncated and trimmed paths from root $r$ to all the vertices in the set $\overline{\mathcal{N}}(\desc{v}\setminus {\desc{c}})$. The same is illustrated in Fig. \ref{fig:motivation_sketch_case_7}. The pictorial representation of this case shown in Fig. \ref{fig:motivation_sketch_case_7} (A). Our algorithm makes sure that node $v_3$ (see that $v_3 \in v_2^\downarrow$) has information about the nodes in the bold path (specially truncated and trimmed paths from root to nodes in $\overline{\mathcal{N}}(v_2^\downarrow \setminus v_3^\downarrow)$) shown in Fig. \ref{fig:motivation_sketch_case_7} (C). The intermediate step is shown in Fig. \ref{fig:motivation_sketch_case_7} (B). Also, coupling it with Lemma \ref{lemma:H_v_a_known_to_all_nodes} node $v_3$ knows $H_{\desc{v_3}}^{v_2} = \gamma(\desc{v_3},\desc{v_2})$.
\begin{figure}
	\centering
	\includegraphics[scale=.50]{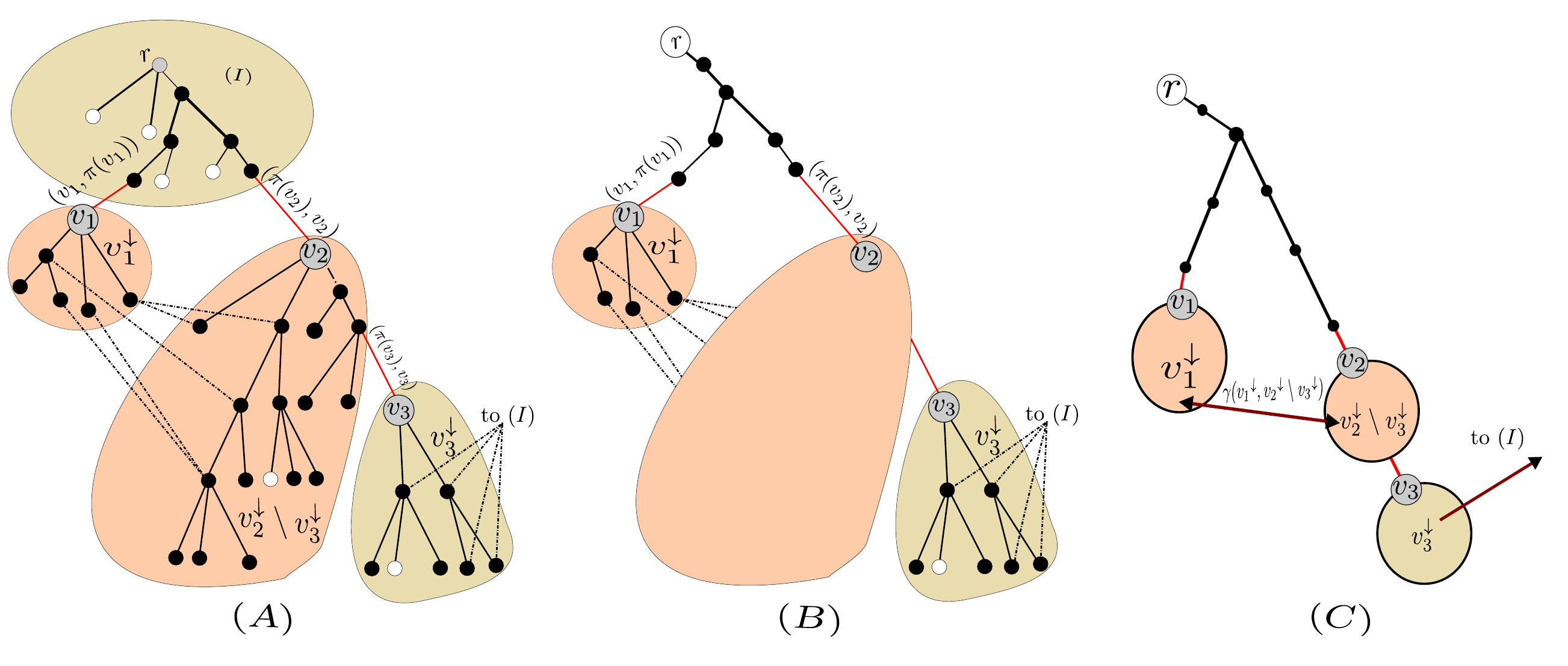}
	\caption[Motivation for Sketching Technique in \CASE{7}]{Motivation for Sketching Technique in \CASE{7}. Edges with dashed stroke style are non-tree edges. Shaded region with vertices denote vertex sets.}
	\label{fig:motivation_sketch_case_7}
\end{figure}

In the next subsection, we will give a formal definition of \emph{sketch} and \emph{reduced sketch}. We will give the definition for a general spanning tree $T$. The sketch is defined for a parameter $k$ which governs the number of branches which can be included in the sketch. Further, we will give distributed algorithms to compute sketch and reduced sketch. 
\input{thesis_subparts/chapter5/sketch_definition.tex}
\input{thesis_subparts/chapter5/sketch_algorithm.tex}
\input{thesis_subparts/chapter5/sketch_application.tex}

%% file: thesis_subparts/chapter5/sketch_definition.tex
%TEX root = ../../thesis.tex
\subsection{Definition of Sketch}
\label{subsec:definition_sketch}
For any node $x$, let $\rho_T(x)$ represent the unique path from root $r$ to the node $x$ in tree $T$.
Further for any vertex set $A \subseteq V$, let $\mathcal P_T(A) \triangleq \curly{\rho_T(x) \mid x \in A}$. Basically, $\mathcal P_T(A)$ is a set of paths.
We say that a tree path $\rho_T(x)$ is parallel to a tree edge $e = (a,b)$, if $\desc[T]{x} \cap \desc[T]{a} = \emptyset$ and $\desc[T]{x} \cap \desc[T]{b} = \emptyset$.
Also, for any vertex set $A \subset V$ and a tree $T$ recall that $\overline{\mathcal N}_T(A) = \curly{y \mid x \in A, (x,y) \in E, (x,y)\text{ is a non-tree edge}}$.

Now, we define canonical tree which is the first structure towards defining sketch. The sketch which we will define is nothing but a  truncation of this canonical tree. For any node $v$, the canonical tree is the graph-union of paths from the root $r$ to non-tree neighbors of node $v$. This notation for a canonical tree is also overloaded for a vertex set as well and formally defined below.
\begin{defn}[Canonical Tree]
	\label{defn:sketch_tree}
	\underline{Canonical tree of a node $v$} is a subtree of some spanning tree $T$, denoted by $R_{T}(v)$ and formed by union (graph-union operation) of tree paths in $\mathcal P_T\paren{\curly{v \cup \overline{\mathcal N}_T\paren{\curly{v}}}}$.
	\underline{Canonical tree of a vertex set $\desc[T]{v}$} is denoted by $R_{T}(\desc[T]{v})$ and formed by union of the paths in $\mathcal P_T\paren{\curly{v \cup \overline{\mathcal N}_T(\desc[T]{v})}}$.
\end{defn}
Further, we also define a reduced canonical tree. We use the same notation since the idea is same.
\begin{defn}
	Let $v$ be an internal (non-leaf and non-root) node of a tree $T$. Let $c \in \desc[T]{v}$, we define the \underline{reduced canonical tree} denoted by $R_{T}(\desc[T]{v}\setminus \desc[T]{c})$ and formed by union of the paths in $\mathcal P_T\paren{\curly{v \cup \overline{\mathcal N}_T(\desc[T]{v} \setminus \desc[T]{c})}}$
\end{defn}
We define sketch as the truncation of the canonical tree and give an algorithm that can compute it in $O(D^2)$ rounds. A canonical tree could be of very large size. Thus its truncation is required. To characterize this truncation we will use branching number as defined in Definition \ref{defn:branching_number}.
Let $\firstBranchNode(T')$ of any rooted tree $T'$ be the branch node (a node which has at least two children in a tree $T'$) closest to root. If the tree $T'$ has no branch node then $\firstBranchNode(T')$ is the root itself.

\begin{defn}[Branching Number]
	\label{defn:branching_number}
	For any tree $T'$, the branching number of a node $b$ in tree $T'$ is denoted by $\bnum[T']{b}$. It is defined as
	$$\bnum[T']{b} \triangleq  \begin{cases}
	1& \level[T']{b} \leq  \level[T']{x}\ \&\ x\neq root(T')\\
	2& {b} = {x} = root(T')\\
	deg_{_{T'}}(\parent[T']{b}) + \bnum[T']{\parent[T']{b}} - 2& \level[T']{b} > \level[T']{x}\\
	\end{cases}
	$$
	where $x = \firstBranchNode(T')$.
\end{defn}
The aforementioned definition is illustrated through examples in Figure \ref{fig:branch_number}. Basically, for any given tree $T'$, branching number of any node in the tree is a function of number of splits in paths from root to that node.
\begin{figure}[h]
	\centering
	\includegraphics[scale=1.5]{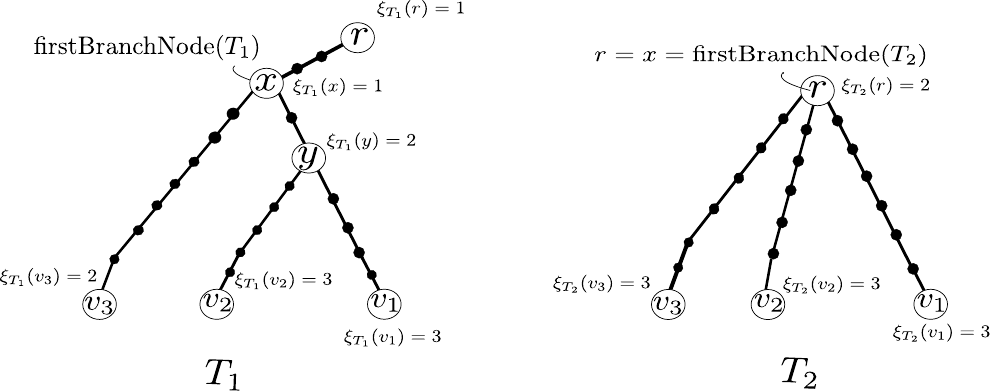}
	\caption[Branching Number]{Illustration of branching number on two separate trees $T_1$ and $T_2$}
	\label{fig:branch_number}
\end{figure}

We will now make a simple observation about branching number and give a characterizing lemma regarding the size of the canonical tree.
\begin{obs}
	\label{obs:branchin_number_increases}
	Let $v$ be a node and $c \in \child_T(v)$. Let $b \in R_T(\desc[T]{c})$, then $\xi_{R_T(\desc[T]{c})}(b) \leq \xi_{R_T(\desc[T]{v})}(b)$.
\end{obs}
\begin{proof}
	The canonical tree $R_T(\desc[T]{c})$ of any node $c \in \child_T(v)$ is the subtree of the canonical tree $R_T(\desc[T]{v})$. Hence the observation.
\end{proof}
\begin{lemma}
	For any tree $T'$, the number of nodes in the tree that has branching number less than $k$ is $O(2^k Depth(T'))$.
	\label{obs:nodes_R_t_v}
\end{lemma}
\begin{proof}
	In the worst case $T'$  may be a binary tree. Then each branching node in the tree will have a degree $3$ and on the path beyond that branching node will have branching number one more than the parent (by the definition of branching node). On every branching node there are two paths which split in the tree $T'$. Thus we may have as many as $O(2^k)$ different branching paths. Each path may be $O(Depth(T'))$ long thus we have $O(2^k Depth(T'))$ such nodes.
\end{proof}
We now define graph-sketch of a node which is defined based on the canonical tree and comes with a parameter $k$ on which the truncation is based. We define the truncation of a tree as below.
\begin{defn}
	For some tree $T'$ and a number $k$, $\trunc(T',k)$ is an sub-tree of $T'$ induced by vertices with branching number less than or equal to $k$.
\end{defn}

Our graph sketch will be called as $k$-Sketch because it comes with a parameter $k$ on which the truncation is based. For every node in the $k$-sketch, meta information is also added. We define the $k$-Sketch of a node as below.

\begin{defn}[$k$-Sketch]
	\label{defn:k_sketch}
	For any node $v$ and a spanning tree $T$, let $R' =  \trunc(R_T(\desc[T]{v}),k)$. The $k$-Sketch of a node $v$ w.r.t. the spanning tree $T$ is denoted as $\Sketch[T]{k}{v}$ and it is defined as $\Sketch[T]{k}{v} \triangleq \curly{R' \cup \curly{\angularbracs{u:\eta_{T}(u),\parent[T]{u},\gamma\paren{\desc[T]{v},\desc[T]{u}}}\ \forall u \in R'}}$
\end{defn}
Basically $k$-Sketch is a truncated canonical tree packaged along with the meta information $\angularbracs{u:\eta_{T}(u),\parent[T]{u},\gamma\paren{\desc[T]{v},\desc[T]{u}}}$ for each node $u$ in the truncated tree.

We will give an algorithm to compute $k$-Sketch for every node $v \in V$ in the next sub-section and further showcase the application of $k$-Sketch to find a min-cut  if it exists as given by \CASE{3} and \CASE{6}. Similar to $k$-Sketch of a node $v$ we define the reduced $k$-Sketch which is based on the reduced canonical tree. This will be used to find a min-cut  if it exists as given by \CASE{7}.
\begin{defn}[Reduced $k$-Sketch]
	\label{defn:sketch-substraction}
	Let v be an internal node of a spanning tree $T$. Let $c \in \desc[T]{v}$ and let $R' =  \trunc(R_T(\desc[T]{v} \setminus \desc[T]{c}),k)$. The reduced $k$-Sketch of a node $v$ and $c$ w.r.t. the spanning tree $T$ is denoted as $\rSketch[T]{k}{v}{c}$ and it is defined as $\rSketch[T]{k}{v}{c} \triangleq \curly{R' \cup \curly{\angularbracs{u:\eta_{T}(u),\parent[T]{u},\gamma\paren{\desc[T]{v} \setminus \desc[T]{c},\desc[T]{u}}}\ \forall u \in R'}}$
\end{defn}
We know give the following Lemma about the size of the $k$-Sketch.
\begin{lemma}
	For any spanning tree $T$, the k-Sketch of a node $v$, $\Sketch[T]{k}{v}$ w.r.t. $T$ is of size $O(2^k Depth(T) \log n)$ bits.
	\label{lemma:size_sketch}
\end{lemma}
\begin{proof}
	For any arbitrary node $u \in \Sketch[T]{k}{v}$, the sketch contains a three tuple $\angularbracs{\eta_{T}(u),\parent[T]{u},\gamma\paren{\desc[T]{v},\desc[T]{u}}}$. Here $\eta_{T}(u),\gamma\paren{\desc[T]{v},\desc[T]{u}} \leq |E| = O(n^2)$ and can be represented in $O(\log n)$ bits. Thus the three tuple is of $O(\log n)$ bits. Now by Lemma \ref{obs:nodes_R_t_v} it is clear that $\Sketch[T]{k}{v}$ is of size $O(2^k Depth(T) \log n)$ bits.
\end{proof}
\begin{corollary}
	For any spanning tree $T$ and an internal node $v$ and some $c \in \desc[T]{v}$ the reduced $k$-Sketch $\rSketch[T]{k}{v}{c}$ w.r.t. $T$ is of size $O(2^k Depth(T) \log n)$ bits.
	\label{lemma:size_reduced_sketch}
\end{corollary}
The $k$-Sketch of a node will be used to find a min-cut as given by \CASE{3} and \CASE{6}. Whereas the reduced $k$-Sketch will be used to find a min-cut as given by \CASE{7}. In the subsequent section, we will give algorithms to compute $k$-Sketch and the reduced $k$-Sketch. We will work with the fixed BFS tree $\tree$ and for simplicity in the notations the $\tree$ will be skipped from subscript or superscript. 

%% file: thesis_subparts/chapter5/sketch_algorithm.tex
%TEX root = ../../thesis.tex
\subsection{Algorithm to Compute Sketch}
\label{subsec:algo_compute_k_sketch}
In this subsection, we will give distributed algorithms to compute $k$-Sketch and the reduced $k$-Sketch. We will prove that our algorithm takes $O(D^2)$ rounds. The idea to compute sketch is as follows: Each node computes its own $k$-Sketch (which is of size $O(D)\log n$ bits) and communicates the same to the parent. The parent node after receiving the sketch from all the children computes its own sketch and communicates the same further up. This process continues and at the end each node has its $k$-Sketch. Here we will use Observation \ref{obs:branchin_number_increases}, to argue that the sketch received from children is enough for a node to compute its sketch.
\begin{lemma}
	For all $v \in V$, $\Sketch{k}{v}$ can be computed in $O(D^2)$ rounds.
	\label{lemma:k-sketch_algorithm}
\end{lemma}
\begin{proof}We describe a detailed algorithm to compute this in Algorithm \ref{algo:compute_sketch}. In Line \ref{algo:computes_rho_}, Algorithm \ref{algo:compute_sketch} calculates $\mathcal P\paren{\curly{v \cup \overline{\mathcal N}(v)}}$ which as per the definitions is the tree paths of the non-tree neighbors of $v$ including $\rho(v)$. \begin{algorithm}[!t]
		\DontPrintSemicolon
		\SetKwProg{module}{Module}{}{}
		\SetKwProg{algo}{Algorithm}{}{}
		\SetKwFor{Forp}{for}{parallely do}{endfor}
		\SetKwProg{phase}{}{}{}
		\setcounter{AlgoLine}{0}
		\phase{Algorithm to be run on each node $a$}{
			\KwOut{k-Sketch $\Sketch{k}{a}$}
			\phase{Past Knowledge}{
				Each node $u\in V$ knows $\level{u}$ and $\eta(u)$ from previous section\;
				For each $u \in \ancestor{a}$,  $a$ knows $H_{{a}}^{u}$ and $\eta(u)$
			}
			$\overline{\mathcal N}(a) \gets \curly{b\mid (a,b) \in E, (a,b)\ \text{is a non-tree edge}} $ \label{algo:line_1}\;
			\Forp{all $b \in \overline{\mathcal N}(a)$}{send tuples $\angularbrac{\level{u},\eta(u),u}$ for all $u \in \ancestor{a}$ to $b$} \label{algo:comm_1}
			$\mathcal P \gets \curly{\rho(a)}$ \tcp{$\rho(a)$ can be computed easily because $a$ has all nodes in $\ancestor{a}$}
			\Forp{all $b \in \overline{\mathcal N}(a)$}{
				\lFor{all $u \in \ancestor{b}$}{receive tuples $\angularbrac{\level{u},\eta(u),u}$}  \label{algo:comm_2}
				Construct the path $\rho(b)$ using the information \;
				$\mathcal P \gets \mathcal P\cup \rho(b)$ \label{algo:computes_rho_}\;       
			}
			perform graph union of all the paths in $\mathcal P$ and form canonical tree $R(a)$ \;
			\For{all nodes $u \in R(a)$}{
				\lIf{$u \in \ancestor{a}$}{$\gamma(a,\desc{u}) = H_{a}^{u}$}
				\lElse{$\gamma(a,\desc{u}) = |\curly{b \mid b\in \overline{\mathcal N}(a) , u \in \ancestor{b}}|$}
				include the tuple $\angularbrac{u:\eta(u), \parent{u},\gamma(a,\desc{u})}$ for the node $u \in R(a)$\label{algo:line_11}
			}   
			\lIf{$a$ is an internal node}{
				wait until $\Sketch{k}{c}$ is received for all $c \in \child(a)$  \label{algo:comm_3}
			}
			$S \gets R(a) \cup_{c \in \child(a)}\Sketch{k}{c}$ \label{line:compute-sketch-tree}\;
			perform graph union of all trees in $S$ to form a tree $T$\;
			\For{node $u \in T$}{
				compute branching number $\bnum[T]{u}$ as per the definition\;
				\If{$\bnum[T]{u} > k$ and $u \in \desc{v}$}{remove node $u$ from $T$}
			}
			\For{node $u \in T$}{
				compute $\gamma(\desc{a},\desc{u})$ by adding appropriately from all $T' \in S$\;
			}
			Construct $\Sketch{k}{a}$ using $T$ by including $\angularbracs{u : \eta(u),\parent{u},\gamma(\desc{a},\desc{u})}$ for all node $u \in T$\;
			Send $\Sketch{k}{a}$ to $\parent{a}$
		}
		\caption{Distributed Algorithm for Computing $k$-Sketch}
		\label{algo:compute_sketch}
	\end{algorithm}

	This can be computed easily because each non-tree neighbor sends all its ancestors and their associated meta information. Having calculated $\mathcal P\paren{\curly{v \cup \overline{\mathcal N}(v)}}$, then it is easy to compute the sketch tree $R(v)$ by Definition \ref{defn:sketch_tree}.
	Now if a node is an internal node we again need to perform a graph union of other sketches received from children.
	This is also trivial and it is guaranteed that we will not lose any node here because of Observation \ref{obs:branchin_number_increases}.
	Further branching number is computed for all the nodes and those nodes which do not satisfy the condition of branching number are removed to form the sketch.
	Also among the three tuples of meta information two remain fixed from the sketch of children and $\gamma(\desc{u},\desc{v})$ for a node $u \in \Sketch{k}{v}$ can be computed by appropriate addition.
	\paragraph{Time Requirements} In Algorithm \ref{algo:compute_sketch}, communication between any two nodes occur only in line \ref{algo:comm_1}, \ref{algo:comm_2} and \ref{algo:comm_3}. In line \ref{algo:comm_1}, only $O(D)$ rounds are taken because a node has at most $O(D)$ ancestors in the BFS tree $\tree$.
	Similarly, line \ref{algo:comm_2} also takes $O(D)$ rounds because each of the neighbors also has $O(D)$ ancestors and the transfer of the three tuples from each of the neighbors happen in parallel.
	Further in line \ref{algo:comm_3}, a node $a$ waits to recieve all the sketch from its children. From Lemma \ref{lemma:size_sketch}, we know that the size of the sketch is $O(D \log n)$ bits when the tree in action is a BFS tree. Now a node at level $l$ will wait for all the sketch from its children which are at level $l+1$ and they inturn depend on all the children which are at level $l+2$ and so on. Thus line \ref{algo:comm_3} takes atmost $O(D^2)$ rounds.
\end{proof}

In Lemma \ref{lemma:k-sketch_algorithm}, we gave an algorithm for computing the $k$-Sketch. We will now move toward an algorithm to find a reduced sketch. There are two steps towards these details of which are given in Observation \ref{reduced_sketch_step_1} and Lemma \ref{lemma:reduced-sketch}.
\begin{obs}
	\label{reduced_sketch_step_1}
	For a constant $k$. For any internal node $a$, and $c' \in \child(a)$, $\rSketch{k}{a}{c'}$ can be computed at node $a$ in $O(D^2)$ rounds.
\end{obs}
\begin{proof}
	Basically, the idea here is not to include the sketch received from child $c$, while computing the sketch of node $a$ and the resultant becomes the reduced sketch $\rSketch{k}{v}{c}$. To do this, we just need to change Line \ref{line:compute-sketch-tree} in Algorithm \ref{algo:compute_sketch} to $\mathcal S \gets R(a) \cup_{c \in \child(a) \setminus \curly{c'}}\Sketch{k}{c}$ for each $c' \in \child(a)$ and this enables the computation $\rSketch{k}{a}{c'}$ at node $a$. 
\end{proof}
\begin{lemma}
	\label{lemma:reduced-sketch}
	For a fixed $k$, there exists a $O(D^2)$ round algorithm such that all nodes $x$ can compute $\rSketch{k}{v}{c}$ for all $v \in \ancestor{x}$.
\end{lemma}
\begin{proof}
	For any internal node $v \in V$, we ensure that for all $c \in \child(v)$ the reduced k-Sketch $\rSketch{k}{v}{c}$ is downcasted to all the nodes in $\desc{c}$.
	This takes $O(D^2)$ rounds. After this step every node $x$ at level $\level{x}$ has the sketch $\rSketch{k}{a}{b}$ for any node $a,b$ at level $i$ and $i+1$ for all $i \in [1,l-1]$ and $a,b \in \ancestor{a}$. Now based on this we will show that node $x$ can compute $\rSketch{k}{v}{x}$ for all $v \in \ancestor{x}$ as per Algorithm     \ref{algo:compute-reduced-sketch}.
	
	Each node $x \in V\setminus r$ has to perform some local computation based on the various reduced sketch received earlier. For $1 \leq i\leq \level{x}-1$ let $\mathcal S_i = \rSketch{k}{a}{b}$ when $a  = \alpha(x,i)$ and $b = \alpha(x,i+1)$ (recall that $\alpha(x,l)$ is the ancestor of node $x$ at level $l$). For some node $v = \alpha(x,l)$ which is the ancestor of node $x$ at some level $l < \level{x}$; to compute $\rSketch{k}{v}{x}$ node $x$ uses $\mathcal S = \cup_{l\leq j\leq\level{x}-1} \mathcal S_j$. Here $\mathcal S$ is basically a set of Sketches. Note that any two sketch in the set $\mathcal S = \cup_{l\leq j\leq\level{x}-1} \mathcal S_j$  are overlapping. That is they have information about disjoint vertex sets. The rest of the steps are exactly same as given in Algorithm \ref{algo:compute_sketch} and are described in detail in Algorithm     \ref{algo:compute-reduced-sketch}
	\begin{algorithm}
		\DontPrintSemicolon
		\KwIn{$\rSketch{k}{a}{b}$ for any node $a,b$ at level $i$ and $i+1$ such that $1 \leq i\leq \level{x}-1$}
		\KwOut{$\rSketch{k}{v}{x}$ for all $v \in \ancestor{x}$}
		\lFor{$1 \leq i\leq \level{x}-1$}{$\mathcal S_i \gets \rSketch{k}{a}{b}$ when $a  = \alpha(x,i)$ and $b = \alpha(x,i+1)$}
		\tcp{recall that $\alpha(x,l)$ is the ancestor of node $x$ at level $l$}
		\For{$1 \leq i\leq \level{x}-1$}{
			$v  = \alpha(x,i)$ \;
			$S \gets \cup_{i\leq j\leq\level{x}-1} \mathcal S_j$\;
			perform graph union of all trees in $\mathcal S$ to form a tree $T$\;
			compute branching number $\bnum[T]{u}\ \forall u \in T$ as per the definition\;
			\For{node $u \in T$}{
				\If{$\bnum[T]{u} > k$ and $u \in \desc{v}$}{
					remove $u$ from $T$
				}
			}
			\For{node $u \in T$}{
				compute $\gamma(\desc{v} \setminus \desc{x},\desc{u})$ by adding appropriately from all $T' \in S$\;
			}
			
			Construct $\rSketch{k}{v}{x}$ using $T$ by including $\angularbracs{u:\eta(u),\parent{u},\gamma(\desc{v} \setminus \desc{x},\desc{u})}$ for all node $u \in T$\;
		}
		\caption{To be run on all node $x \in V\setminus \curly{r}$}
		\label{algo:compute-reduced-sketch}
	\end{algorithm}
\end{proof}

%% file: thesis_subparts/chapter5/sketch_application.tex
%TEX root = ../../thesis.tex
\subsection{Application of graph sketch}
\label{subsec:application_graph_sketch}

Now we will describe the application of graph sketch to find a min-cut of size $3$ if it exists as given by \CASE{3}, \CASE{6} and \CASE{7}. We prove the same in Lemma \ref{lemma:3-cut-case-3-using-sketch}. Here, \CASE{3} is a direct application of $3$-Sketch; in \CASE{6} we will be required  to move $3$-Sketch in a strategic way and in \CASE{7} reduced $2$-Sketch will be used. We give further details of each of the cases in Lemma \ref{lemma:3-cut-case-3-using-sketch}, \ref{lemma:3-cut-case-6-using-sketch} and \ref{lemma:3-cut-case-7-using-sketch}. 

For any two nodes $v_1,v_2$, let $\LCA_{T}(v_1,v_2)$ be the lowest common ancestor of node $v_1$ and $v_2$ in tree $T$. 

\begin{lemma}
\label{lemma:3-cut-case-3-using-sketch}
For some $v_1,v_2\in V\setminus r$ and a non-tree edge $e$, if $\curly{(\parent{v_1},v_1),(\parent{v_2},v_2),e}$ is a min-cut as given in \CASE{3}, then node $v_1$ can make a decision about the said min-cut using $\Sketch{3}{v_1}$.
\end{lemma}	
\begin{proof}
As per Lemma \ref{lemma:gen_3_cut_2_respect}, we know that node $v_1$ can decide for such a min-cut if it knows three quantities: $\eta(v_1),\eta(v_2)$ and $\gamma(\desc{v_1},\desc{v_2})$. Also, for any node $u,v$, we know that if there is a node $u \in \Sketch{3}{v}$ then the sketch also contains $\eta(u),\gamma(\desc{u},\desc{v})$. And every node $v$ in the network knows $\eta(v)$ from previous section. Thus here to prove this lemma we have to prove that node $v_2 \in \Sketch{3}{v_1}$. Then node $v_1$ can enumerate through all the nodes in $\Sketch{3}{v_1}$ which are not in $\ancestor{v}$ and apply the condition of Lemma \ref{lemma:gen_3_cut_2_respect} to test if there exists such a min-cut. 

\begin{figure}[h]
\centering
\includegraphics{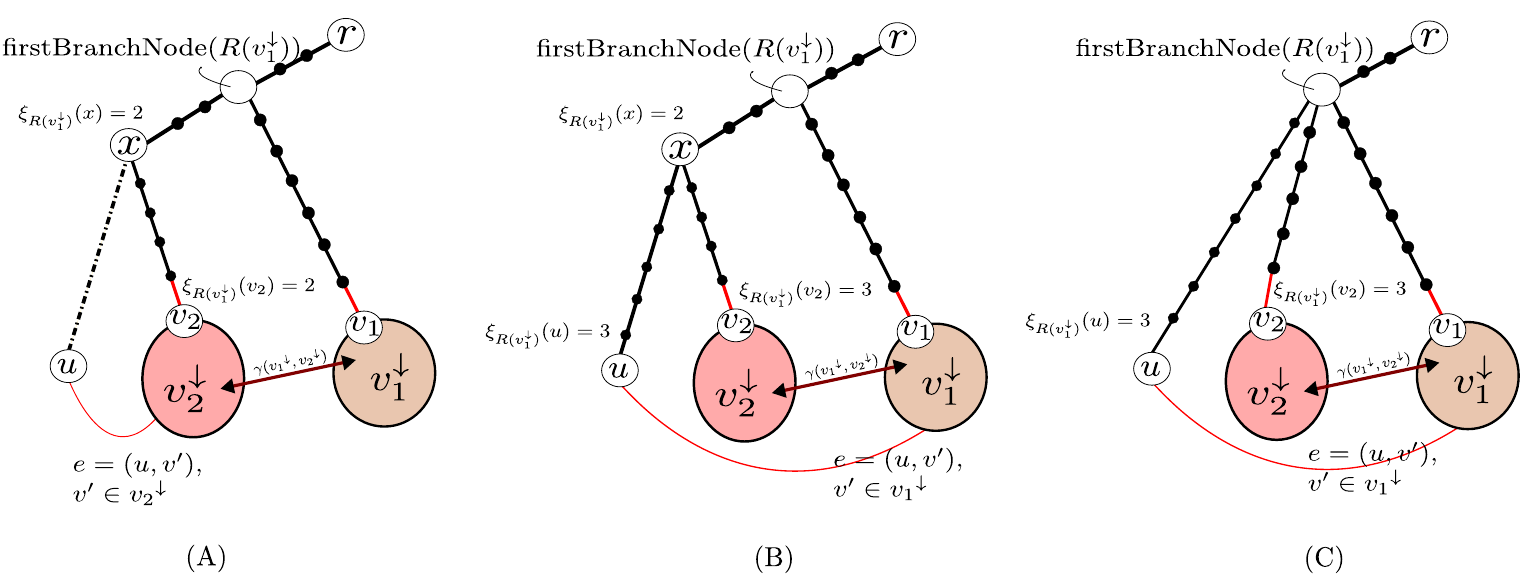}
\caption[Different types of sub-cases which occur when there exists a min-cut of size $3$ as in \CASE{3}]{Illustrating different types of sub-cases which might occur when there exists a min-cut of size $3$ as in \CASE{3}. The figures presented here are tree snippets. In each of the sub-cases we demonstrate the $3$-Sketch as computed by the node $v_1$. Solid black line denote the tree paths in Sketch of $v_1$. Dashed line represent other paths in tree $\tree$. Edges in red are cut edges. Here all the edges that go out of the vertex set $\desc{v_1}$ have their other end points in the vertex set $\desc{v_2}$ barring one non-tree edge $e$ in (B) and (C). The important fact here is that in all the different cases $v_2 \in \Sketch{3}{v_1}$}. 
\label{fig:case_3_lemma_sketch_will_contain_info}
\end{figure}
We will now show that if such a min-cut exists then node $v_2 \in \Sketch{3}{v_1}$. As illustrated in Fig. \ref{fig:case_3_lemma_sketch_will_contain_info} for all the different sub-cases node $v_2 \in \Sketch{3}{v_1}$. In (A) when the other non-tree edge $e$ has one end-point in $\desc{v_2}$ then the branching number of $v_2$ in the sketch-tree $R(\desc{v_1})$ is $\bnum[R(\desc{v_1})]{v_2} = 2$ thus $v_2 \in \Sketch{3}{v_1}$. Both (B) and (C) are similar in terms of the fact that the non-tree cut edge $e$ has the other endpoint in $\desc{v_1}$ but differ in terms of the branching node. Nevertheless here also $\bnum[R(\desc{v_1})]{v_2} = 3$ thus $v_2 \in \Sketch{3}{v_1}$

\end{proof}

\begin{lemma}
\label{lemma:3-cut-case-6-using-sketch}
For some $v_1,v_2,v_3\in V\setminus r$, let $\desc{v_1},\desc{v_2},\desc{v_3}$ be pair wise mutually disjoint. If there exists a min cut as in \CASE{6} such that $\curly{(\parent{v_1},v_1),(\parent{v_2},v_2),(\parent{v_3},v_3)}$ is a min-cut, then it can be found in $O(D^2)$ time.
\end{lemma}
\begin{proof}
For such a min-cut to exists then at least two of $\gamma(\desc{v_1},\desc{v_2}),\gamma(\desc{v_1},\desc{v_3}),\gamma(\desc{v_2},\desc{v_3})$ need to be non-zero otherwise the vertex sets $\desc{v_1} \cup \desc{v_2} \cup \desc{v_3}$ may not form a connected component. WLOG here we may have two non-isomorphic cases as illustrated in Fig. \ref{fig:non_ismorohpic_cases_of_case_6}.	
\begin{figure}[h]
   \centering
    \begin{subfigure}[t]{0.45\textwidth}
	   \centering
	   \includegraphics[height=1.4in]{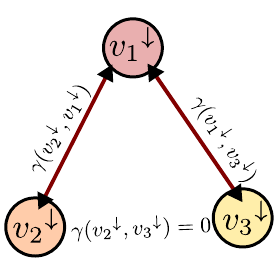}
	   \caption{One of $\gamma(\desc{v_1},\desc{v_2}),\gamma(\desc{v_1},\desc{v_3}),\gamma(\desc{v_2},\desc{v_3})$ is zero. WLOG let $\gamma(\desc{v_2},\desc{v_3}) = 0$}
	   \label{sub-fig:line-graph}
    \end{subfigure}    \hspace{4mm}
    \begin{subfigure}[t]{0.45\textwidth}
		\centering
		\includegraphics[height=1.4in]{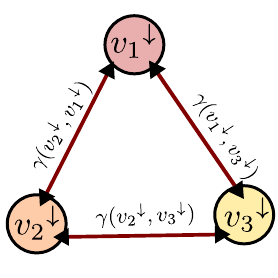}
		\caption{All three of $\gamma(\desc{v_1},\desc{v_2}),\gamma(\desc{v_1},\desc{v_3}),\gamma(\desc{v_2},\desc{v_3})$ are non-zero}
		\label{sub-fig:all-connected}
    \end{subfigure}
    \caption[Two different non-isomorphic sub-cases of a min-cut as given by \CASE{6}]{Two different non-isomorphic sub-cases of a min-cut as in \CASE{6}. Here a circle correspond to a vertex set. And an double arrow line between them indicate that there are some edges which have an end point in each of the vertex set.}
    \label{fig:non_ismorohpic_cases_of_case_6}
\end{figure}
Also, as per Lemma \ref{lemma:gen_3_cut_3_respect} we need $\eta(v_1), \eta(v_2),\eta(v_3),\gamma(\desc{v_1},\desc{v_2}),\gamma(\desc{v_2},\desc{v_3})$ and $\gamma(\desc{v_1},\desc{v_3})$ at at least one node to decide for a min-cut as given by \CASE{6}. 

We will first work with the case in Fig. \ref{sub-fig:line-graph}. Here, $\gamma(\desc{v_2},\desc{v_3}) = 0$. In this sub-case node $v_1$ can make the decision based on $\Sketch{3}{v_1}$. We will prove that if such a min-cut exists then $v_2,v_3 \in \Sketch{3}{v_1}$. We demonstrate the different cases in Figure \ref{fig:case_6_lemma_sketch_will_contain_info}. The three different cases are in terms of the intersection of the paths $\rho(v_1),\rho(v_2),\rho(v_3)$. In all the sub-cases demonstrated in the Fig. \ref{fig:case_6_lemma_sketch_will_contain_info} we can see that the branching number $\bnum{v_1}{v_2}, \bnum{v_1}{v_3} \leq 3$. Thus $v_2,v_3 \in \Sketch{3}{v_1}$. Here, $v_1$ just need to pick pairs of nodes $a,b \in \Sketch{3}{v_1}$ such that $\desc{a},\desc{b},\desc{v_1}$ are pair-wise mutually disjoint and compare $\eta(a),\eta(b),\eta(v_1),\gamma(\desc{a},\desc{v_1})$ and $\gamma(\desc{b},\desc{v_1})$ as per Lemma \ref{lemma:gen_3_cut_3_respect}.

\begin{figure}[h]
\centering
\includegraphics{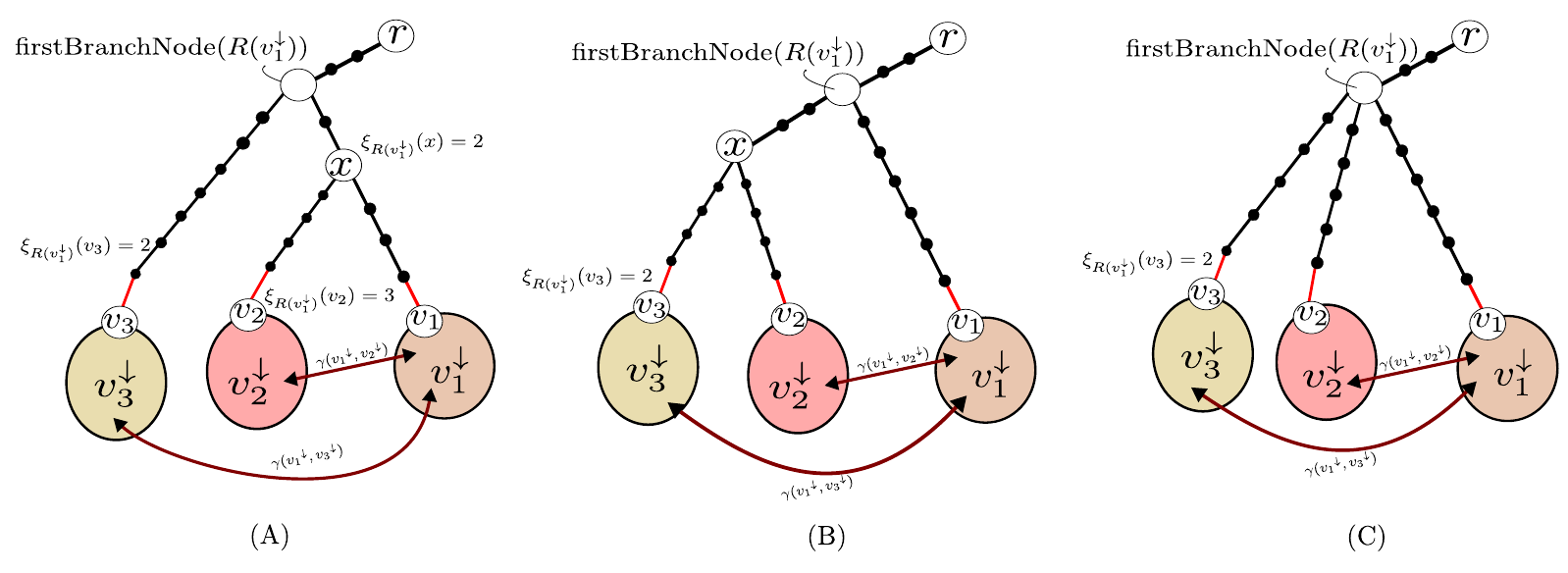}
\caption[3-Sketch of different cases as viewed from node $v$]{Illustrating different types of sub-cases. We demonstrate the $3$-Sketch as per the node $v_1$. Solid black line denote the tree paths in Sketch of $v_1$. Dashed line represent other paths in tree $\tree$. Edges in red are cut edges. The different sub-cases occur based on the intersection of three paths $\rho(v_1),\rho(v_2),\rho(v_3)$. In sub-case (A), $\LCA(v_1,v_3) = \LCA(v_2,v_3) \neq \LCA(v_1,v_2)$. In sub-case (B), $\LCA(v_1,v_2) = \LCA(v_1,v_3) \neq  \LCA(v_2,v_3)$. In sub-case (C), $\LCA(v_1,v_2) = \LCA(v_1,v_3) = \LCA(v_2,v_3)$. }
\label{fig:case_6_lemma_sketch_will_contain_info}
\end{figure}

Now, we will take the case as mentioned in Fig. \ref{sub-fig:all-connected}. By the same argument as above we can say that $v_1,v_2 \in \Sketch{3}{v_3}$ and $v_1,v_3 \in \Sketch{3}{v_2}$. In this case $\gamma(\desc{v_2},\desc{v_3}) \neq 0$ and $\Sketch{3}{v_1}$ does not have $\gamma(\desc{v_2},\desc{v_3})$ which is required  to make a decision about the said min-cut as per Lemma \ref{lemma:gen_3_cut_3_respect}. Thus node $v_1$ cannot make the required decision based on only $\Sketch{3}{v_1}$. 

Amidst, the above mentioned challenge, there is an easy way around. 
Each node $v \in V$ downcasts $\Sketch{3}{v}$ (its 3-Sketch) to all nodes in $\desc{v}$. 
Since 3-Sketch is of size $O(D \log n)$ thus this takes $O(D^2)$ rounds. 
After this step every node $u$ has $\Sketch{3}{v}\ \forall v\in \ancestor{u}$. 
Now each node $u$ sends $\Sketch{3}{a}\ \forall a \in \ancestor{u}$ to its non-tree neighbors. 
Since there are as many as $D$ such sketches thus this takes $O(D^2)$ time. Also, $\gamma(\desc{v_1},\desc{v_2}) > 0$ thus we know that there is at least one node node $v_1' \in \desc{v_1}$ and $v_2' \in \desc{v_2}$ such that $(v_1',v_2')$ is a non-tree edge. 
After the above mentioned steps node $v_1'$ as well as node $v_2'$ will have both $\Sketch{3}{v_1}$ and $\Sketch{3}{v_2}$. 
Now, both ${v_1'}$ and ${v_2'}$ can make the decision about the said min-cut. 
We will discuss steps which may be undertaken at node $v_1'$ for the sane, 
for each $v \in \ancestor{v_1'}$, $v_1'$ locally iterates through all the parallel paths of the $\Sketch{3}{v}$ and picks all possible $x,y$ such that $\desc{x},\desc{y},\desc{v}$ are mutually disjoint. 
Notice that through $\Sketch{3}{v}$, $v_1'$ has $\gamma(\desc{x},\desc{v}),\gamma(\desc{y},\desc{v}),\eta(x),\eta(y)$. Also it has $\eta(v)$ from previous calculations. Now from the sketches received through its non-tree neighbors $v_1'$ looks for $\gamma(\desc{x},\desc{y})$. If $\gamma(\desc{x},\desc{v}),\gamma(\desc{y},\desc{v}),\eta(x),\eta(y)$ and $\gamma(\desc{x},\desc{y})$ then satisfies Lemma \ref{lemma:gen_3_cut_3_respect} then it can make the decision about the required min-cut.
\end{proof}                                                                                          	

Next, we move to \CASE{7}. Here we will use reduced sketch as given in Definition \ref{defn:sketch-substraction}. First we the use the reduced $2$-sketch. 
\begin{obs}
\label{obs:3-cut-case-7-node-in-sketch}
For some $v_1,v_2,v_3\in V\setminus r$, let $\desc{v_3} \subset \desc{v_2}$, $\desc{v_1} \cap \desc{v_2} = \emptyset$ and $\desc{v_1} \cap \desc{v_3} = \emptyset$. If there exists a min cut as in \CASE{7} such that $\curly{(\parent{v_1},v_1),(\parent{v_2},v_2),(\parent{v_3},v_3)}$ is a min-cut then $v_1 \in \rSketch{2}{v_2}{v_3}$
\end{obs}
\begin{proof}
If such a min-cut exists then all the edges going out of the vertex set $\desc{v_2} \setminus \desc{v_3}$ apart from $(\parent{v_2},v_2)$ and $((\parent{v_3},v_3))$ goes to the vertex set $\desc{v_1}$. Thus the reduced canonical tree $R(\desc{v_2} \setminus \desc{v_3})$ contains node $v_1$. We showcase this scenario in Fig. \ref{lemma:3-cut-case-7-using-sketch}. Also the branching number of $v_1$ will be $\bnum[R(\desc{v_2} \setminus \desc{v_3})]{v_1} = 2$ based on the definition. Thus $v_1 \in  \rSketch{2}{v_2}{v_3}$
\begin{figure}[h]
   \centering
   \includegraphics[scale=1.2]{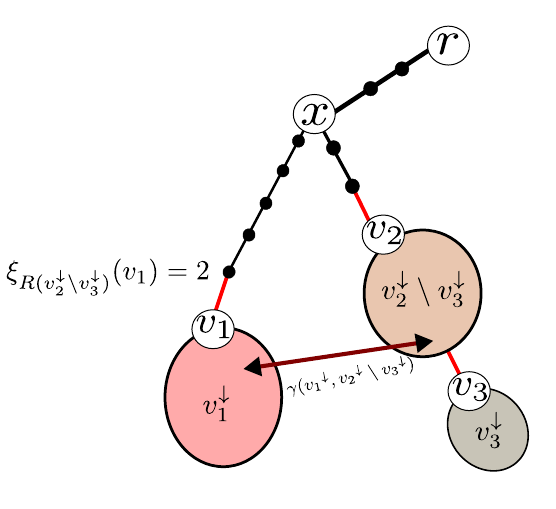}
    \caption[Reduced sketch $\rSketch{2}{v_2}{v_3}$ when for some $v_1,v_2,v_3\in V\setminus r$ as in \CASE{7}]{Reduced sketch $\rSketch{2}{v_2}{v_3}$ when for some $v_1,v_2,v_3\in V\setminus r$, let $\desc{v_3} \subset \desc{v_2}$, $\desc{v_1} \cap \desc{v_2} = \emptyset$ and $\desc{v_1} \cap \desc{v_3} = \emptyset$ and $\curly{(\parent{v_1},v_1),(\parent{v_2},v_2),(\parent{v_3},v_3)}$ is a min-cut}
\end{figure}
\end{proof}
\begin{lemma}
\label{lemma:3-cut-case-7-using-sketch}
For some $v_1,v_2,v_3\in V\setminus r$, let $\desc{v_3} \subset \desc{v_2}$, $\desc{v_1} \cap \desc{v_2} = \emptyset$ and $\desc{v_1} \cap \desc{v_3} = \emptyset$. If there exists a min cut as in \CASE{7} such that $\curly{(\parent{v_1},v_1),(\parent{v_2},v_2),(\parent{v_3},v_3)}$ is a min-cut then it can be found in $O(D^2)$ time.
\end{lemma}
\begin{proof}
Here we run Algorithm \ref{algo:case-7} on each node. 
\begin{algorithm}[h]
\DontPrintSemicolon
\SetKwProg{phase}{}{}{}
\phase{Past Knowledge}{
	$\rSketch{2}{v}{x}$ for all $v \in \ancestor{x}$ using Lemma \ref{lemma:reduced-sketch}\;
	For each $u \in \ancestor{x}$, $x$ knows $H_{\desc{x}}^{u}$ and $\eta(u)$
}
\setcounter{AlgoLine}{0}
\For{$v \in \ancestor{x} \setminus x$}{
	\For{$u \in \rSketch{2}{v}{x}\ \&\ u \notin \ancestor{x}$}{
		\label{line:quantitites_case-7} \If{$\eta(v) - 1 = H_{\desc{x}}^{v} + \gamma(\desc{u},\desc{v}\setminus\desc{x})\ \&$ $\eta(u) - 1 = \gamma(\desc{u},\desc{v}\setminus\desc{x})\ \&$ $\eta(x) - 1 = H_{\desc{x}}^{v}$} 
		{ 		\label{line:algo-case-7-if-block} 
			$\curly{(\parent{x},x),(\parent{v},v),(\parent{u},u)}$ is a min-cut
		}
	}
}
\caption{Algorithm to find an induced cut fo size $3$ as given in \CASE{7} to be run on all node $x \in V\setminus r$}
\label{algo:case-7}
\end{algorithm}
By Lemma \ref{lemma:gen_3_cut_3_respect} the reported min-cut is correct. Further when Algorithm \ref{algo:case-7} is run on node $v_3$ it will be able to make decision about the required min-cut because by Observation \ref{obs:3-cut-case-7-node-in-sketch} if there exists such a min-cut then $v_1 \in \rSketch{2}{v_2}{v_3}$. 

Also, this just requires $O(D^2)$ time because computing reduced sketch $\rSketch{2}{v}{x}$ at any node $x$ for all $v \in \ancestor{x}$ just takes $O(D^2)$ rounds as per Lemma \ref{lemma:reduced-sketch}.
\end{proof}

%% file: thesis_subparts/chapter5/layered_algorithm.tex
%!TEX root = ../../thesis.tex
\section{Layered Algorithm} \label{subsec:layered_technique}
In the last section, we gave a technique to find a min-cut of size 3 as in \CASE{3}, \CASE{6} and \CASE{7} using a special graph-sketch. 
In this section, we will give an algorithm to find the min-cut as given by \CASE{5}. 
A $k$-Sketch cannot be used to find a min-cut as in \CASE{5} because here it is challenging for one of the nodes to know the $6$ quantities as required by Lemma \ref{lemma:gen_3_cut_3_respect} using a $k$-sketch. 
To resolve this challenge, we give layered algorithm where we solve for min-cut iteratively many times.

Recall a min-cut as given by \CASE{5} is as follows: for some node $v_1,v_2,v_3 \in V\setminus r$, $\curly{(\parent{v_1},v_1),(\parent{v_2},v_2),(\parent{v_3},v_3)}$ is a min-cut such that $\desc{v_2} \subset \desc{v_1}$ and $\desc{v_3} \subset \desc{v_1}$ and $\desc{v_2} \cap \desc{v_3} = \emptyset$. For the introduction of layered algorithm, let us assume that such a min-cut exists. Further let $v_1,v_2$ and $v_3$ be these specific nodes.
In Fig. \ref{fig:min-cut_case5}, we show a pictorial representation of such a min-cut.

	\begin{figure}[h]
	\centering
	\includegraphics[height=2.5in]{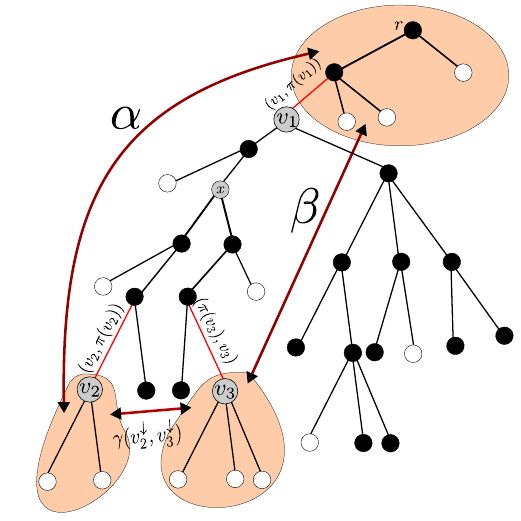}
	\caption{Min-cut of size $3$ as given by \CASE{5}}
	\label{fig:min-cut_case5}
\end{figure}

Similar to the previous section, we will use the characterization Lemma \ref{lemma:gen_3_cut_3_respect} which requires six quantities to make a decision about the min-cut. 
These are $\eta(v_1),\eta(v_2),\eta(v_3),\gamma(\desc{v_1},\desc{v_2}),\gamma(\desc{v_3},\desc{v_2})$ and $\gamma(\desc{v_1},\desc{v_3})$.

In this subsection, we will show that node $x = \LCA(v_2,v_3)$ can find all these six quantities. The challenge here is the fact that some of these quantities have to come from node $v_2$ and/or $v_3$.
Since node $x$ is higher up in the tree $\tree$ than $v_2,v_3$ thus the information from node $v_2,v_3$ when it travels up the tree may face contention from other nodes in $\desc{x}$. From Recall (from chapter \ref{chpater:2_prelim_background}) that convergecast is such a technique to send data to nodes which are higher up in the tree from nodes at a lower level. For convergecast to be efficient here we need to couple it with a contention resolution mechanism. The idea is as follows:
\begin{itemize}
	\item We calculate min-cuts of size one and two in sub-graphs \emph{pivoted} at nodes at different level in the tree $\tree$ (details are deferred to Section \ref{subsec:defn_layered_algo})
	\item After the above step, each node computes its \emph{one-cut detail} and \emph{two-cut detail} (definitions are deferred to Section \ref{sec:appl_layered_algorithm})
	\item Based on the above computation, if a min-cut exists as given by $\CASE{5}$, then node $v_2,v_3$ can designate themselves as special nodes using Lemma \ref{lemma:case_5_layered_algo_characterization_A}. Thus, the information from node $v_2$ and node $v_3$ could reach $LCA(v_2,v_3)$ efficiently.
\end{itemize}

\subsection{Definition of Layered Algorithm} \label{subsec:defn_layered_algo}
For any node $v \in V\setminus r$, let $\edgesatvertex{v}$ be the set of edges whose both endpoints are in the vertex set $v^\downarrow$.  Let subgraph \textit{pivoted} at any vertex $v\neq r$ be $\graphatvertex{v} \triangleq (v^\downarrow \cup \parent v,\edgesatvertex{v} \cup (\parent v,v))$. Note that this definition is different from the subgraph rooted at $v$ because here we add the extra edge $(\parent{v},v)$ and vertex $\pi(v)$.

\emph{Layered Algorithm} is a special type of algorithm where we will solve for induced-cut of size one and two in a layered fashion. 
As usual, we will work with the fixed BFS tree $\tree$. Here we will solve for induced-cut of size $1$ and $2$ repeatedly for $Depth(\tree) = D$ times. 
In the first iteration, our algorithm solves for these induced cuts in the whole graph $G$. 
Then it solves for the min-cuts of size $1$ and $2$ in all the sub-graph pivoted at nodes at level $1$; that is in all sub-graphs in the set $\curly{G_v \mid \level{v} = 1}$, 
subsequently in the sub-graphs for all nodes pivoted at level $2$ and so on until level $D$. 
This is the \emph{Layered Algorithm for min-cut}. 
We will showcase the utility of this algorithm later. 
In Observation \ref{obs:layered_algorithm_run_time}, we will argue that the layered algorithm for min-cut takes $O(D^2)$ rounds.

\begin{obs}
	\label{obs:layered_algorithm_run_time}
	The layered algorithm for min-cut runs in $O(D^2)$ rounds.
\end{obs}
\begin{proof}
	From Theorem \ref{thm:min-cut-size-1} and Lemma \ref{lemma:induced-cut-2-cut-2-respects}, we know that cuts of size 1,2 can be found in $O(D)$ rounds. 
	For all graphs $\curly{G_v \mid \level{v} = i}$ pivoted at any level $1 \leq i \leq D$ the algorithm to compute min-cut of size $1,2$ can be run independently because none of the sub-graphs in the set $\curly{G_v \mid \level{v} = i}$ share an edge. Thus it takes $O(D)$ round to run the min-cut algorithm in all such graphs. Now to run for all graphs pivoted at nodes at all levels it takes $O(D^2)$ rounds.
\end{proof}

We will now define the notation for the sequence of subgraphs. This is a short hand notation to say that a property is satisfied for a set of pivoted sub-graphs. Let $r\rightarrow w_1 \rightarrow w_2 \rightarrow \cdots \rightarrow w_{t}$ be a path in the BFS tree. Then we know that $\graphatvertex{w_1}, \graphatvertex{w_2}, \ldots, \graphatvertex{w_{t}}$ are subgraphs of $G$ as defined earlier. Let $i,j \in [1,t]$ and $i\leq j$, then let  $\subgraphsequence{w_i}{w_j}$ be a set of subgraphs such that $\subgraphsequence{w_i}{w_j} \triangleq \curly{G_{w_h}\mid i\leq h \leq j}$. We say a property is true for $\subgraphsequence{w_i}{w_j}$ if it is true for all the subgraphs in the set.

\begin{obs}
	\label{obs:layered_algor_graph_sequence_min_cut}
	For some nodes $x$ and $v$, let $x \in \desc{v}$. If $(\parent{x},x)$ is a min-cut for the sub-graph $G_v$ then $(\parent{x},x)$ is a min-cut for all sub-graphs in $\subgraphsequence{x}{v}$. Similarly, let $a,b$ be two nodes. Let node $z$ be the LCA of node $a,b$ in $tree$ and $v$ be some node such that ${v} \in \ancestor{z}$. If $\curly{(\parent{a},a),(\parent{b},b)}$ is a min cut for the sub-graph $G_v$ then it a min-cut for all the sub-graphs in $\subgraphsequence{z}{v}$.
\end{obs}
%\begin{proof}
%\ToDoY{Mohit}{write proof}
%\end{proof}

\subsection{Application of layered algorithm}
\label{sec:appl_layered_algorithm}
We will now showcase the application of the above-mentioned layered algorithm in finding a min-cut as given by \CASE{5}. First, we will give a simple process in Algorithm \ref{algo:1-2-cut-detail} to be run on individual nodes after the run of the \emph{Layered Algorithm for min-cut}. For each node $a$, this collects two quantities 1-cut detail $\mathcal D^{1}(a)$ and 2-cut detail $\mathcal D^{2}(a)$  which relevant information. After this information is collected by each node $a$ it will be sent upwards via a convergecast technique described later. 

Further, we will characterize \CASE{5} in  Lemma \ref{lemma:case_5_layered_algo_characterization_A}. Recall that in \CASE{5}, for some node $v_1,v_2,v_3 \in V\setminus r$,  $\curly{(\parent{v_1},v_1),(\parent{v_2},v_2),(\parent{v_3},v_3)}$ is a min-cut such that $\desc{v_2} \subset \desc{v_1}$ and $\desc{v_3} \subset \desc{v_1}$ and $\desc{v_2} \cap \desc{v_3} = \emptyset$. This will help us prove that one of the information from $\curly{\mathcal D^2\paren{v_2},\mathcal D^2\paren{v_3},\curly{\mathcal D^1\paren{v_3},\mathcal D^1\paren{v_2}}}$ will be received by $\LCA(v_2,v_3)$ after the convergecast algorithm. 
And this will be enough to make a decision about the min-cut based on the characterization given in Lemma \ref{lemma:gen_3_cut_3_respect}.

\begin{algorithm}
	\DontPrintSemicolon
	\SetKwFor{Forp}{for}{parallely}{endfor}
	\SetKwFor{Forw}{for}{wait}{endfor}
	\SetKwProg{phase}{}{}{}
	\phase{computing 1-cut detail}{
		$v \gets$ closest node to $r$ such that $(\parent{a},a)$ is a min-cut in $G_v$\;
		$\mathcal D^1(a) =  \angularbrac{node:a,parent:\parent{a},eta:\eta(a),pivotnode:v,numedges:H_{\desc{a}}^v,pivotnodelevel:\level{v}}$
	} \setcounter{AlgoLine}{0}
	\phase{computing 2-cut detail}{
		$vSet \gets \{u \mid \exists \text{ node } x\ s.t.\ \desc{x} \cap \desc{a} = \emptyset;\ \gamma(\desc{x},\desc{a}) > 0;\ \level{a} \geq \level{x};\curly{(\parent{a},a),(\parent{x},x)}$ $\text{ is a induced cut in }G_u\}$ \;
		\lIf{ $vSet = \emptyset$
		}{$\mathcal D^2(a) = \phi$}
		\Else{    
			$v \gets \underset{u \in vSet}{argmin}\ \level{u}$\;
			
			$bSet \gets \{x \mid \desc{x} \cap \desc{a} = \emptyset;\ \gamma(\desc{x},\desc{a}) > 0;\ \level{a} \geq \level{x};\curly{(\parent{a},a),(\parent{x},x)}$
			\quad \quad $\text{ is a induced cut in }G_v\}$ \;
			$b \gets\underset{x \in bSet}{argmin}\ \level{x}$\;
			$z \gets \LCA\paren{a,b}$ \;
			$\mathcal D^2(a) \triangleq \angularbrac{node1:a,parent1:\parent{a},eta1:\eta(a),node1PivotOutEdges: H^{v}_{\desc{a}}, node2:b,parent2:\parent{b},eta2:\eta(b),node2PivotOutEdges: H^{v}_{\desc{b}},betweenedges:\gamma(\desc{a},\desc{b}),lca:z,lcalevel:\level{z},pivotnode:v,pivotenodelevel:\level{v}}$
		}

		%\If{no $b$ exists such that $\curly{(\parent{a},a),(\parent{b},b)}$ is a min-cut in $G_v$ for some $v \in \ancestor{a}$}{
		%    $\mathcal D^2(a) = \phi$
		%}\Else{
		%    $v \gets$ closest node to $r$ such that $\curly{(\parent{a},a),(\parent{x},x)}$ is a induced cut in $G_v$ for some node $x$ such that $\desc{x} \cap \desc{a} = \emptyset$ \;
		%    $b \gets$ closest node to $r$ such that $\curly{(\parent{a},a),(\parent{b},b)}$ is a induced cut in $G_v$ and $\desc{a} \cap \desc{b} = \emptyset$ \;
		%    $\mathcal D^2(a) \triangleq \angularbrac{a,\parent{a},\eta(a),b,\parent{b},\eta(b),\gamma(\desc{a},\desc{b}),x,\level{x},v,H_{\desc{a}}^v,\level{v}}$
		%    }
	}
	\caption{1-cut details and 2-cut details (to be run after layered algorithm for min-cut on all nodes $a$)}
	\label{algo:1-2-cut-detail}
\end{algorithm}

\begin{obs}
	For each node $x$, $\mathcal D^1(x)$ (1-cut detail) and $\mathcal D^2(x)$ (2-cut detail) can be computed in $O(D^2)$ time as given in Algorithm \ref{algo:1-2-cut-detail}
\end{obs}
\begin{proof}
	In Observation \ref{obs:layered_algorithm_run_time}, we saw that the layered algorithm for min-cut runs in $O(D^2)$ time. We claim that after  the layered algorithm for min-cut is executed we have all the information to compute both 1-cut detail and 2-cut detail. For 1-cut detail of a node $a$ we require $\eta(a)$ which node $a$ knows from the first section. Further, we require node $v$ closest to the root such that the edge $(\parent{a},a)$ persists as a min-cut in the sub-graph $G_v$. This node can be found after layered algorithm for min-cut is run. Also $H_{\desc{a}}^{v}$ is available with node $a$ from Section \ref{subsection:size-of-tree-cuts}.
	
	Similarly for 2-cut detail. Here the layered algorithm uses Lemma \ref{lemma:2-cut-2-respect-disjoint}, which takes $O(D)$ rounds. At the end of the layered algorithm for the min-cut every node $a$ knows all the required induced-cut of size $2$ that it participates in subgraph $G_v$ for all $v \in \ancestor{a}$. Based on that, it can calculate 2-cut detail.
	%\ToDoY{mohit}{in the argument write down why it can find the lca $z$ using Observation also $H_{\desc{b}}^{v}$ \ref{obs:layered_algor_graph_sequence_min_cut}}
\end{proof}
Now, we will give characterization of $\CASE{5}$. This characterization will help us to desginate node $v_2,v_3$ as special nodes in $\desc{x}$.

\begin{lemma}
	\label{lemma:case_5_layered_algo_characterization_A}
	Let the graph $G$ be 3-connected (there are no min-cuts of size $1$ and $2$ in $G$). Let $v_1,v_2,v_3 \in V\setminus \curly{r}$, $\desc{v_2} \subset \desc{v_1}$ and $\desc{v_3} \subset \desc{v_1}$ and $\desc{v_2} \cap \desc{v_3} = \emptyset$, let $\curly{(\parent{v_1},v_1),(\parent{v_2},v_2),(\parent{v_3},v_3)}$ be a min-cut (as in \CASE{5}). Then the following statements are true
	\begin{enumerate}[i)]
		\item Let $x \in \desc{v_1} \setminus \paren{\desc{v_2} \cup \desc{v_3}}$ and  $x$ is not on either of the path $\rho(v_2),\rho(v_3)$  then the edge $(\parent{x},x)$ is not a min-cut in $G_{v_1}$ \label{lemma:case_5_layered_algo_characterization_A_subpt_1}
		\item when $\gamma(\desc{v_2},\desc{v_3}) = 0$, $(\parent{v_2},v_2)$ and $(\parent{v_3},v_3)$ are 1-cuts for sub-graph in $\subgraphsequence{v_2}{v_1}$ and $\subgraphsequence{v_3}{v_1}$ respectively  \label{lemma:case_5_layered_algo_characterization_A_subpt_3}
		\item Let $x \in \desc{v_1} \setminus \paren{\desc{v_2} \cup \desc{v_3}}$ be a node and let $x$ not be on either of the two paths $\rho(v_2),\rho(v_3)$ and also let $\level{x} \in \curly{\level{v_2},\level{v_3}}$. When $\gamma(\desc{v_2},\desc{v_3}) > 0$, then there does not exists a node $y$ such that $\desc{x} \cap \desc{y} = \emptyset$, $\gamma(\desc{x},\desc{y}) > 0$ and the edge set $\curly{(\parent{x},x),(\parent{y},y)}$ is a induced cut of size $2$ in $G_{v_1}$ \label{lemma:case_5_layered_algo_characterization_A_subpt_2}
		\item Let node $z$ be the LCA of $v_2,v_3$. When $\gamma(\desc{v_2},\desc{v_3}) > 0$ then $\curly{(\parent{v_2},v_2),(\parent{v_3},v_3)}$ is a 2-cut for sub-graph sequence $\subgraphsequence{z}{v_1}$  \label{lemma:case_5_layered_algo_characterization_A_subpt_4}
		\item  When $\gamma(\desc{v_2},\desc{v_3}) > 0$ then there does not exists a node $x$ on the path $\rho(v_2)$ such that $\level{x} < \level{v_2}$, $\desc{x} \cap \desc{v_3} = \emptyset$ and $\curly{(\parent{x},x),(\parent{v_3},v_3)}$ is an induced cut of $G_{v_1}$ \label{lemma:case_5_layered_algo_characterization_A_subpt_5}
	\end{enumerate}
\end{lemma}
\begin{obs}
	\label{obs:case_5_simple_observation}
	If there is an induced cut as given in Lemma \ref{lemma:case_5_layered_algo_characterization_A}, then no node in the set $\desc{v_1} \setminus \paren{\desc{v_2} \cup \desc{v_3}}$ is connected to any node in $V \setminus \desc{v_1}$. Also apart from the edges $(\parent{v_2},v_2)$ and $(\parent{v_3},v_3)$ node in $\paren{\desc{v_2} \cup \desc{v_3}}$ do not have any edge with nodes in $\desc{v_1} \setminus \paren{\desc{v_2} \cup \desc{v_3}}$.
\end{obs}
\begin{proof}
	If such was the case then $\curly{(\parent{v_1},v_1),(\parent{v_2},v_2),(\parent{v_3},v_3)}$ would not have been a min-cut. We will prove the above claims using this simple observation.
\end{proof}
\vspace{0.5 cm}
\begin{proof}[Proof of Lemma \ref{lemma:case_5_layered_algo_characterization_A}]
	\begin{enumerate}[i)]
		\item As per given condition $v_2 \notin \desc{x}$ and $v_3 \notin \desc{x}$. By Observation \ref{obs:case_5_simple_observation}, the only nodes in $\desc{v_1}$ which have an adjacent edge to the nodes in $V\setminus \desc{v_1}$ are in the vertex sets $\desc{v_2}$ or $\desc{v_3}$. Now since $(\parent{x},x)$ is a min-cut in $G_{v_1}$ thus no node in $\desc{x}$ has an adjacent edge with a node in $V\setminus \desc{v_1}$ . Thus if $(\parent{x},x)$ is a min-cut for the sub-graph $G_{v_1}$ then it is also a min-cut edge in the whole graph $G$ which is a contradiction because the graph is 3-connected.
		\item Since $\gamma(\desc{v_2},\desc{v_3}) = 0$ thus there are no edges between vertices in the sets $\desc{v_2}$ and $\desc{v_3}$. Also there are no edges between vertices in the vertex set $\desc{v_1}\setminus \desc{v_2}$ and $\desc{v_2}$. Thus $(\parent{v_2},v_2)$ is a cut-edge in subgraph $G_{v_1}$ and in all the subgraphs in $\subgraphsequence{v_2}{v_1}$ by Observation \ref{obs:layered_algor_graph_sequence_min_cut}. Similarly for the edge $(\parent{v_3},v_3)$.
		
		\item
		Lets consider the simple case when $y$ is in either $\desc{v_2}$ or $\desc{v_3}$. Here $\gamma\paren{\desc{x},\desc{y}} = 0$ by Observation \ref{obs:case_5_simple_observation}

		The other case when $(\parent{y},y)$ is on one of the paths out of $\rho(v_2)$ or $\rho(v_3)$. WLOG let $(\parent{y},y)$ be on the path $\rho(v_2)$. Thus $\desc{v_2} \subseteq \desc{y}$. Also given that $\gamma(\desc{v_2},\desc{v_3}) > 0$ thus $\curly{(\parent{x},x),(\parent{y},y)}$ cannot be a cut induced by $\delta(\desc{x} \oplus \desc{y})$ in $G_{v_1}$ since $\desc{v_2} \subseteq \desc{y}$ and $\desc{v_3} \subseteq \desc{v_1} \setminus (\desc{x} \cup \desc{y})$ thus the vertex set $\desc{x} \cup \desc{y}$ and $\desc{v_1} \setminus (\desc{x} \cup \desc{y})$ have at least one edge between them

		Now let $y$ be not on either of the path $\rho(v_2)$ or $ \rho(v_3)$. Also $y$ is not in either of the vertex set $\desc{v_2}$ and $\desc{v_3}$. Now, as per given condition, $(\parent{x},x)$ is not on either of the path $\rho(v_2)$ or $ \rho(v_3)$.     Since both the edges are not on the paths $\rho(v_2)$ and $\rho(v_3)$ thus neither of $v_2$ or $v_3$ are in the vertex set $\desc{x}$ or $\desc{y}$. If here $\curly{(\parent{x},x),(\parent{y},y)}$ is a min-cut then there are edges between vertices of the two sets $\desc{x}$ and $\desc{y}$ and they do not have any other edge to the vertices in $\desc{v_1}\setminus (\desc{x} \cup \desc{y})$. Also here by Observation \ref{obs:case_5_simple_observation} nodes in $\desc{x} \cup \desc{y}$  do not have any edge in $V\setminus \desc{v_1}$. Thus by the above two statements nodes in $\desc{x} \cup \desc{y}$ do not have any adjacent edges to nodes in $V\setminus (\desc{x} \cup \desc{y})$ hence, the edge set $\curly{(\parent{x},x),(\parent{y},y)}$ persists as a min-cut in the whole graph which is a contradiction because the graph is 3-connected.

		\item Since $\gamma(\desc{v_2},\desc{v_3}) > 0$ thus there are some edges between vertices of the sets $\desc{v_2}$ and $\desc{v_3}$. Also by Observation \ref{obs:case_5_simple_observation} there are no edges between vertices in the vertex set $\desc{v_1}\setminus (\desc{v_2} \cup \desc{v_3})$ and $\desc{v_2}$. Similarly there are no edges between vertices in the vertex set $\desc{v_1}\setminus (\desc{v_2} \cup \desc{v_3})$ and $\desc{v_3}$. Here $z$ is the LCA of $v_2,v_3$ and there are edges which go between $\desc{v_2}$ and $\desc{v_3}$.  Also, by Observation \ref{obs:case_5_simple_observation} no edges other than $\curly{(\parent{v_2},v_2),(\parent{v_3},v_3)}$ connect nodes in $\desc{z}/ (\desc{v_2} \cup \desc{v_3})$ to nodes in $\desc{v_2} \cup \desc{v_3}$. Thus $\curly{(\parent{v_2},v_2),(\parent{v_3},v_3)}$ is a 2-cut in $G_z$. Similar argument can be given for all sub-graphs in the graph sequence $\subgraphsequence{z}{v_1}$
		%\ToDoY{Mohit}{Other way here is to prove for $G_{v_1}$ and then use Observation \ref{obs:layered_algor_graph_sequence_min_cut}}

		\item As per the given condition $x \in \desc{v_1} \setminus (\desc{v_2} \cup \desc{v_3})$ and $v_2 \in \desc{x}$. We will focus on the vertex set $\desc{x} \setminus \desc{v_2}$ which is not empty because $\level{x} < \level{v_2}$. If $\curly{(\parent{x},x),(\parent{v_3},v_3)}$ is an induced cut of $G_{v_1}$ then vertices in the set $\desc{x} \cup \desc{v_3}$ do not have an adjacent edge to vertices in $\desc{v_1} \setminus (\desc{x} \cup \desc{v_3})$. Thus vertices in $\desc{x} \setminus \desc{v_2}$ does not have an adjacent edge to vertices in $\desc{v_1} \setminus (\desc{x} \cup \desc{v_3})$. Further  $\desc{x} \setminus \desc{v_2} \subset  \desc{v_1} \setminus (\desc{v_2} \cup \desc{v_3})$ thus by Observation \ref{obs:case_5_simple_observation} vertices in $\desc{x} \setminus \desc{v_2}$ do not have an adjacent edge to vertices in $V \setminus \desc{v_1}$. Hence $\delta(\desc{x} \setminus \desc{v_2}) = \curly{(\parent{x},{x}),(\parent{v_2},{v_2})}$ is an induced cut in the whole graph which is a contradiction.
	\end{enumerate}
\end{proof}

Now we will give a variant of a convergecast algorithm for communicating  1-cut details and 2-cut detail of a node to the ancestors in Algorithm \ref{algo:converge-cast-algorithm}. Similarly for 2-cut detail. Note that for any node $v$ it will not receive all $\mathcal D^1(x)$ and $\mathcal D^2(x)$ for all $x \in \desc{v}$ but the ones which are relevant among others. We will characterize this in Observation \ref{obs:layered_algo_convergecasting_charecterize}.

\begin{algorithm}[h]
	\DontPrintSemicolon
	\SetKwFor{Forp}{for}{parallely}{endfor}
	\SetKwFor{Forw}{for}{wait}{endfor}
	\SetKwProg{phase}{}{}{}
	\phase{convergecasting 1-cut detail}{
		\If{node $a$ is a leaf node}{
			send $<\mathcal{D}^1(a),\level{a}>$ to $\parent{a}$
		}\ElseIf{node $a$ is a internal node}{
			\lFor{round $t = 0$}{send $<\mathcal{D}^1(a),\level{a}>$ to $\parent{a}$}
			\For{each subsequet round $t = 1$ to $t = D - \level{a}$}{
				collect message $\msg_c$ from all node $c \in \child(a)$\;
				\lIf{$\msg_c = \phi\ \forall c \in \child(a)$}{
					send $<\phi,\level{a} + t>$ to $\parent{a}$
				}
				\Else{
					$\mathcal D\gets$ among $\msg_c \forall c \in \child(a)$ choose the 1-cut detail from which 1-cut persists for the sub-graph closest to the root that is the lowest $6^{th}$ element (last) of the 1-cut detail tuple \label{line:converge_cast_contention_for_children}\;
					send $<\mathcal D,\level{a} + t>$ to $\parent{a}$
				}
				
			}
	}}    
	\setcounter{AlgoLine}{0}
	
	\phase{convergecasting 2-cut detail}{
		\If{node $a$ is a leaf node}{
			send $<\mathcal{D}^2(a),\level{a}>$ to $\parent{a}$
		}\ElseIf{node $a$ is a internal node}{
			\lFor{round $t = 0$}{send $<\mathcal{D}^2(a),\level{a}>$ to $\parent{a}$}
			\For{each subsequet round $t = 1$ to $t = D - \level{a}$}{
				collect message $\msg_c$ from all node $c \in \child(a)$\;
				\lIf{$\msg_c = \phi\ \forall c \in \child(a)$}{
					send $<\phi,\level{a} + t>$ to $\parent{a}$
				}
				\Else{
					$\mathcal D \gets$ among $\msg_c \forall c \in \child(a)$ choose the 2-cut detail for which 2-cut persists for the sub-graph closest to the root that is the lowest $11^{th}$ element (last) of the 2-cut detail tuple \label{line:converge_cast_contention_for_children_detail_2}\;
					send $<\mathcal D,\level{a} + t>$ to $\parent{a}$
				}
				
			}
	}}    
	\caption{Converge-cast algorithm to be run at every node $a$}
	\label{algo:converge-cast-algorithm}
\end{algorithm}

\begin{obs}
	\label{obs:layered_algo_convergecasting_charecterize}
	For some nodes $v_1,v_2,v_3 \in V\setminus \curly{r}$, let $\desc{v_2} \subset \desc{v_1}$ and $\desc{v_3} \subset \desc{v_1}$ and $\desc{v_2} \cap \desc{v_3} = \emptyset$. Let $\curly{(\parent{v_1},v_1),(\parent{v_2},v_2),(\parent{v_3},v_3)}$ be a min-cut (as in \CASE{5}). If $\gamma(\desc{v_2},\desc{v_3}) = 0$ then at the end of the Algorithm \ref{algo:converge-cast-algorithm}, $\LCA(v_2,v_3)$ has both ${\mathcal D^{1}(v_2)},\mathcal D^{1}(v_3)$. And when $\gamma(\desc{v_2},\desc{v_3}) >0$ then it has either of $\mathcal D^{2}(v_2)$, $\curly{\mathcal D^{2}(v_3)}$. Also Algorithm \ref{algo:converge-cast-algorithm} takes $O(D)$ rounds.
\end{obs}
\begin{proof}
	Let $z$ be the LCA of node $v_2,v_3$. Lets first consider $\gamma(\desc{v_2},\desc{v_3}) = 0$. Here based on Lemma \ref{lemma:case_5_layered_algo_characterization_A} (\ref{lemma:case_5_layered_algo_characterization_A_subpt_1} and \ref{lemma:case_5_layered_algo_characterization_A_subpt_3}) we can say that there does not exists a node  $x \in \desc{v_1} \setminus \paren{\desc{v_2} \cup \desc{v_3}}$ such that $(\parent{x},x)$ persists as a min-cut in $G_{v_1}$. Thus at line \ref{line:converge_cast_contention_for_children} of Algorithm \ref{algo:converge-cast-algorithm}, $\mathcal{D}^1(v_2)$ will not have a contention at any node $u \in \ancestor{v_2}$ such that $\level{u} > \level{z}$. Similarly for $v_3$. The first contention happen at node $z$ (LCA of $v_2$ and $v_3$). Let $z$ be the LCA of node $v_2,v_3$. Here $z$ will have both $\mathcal D^{1}(v_2)$ and $\mathcal D^{1}(v_3)$ received from two different children $c_1,c_2 \in \child(z)$.
	
	Now when $\gamma(\desc{v_2},\desc{v_3}) > 0$. Here by Lemma \ref{lemma:case_5_layered_algo_characterization_A} (pt. \ref{lemma:case_5_layered_algo_characterization_A_subpt_4}) we know that $\curly{(\parent{v_2},{v_2}),(\parent{v_3},{v_3})}$ is an induced 2-cut. WLOG Lets assume $\level{v_2} \leq \level{v_3}$.  Thus using Algorithm \ref{algo:comm_2} node $v_2$ can make a decision about the said induced cut of size 2. Also $\mathcal D^{2}(v_2)$ (2-cut detail of node $v_2$) contains information about the induced cut $\curly{(\parent{v_2},{v_2}),(\parent{v_3},{v_3})}$ because there does not exists any other node $x$ as per Lemma \ref{lemma:case_5_layered_algo_characterization_A} (pt. \ref{lemma:case_5_layered_algo_characterization_A_subpt_5}) such that $\level{x} < \level{v_3}$ and $\curly{(\parent{x},{x}),(\parent{v_2},{v_2})}$ is an induced cut of size 2 in $G_{v_1}$. Further using the Lemma \ref{lemma:case_5_layered_algo_characterization_A} (pt. \ref{lemma:case_5_layered_algo_characterization_A_subpt_3}) there will not be any contention for $\mathcal D^2(v_2)$ to reach to $\LCA(v_2,v_3)$.
	
	Also both the converge-casting modules in Algorithm \ref{algo:converge-cast-algorithm} run for $O(D)$ time because both 1-cut detail and 2-cut detail are of the size $O(\log n)$ bits and for each node the modules run for atmost $O(D)$ rounds.
\end{proof}
\begin{lemma}
	For some nodes $v_1,v_2,v_3 \in V\setminus \curly{r}$, let $\desc{v_2} \subset \desc{v_1}$ and $\desc{v_3} \subset \desc{v_1}$ and $\desc{v_2} \cap \desc{v_3} = \emptyset$ . Let $\curly{(\parent{v_1},v_1),(\parent{v_2},v_2),(\parent{v_3},v_3)}$ be a min-cut (as in \CASE{5}). Then we can find this min-cut in $O(D^2)$ rounds.
\end{lemma}
\begin{proof}
	Here we run Algorithm \ref{algo:case_5} on each node $x$. This algorithm is run after the convergecast of 1-cut details and 2-cut details. This algorithm has two parts: finding min-cut from 1-cut detail and from 2-cut detail as given in Algorithm \ref{algo:converge-cast-algorithm}. We will argue that all the min-cuts of the form as given by \CASE 5 are correctly found by this algorithm. And also prove that this process takes $O(D^2)$ rounds.
	
	We claim that when Algorithm \ref{algo:case_5} is run on LCA of $v_2,v_3$ then it can find the required min-cut. Firstly, let see the case where $\gamma(\desc{v_2},\desc{v_3}) = 0$ here as per Observation \ref{obs:layered_algo_convergecasting_charecterize} the LCA of $v_2,v_3$ has both $\mathcal D^1(v_2)$ and $\mathcal D^1(v_3)$. Let $v$ be the pivotnode from $\mathcal D^1(v_2)$ and  $u$ be the pivotnode from $\mathcal D^1(v_3)$. The information regarding $H_{\desc{v_2}}^{v}$ and $H_{\desc{v_3}}^{u}$ is available from the respective one cut detail. WLOG let $\level{v} \geq \level{u}$. Thus here $H_{\desc{v_3}}^{u} = H_{\desc{v_3}}^{v}$ because the number of edges going out of the vertex set $\desc{v_3}$ which also goes out of the vertex $\desc{u}$ is same as the number of edges going out of the vertex set $\desc{v_3}$ which also goes out of the vertex $\desc{v}$. Now if the condition in Line \ref{line:algo-case5-if-condition-1} is satisfied then the min-cut is correctly found as per Lemma \ref{lemma:gen_3_cut_3_respect}.
	
	Further, moving to case when $\gamma(\desc{v_2},\desc{v_3}) > 0$. Here as per Observation \ref{obs:layered_algo_convergecasting_charecterize} one of $\mathcal D^2(v_2)$ or $\mathcal D^2(v_3)$ is present at node $v_3$. And when condition at Line \ref{line:algo-case5-if-condition-2} then the min-cut is correctly found as per Lemma \ref{lemma:gen_3_cut_3_respect}.
	\begin{algorithm}[h]
		\DontPrintSemicolon
		\SetKwFor{Forp}{for}{parallely}{endfor}
		\SetKwFor{Forw}{for}{wait}{endfor}
		\SetKwProg{phase}{}{}{}
		\phase{from 1-cut detail}{
			$nodePair \gets \curly{(a,b) \mid \mathcal D^1(a),\mathcal D^1(b)\text{ where received from two nodes } c_1,c_2 \in \child(v)}$\;
			\For{each $(a,b) \in nodePair$}{
				$v \gets pivotnode$ from $\mathcal D^1(a), u \gets pivotnode$ from $\mathcal D^1(b)$\;
				$z \gets \underset{z' \in \curly{v,u}}{argmax\ } \level{z'}$\;
				\If{$\eta(z) - 1 = H_{\desc{a}}^{z} + H_{\desc{b}}^{z}$ \& $\eta(a) - 1 = H_{\desc{a}}^{z}$ \& $\eta(b) - 1 = H_{\desc{b}}^{z}$}{ \label{line:algo-case5-if-condition-1}
					$\curly{(\parent{z},{z}),(\parent{a},{a}),(\parent{b},{b})}$ is a min-cut
				}
			}
		}    
		\setcounter{AlgoLine}{0}
		
		\phase{from 2-cut detail}{
			\For{ all node $a$ such that $\mathcal D^2(a)$ was received}{
				$a,b$ be $node1$ and $node2$ in $\mathcal D^2(a)$; $v$ be the pivotNode from $\mathcal D^2(a)$\;
				\If{$\eta(v) - 1 = H_{\desc{a}}^{v} + H_{\desc{b}}^{v}$ \& $\eta(a) - 1 = H_{\desc{a}}^{v} + \gamma(\desc{a},\desc{b}) $\& $\eta(b) - 1 = H_{\desc{b}}^{v}+\gamma(\desc{a},\desc{b}) $}{
					\label{line:algo-case5-if-condition-2}
					$\curly{(\parent{v},{v}),(\parent{a},{a}),(\parent{b},{b})}$ is a min-cut
				}
			}
		}    
		\caption{Algorithm to find min-cut as given by \CASE 5 (to be run at each internal node $x$)}
		\label{algo:case_5}
	\end{algorithm}
	
	Also, this whole process just takes $O(D^2)$ rounds because the layered algorithm for min-cut takes $O(D^2)$ rounds and the covergecast algorithm given in Algorithm \ref{algo:converge-cast-algorithm} takes just $O(D)$ rounds.
\end{proof}

%% file: thesis_subparts/chapter5/summary.tex
%TEX root = ../../thesis.tex
\section{Summary}
In this chapter, we gave noval deterministic algorithm to find min-cuts of size $3$.  The backbone of our algorithm is the characterization presented in Lemma \ref{lemma:gen_3_cut_2_respect} and Lemma \ref{lemma:gen_3_cut_3_respect}. In this chapter, we showed that whenever there exists a min-cut of size $3$ at least one node in the whole network will receive the required quantities as per the characterization  in $O(D^2)$ rounds and thus the min-cut will be found. 

The communication strategies introduced in this chapter are of two flavours: \emph{Sketching} technique where we collected small but relevant information and \emph{Layered Algorithm} which is a convergecast technique based on a contention resolution mechanism.

%% file: thesis_subparts/futureWork.tex
%TEX root = ../thesis.tex
We have discussed techniques for finding min-cut. Before our result, a decade old work existed to find small min-cuts by \cite{pritchard2008fast}. They gave an algebraic technique of random circulation for finding min-cut of size $1$ and $2$. On the contrary, we give deterministic and purely combinatorial techniques for finding min-cuts of size $1,2$ and $3$. Our algorithm takes $O(D)$ rounds for finding min-cut of size $1,2$ and $O(D^2)$ rounds for finding min-cut of size $3$. An immediate future goal could be to prove that the lower bound $O(D^2)$ for a min-cut of size $3$ which is not trivial.

Currently, there are some fundamental issues for finding min-cuts of size $4$ and above using our techniques. It would be interesting to extend our techniques for min-cut of size $4$. We are hopeful about it and give the following conjecture

\begin{conj}
There exists a distributed algorithm to find if a min-cut of size $k = o(\log n)$ exists in $\tilde{O}(D^k)$ rounds.
\end{conj}

%% file: thesis_subparts/cv.tex
\begin{tabular}{clll}
	1.&\textbf{NAME} & : & Mohit Daga\\
	2.&\textbf{DATE OF BIRTH} & : & July 5, 1993\\
	3.&\textbf{EDUCATION QUALIFICATIONS}     &  & \\
\end{tabular}
\begin{enumerate}[(a)]
	\item \textbf{Bachelor of Technology}
	
	\begin{tabular}{lll}
		{Institute} & :&The LNM Institute of Information Technology, Jaipur\\
		{Specialization} & : & Communication and Computer Engineering\\
	\end{tabular}
	\item \textbf{M.S. By Research}
	
	\begin{tabular}{lll}
	{Institute} & :&Indian Institute of Technology - Madras\\
	{Registration Date} & : & January 6, 2016\\
	\end{tabular}
\end{enumerate}

%% file: thesis_subparts/gtc.tex
\begin{tabular}{lll}
	\textbf{CHAIRPERSON} & : & Prof. N.S. Narayanaswamy                       \\
	&   & Professor                                      \\
	&   & Department of Computer Science and Engineering \\
	&   &                                                \\
	\textbf{GUIDE}       & : & Dr. John Augustine                             \\
	&   & Associate Professor                            \\
	&   & Department of Computer Science and Engineering \\
	&   &                                                \\
	\textbf{MEMBERS}     & : & Dr. Shweta Agrawal                             \\
	&   & Assistant Professor                            \\
	&   & Department of Computer Science and Engineering \\
	&   &                                                \\
	&   & Dr. Krishna Jagannathan                        \\
	&   & Assistant Professor                            \\
	&   & Department of Electrical Engineering          
\end{tabular}

%% file: thesis.bbl
\begin{thebibliography}{30}
\expandafter\ifx\csname natexlab\endcsname\relax\def\natexlab#1{#1}\fi
\expandafter\ifx\csname url\endcsname\relax
  \def\url#1{{\tt #1}}\fi
\expandafter\ifx\csname urlprefix\endcsname\relax\def\urlprefix{URL }\fi

\bibitem[{Ahuja and Zhu(1989)}]{ahuja1989efficient}
{\bf Ahuja, M.} and {\bf Y.~Zhu}, An efficient distributed algorithm for
  finding articulation points, bridges, and biconnected components in
  asynchronous networks.
\newblock {\em In\/} {\em Foundations of Software Technology and Theoretical
  Computer Science\/}, FSTTCS. 1989.

\bibitem[{Bondy and Murty()}]{bondy1976graph}
{\bf Bondy, J.~A.} and {\bf U.~S.~R. Murty}, {\em Graph theory with
  applications\/}, volume 290.
\newblock Citeseer, .

\bibitem[{Das~Sarma {\em et~al.\/}(2011)Das~Sarma, Holzer, Kor, Korman,
  Nanongkai, Pandurangan, Peleg, and
  Wattenhofer}]{DasSarma:2011:DVH:1993636.1993686}
{\bf Das~Sarma, A.}, {\bf S.~Holzer}, {\bf L.~Kor}, {\bf A.~Korman}, {\bf
  D.~Nanongkai}, {\bf G.~Pandurangan}, {\bf D.~Peleg}, and {\bf
  R.~Wattenhofer}, Distributed verification and hardness of distributed
  approximation.
\newblock {\em In\/} {\em Proceedings of the 43rd Annual ACM Symposium on
  Theory of Computing\/}, STOC. 2011.

\bibitem[{Elkin(2017)}]{Elkin:2017:SDD:3087801.3087823}
{\bf Elkin, M.}, A simple deterministic distributed mst algorithm, with
  near-optimal time and message complexities.
\newblock {\em In\/} {\em Proceedings of the ACM Symposium on Principles of
  Distributed Computing\/}, PODC. 2017.

\bibitem[{Elkin and Neiman(2017)}]{DBLP:journals/corr/ElkinN17}
{\bf Elkin, M.} and {\bf O.~Neiman} (2017).
\newblock Linear-size hopsets with small hopbound, and distributed routing with
  low memory.
\newblock \urlprefix\url{http://arxiv.org/abs/1704.08468}.

\bibitem[{Gabow(1991)}]{Gabow:1991:MAF:103418.103436}
{\bf Gabow, H.~N.}, A matroid approach to finding edge connectivity and packing
  arborescences.
\newblock {\em In\/} {\em Proceedings of the Twenty-third Annual ACM Symposium
  on Theory of Computing\/}, STOC '91. ACM, New York, NY, USA, 1991.
\newblock ISBN 0-89791-397-3.
\newblock \urlprefix\url{http://doi.acm.org/10.1145/103418.103436}.

\bibitem[{Ghaffari and Haeupler(2016)}]{ghaffari2016distributed}
{\bf Ghaffari, M.} and {\bf B.~Haeupler}, Distributed algorithms for planar
  networks - ii: Low-congestion shortcuts, mst, and min-cut.
\newblock {\em In\/} {\em Proceedings of the Twenty-Seventh Annual ACM-SIAM
  Symposium on Discrete Algorithms\/}, SODA. 2016.

\bibitem[{Ghaffari and Kuhn(2013)}]{ghaffari2013distributed}
{\bf Ghaffari, M.} and {\bf F.~Kuhn}, Distributed minimum cut approximation.
\newblock {\em In\/} {\em International Symposium on Distributed Computing\/}.
  Springer, 2013.

\bibitem[{Henzinger {\em et~al.\/}(2017)Henzinger, Rao, and
  Wang}]{Henzinger:2017:LFP:3039686.3039811}
{\bf Henzinger, M.}, {\bf S.~Rao}, and {\bf D.~Wang}, Local flow partitioning
  for faster edge connectivity.
\newblock {\em In\/} {\em Proceedings of the Twenty-Eighth Annual ACM-SIAM
  Symposium on Discrete Algorithms\/}, SODA. 2017.

\bibitem[{Hohberg(1990)}]{Hohberg:1990:FBC:89472.89477}
{\bf Hohberg, W.} (1990).
\newblock How to find biconnected components in distributed networks.
\newblock {\em J. Parallel Distrib. Comput.\/}, {\bf 9}(4), 374--386.
\newblock ISSN 0743-7315.
\newblock \urlprefix\url{http://dx.doi.org/10.1016/0743-7315(90)90122-6}.

\bibitem[{Karger(1993)}]{Karger:1993:GMR:313559.313605}
{\bf Karger, D.~R.}, Global min-cuts in rnc, and other ramifications of a
  simple min-out algorithm.
\newblock {\em In\/} {\em Proceedings of the Fourth Annual ACM-SIAM Symposium
  on Discrete Algorithms\/}, SODA '93. Society for Industrial and Applied
  Mathematics, Philadelphia, PA, USA, 1993.
\newblock ISBN 0-89871-313-7.
\newblock \urlprefix\url{http://dl.acm.org/citation.cfm?id=313559.313605}.

\bibitem[{Karger(1994{\natexlab{{\em a\/}}})}]{Karger:1994:RSC:195058.195422}
{\bf Karger, D.~R.}, Random sampling in cut, flow, and network design problems.
\newblock {\em In\/} {\em Proceedings of the Twenty-sixth Annual ACM Symposium
  on Theory of Computing\/}, STOC '94. ACM, New York, NY, USA,
  1994{\natexlab{{\em a\/}}}.
\newblock ISBN 0-89791-663-8.
\newblock \urlprefix\url{http://doi.acm.org/10.1145/195058.195422}.

\bibitem[{Karger(1994{\natexlab{{\em b\/}}})}]{Karger:1994:URS:314464.314582}
{\bf Karger, D.~R.}, Using randomized sparsification to approximate minimum
  cuts.
\newblock {\em In\/} {\em Proceedings of the Fifth Annual ACM-SIAM Symposium on
  Discrete Algorithms\/}, SODA '94. Society for Industrial and Applied
  Mathematics, Philadelphia, PA, USA, 1994{\natexlab{{\em b\/}}}.
\newblock ISBN 0-89871-329-3.
\newblock \urlprefix\url{http://dl.acm.org/citation.cfm?id=314464.314582}.

\bibitem[{Karger(2000)}]{Karger:2000:MCN:331605.331608}
{\bf Karger, D.~R.} (2000).
\newblock Minimum cuts in near-linear time.
\newblock {\em J. ACM\/}, {\bf 47}(1), 46--76.
\newblock ISSN 0004-5411.
\newblock \urlprefix\url{http://doi.acm.org/10.1145/331605.331608}.

\bibitem[{Karger and Stein(1993)}]{karger1993o}
{\bf Karger, D.~R.} and {\bf C.~Stein}, An {\~o} (n 2) algorithm for minimum
  cuts.
\newblock {\em In\/} {\em Proceedings of the twenty-fifth annual ACM symposium
  on Theory of computing\/}. ACM, 1993.

\bibitem[{Kawarabayashi and Thorup(2015)}]{kawarabayashi2015deterministic}
{\bf Kawarabayashi, K.-i.} and {\bf M.~Thorup}, Deterministic global minimum
  cut of a simple graph in near-linear time.
\newblock {\em In\/} {\em Proceedings of the forty-seventh annual ACM symposium
  on Theory of computing\/}, STOC. 2015.

\bibitem[{Kutten and Peleg(1998)}]{Kutten:1998:FDC:296903.296906}
{\bf Kutten, S.} and {\bf D.~Peleg} (1998).
\newblock Fast distributed construction of smallk-dominating sets and
  applications.
\newblock {\em J. Algorithms\/}, {\bf 28}(1).

\bibitem[{Luby(1985)}]{Luby:1985:SPA:22145.22146}
{\bf Luby, M.}, A simple parallel algorithm for the maximal independent set
  problem.
\newblock {\em In\/} {\em Proceedings of the Seventeenth Annual ACM Symposium
  on Theory of Computing\/}, STOC. 1985.

\bibitem[{Matula(1993)}]{Matula:1993:LTP:313559.313872}
{\bf Matula, D.~W.}, A linear time 2 \&plus; \&\&egr;\&egr; approximation
  algorithm for edge connectivity.
\newblock {\em In\/} {\em Proceedings of the Fourth Annual ACM-SIAM Symposium
  on Discrete Algorithms\/}, SODA '93. Society for Industrial and Applied
  Mathematics, Philadelphia, PA, USA, 1993.
\newblock ISBN 0-89871-313-7.
\newblock \urlprefix\url{http://dl.acm.org/citation.cfm?id=313559.313872}.

\bibitem[{Nagamochi and Ibaraki(1992)}]{nagamochi1992computing}
{\bf Nagamochi, H.} and {\bf T.~Ibaraki} (1992).
\newblock Computing edge-connectivity in multigraphs and capacitated graphs.
\newblock {\em SIAM Journal on Discrete Mathematics\/}, {\bf 5}(1), 54--66.

\bibitem[{Nanongkai and Su(2014)}]{nanongkai2014almost}
{\bf Nanongkai, D.} and {\bf H.-H. Su}, Almost-tight distributed minimum cut
  algorithms.
\newblock {\em In\/} {\em 24th International Symposium on Distributed
  Computing\/}, DISC. 2014.

\bibitem[{Pandurangan {\em et~al.\/}(2017)Pandurangan, Robinson, and
  Scquizzato}]{Pandurangan:2017:TMD:3055399.3055449}
{\bf Pandurangan, G.}, {\bf P.~Robinson}, and {\bf M.~Scquizzato}, A time- and
  message-optimal distributed algorithm for minimum spanning trees.
\newblock {\em In\/} {\em Proceedings of the 49th Annual ACM SIGACT Symposium
  on Theory of Computing\/}, STOC. 2017.

\bibitem[{Peleg(2000)}]{peleg2000distributed}
{\bf Peleg, D.}, {\em Distributed computing: a locality-sensitive approach\/}.
\newblock SIAM, 2000.

\bibitem[{Pritchard(2006)}]{DBLP:journals/corr/abs-cs-0602013}
{\bf Pritchard, D.} (2006).
\newblock An optimal distributed edge-biconnectivity algorithm.
\newblock {\em Accepted as a Poster at 25th ACM symposium on Principles of
  distributed computing. PODC.\/}.
\newblock \urlprefix\url{http://arxiv.org/abs/cs/0602013}.

\bibitem[{Pritchard(2008)}]{pritchard2008fast}
{\bf Pritchard, D.}, Fast distributed computation of cuts via random
  circulations.
\newblock {\em In\/} {\em International Colloquium on Automata, Languages, and
  Programming\/}, ICALP. 2008.

\bibitem[{Pritchard and Thurimella(2011)}]{pritchard2011TALG}
{\bf Pritchard, D.} and {\bf R.~Thurimella} (2011).
\newblock Fast computation of small cuts via cycle space sampling.
\newblock {\em ACM Transactions on Algorithms (TALG)\/}, {\bf 7}(4), 46.

\bibitem[{Stoer and Wagner(1997)}]{Stoer:1997:SMA:263867.263872}
{\bf Stoer, M.} and {\bf F.~Wagner} (1997).
\newblock A simple min-cut algorithm.
\newblock {\em J. ACM\/}, {\bf 44}(4), 585--591.
\newblock ISSN 0004-5411.
\newblock \urlprefix\url{http://doi.acm.org/10.1145/263867.263872}.

\bibitem[{Thorup(2001)}]{thorup2001fully}
{\bf Thorup, M.}, Fully-dynamic min-cut.
\newblock {\em In\/} {\em Proceedings of the thirty-third annual ACM symposium
  on Theory of computing\/}. ACM, 2001.

\bibitem[{Thurimella(1995)}]{thurimella1995sub}
{\bf Thurimella, R.}, Sub-linear distributed algorithms for sparse certificates
  and biconnected components.
\newblock {\em In\/} {\em Proceedings of the fourteenth annual ACM symposium on
  Principles of distributed computing\/}, PODC. 1995.

\bibitem[{Tsin(2006)}]{tsin2006efficient}
{\bf Tsin, Y.~H.} (2006).
\newblock An efficient distributed algorithm for 3-edge-connectivity.
\newblock {\em International Journal of Foundations of Computer Science\/},
  {\bf 17}(03), 677--701.

\end{thebibliography}
